\documentclass{article}

\usepackage[english]{babel}
\usepackage{natbib}
\usepackage[letterpaper,top=2cm,bottom=2cm,left=3cm,right=3cm,marginparwidth=1.75cm]{geometry}
\usepackage{amsmath, amsthm, amssymb}
\usepackage{graphicx}
\usepackage[colorlinks=true, allcolors=blue]{hyperref}
\usepackage{subcaption}
\usepackage{mathtools}
\usepackage{wrapfig}
\newcommand{\blind}{1}
\usepackage{bibunits}

\defaultbibliographystyle{abbrvnat}
\usepackage[toc,page]{appendix}
\newcommand{\jasa}{0}
\usepackage{booktabs}
\usepackage{bbm}

\usepackage{algorithm}
\usepackage{algpseudocode}

\usepackage{multirow}

\DeclareMathOperator*{\argmax}{argmax}

\newtheorem{theorem}{Theorem}[section]
\newtheorem{proposition}[theorem]{Proposition}

\newtheorem{corollary}{Corollary}[section]

\newtheorem{definition}{Definition}
\newtheorem*{assumption*}{Assumption}
\newtheorem{assumption}{Assumption}

\newif\ifbiometrika
\biometrikafalse   

\providecommand{\keywords}[1]{\textbf{Key words:} #1}

\title{Laplace Approximations for Mixed-Effects and Gaussian Process Quantile Regression}
\author{
  Andrea Nava\footnotemark[1]  \footnotemark[2] \footnotemark[3]\\
  \and
  Fabio Sigrist\footnotemark[1] \footnotemark[2] 
}

\date{}

\begin{document}
\maketitle

\footnotetext[1]{Seminar for Statistics, ETH Zurich}
\footnotetext[2]{Lucerne University of Applied Sciences and Arts}
\footnotetext[3]{Corresponding author: navaan@ethz.ch}

\begin{abstract}
Laplace approximations are a standard tool for computationally efficient inference in latent Gaussian models, but they fail for quantile regression with the asymmetric Laplace likelihood because the observed Hessian vanishes almost everywhere. We show that this obstacle can be overcome without smoothing the likelihood: the relevant local curvature is given not by the observed Hessian, but by the Fisher information when the model is correctly specified and by the population curvature of the expected loss under misspecification. On this basis, we develop a Laplace approximation framework for quantile regression with mixed-effects and Gaussian process models. We propose practical curvature estimators, including the triangular kernel curvature (TKC) estimator, that yield approximations for posterior distributions and marginal likelihoods, and we establish their asymptotic validity. Empirically, the proposed methods are scalable and numerically stable, and for latent Gaussian models, they achieve accuracy comparable to or better than MCMC and variational competitors at substantially lower computational costs. More broadly, the framework clarifies how Laplace approximations can be justified for non-smooth generalized posteriors through local quadratic behavior of the expected loss.

\end{abstract}


\section{Introduction}

Quantile regression provides a flexible way to model aspects of a distribution beyond the mean. The goal of quantile regression is to estimate the conditional $\tau$-quantile of $Y|X=x$ defined as the value $Q_{\tau}(Y|X=x)$ such that
\begin{equation}
\mathbb{P}[Y \leq Q_{\tau}(Y|X=x)|X=x] = \tau,
\end{equation}
where $Y$ is a response variable and $X$ are input variables.
This is valuable in applications where conditional mean estimation is insufficient, for instance, when analyzing tail outcomes (see Figure~\ref{fig:mm_hsb} for an example), constructing prediction intervals with minimal distributional assumptions, or robustness to outliers is required. \citet{1aa6b708-cbcf-320f-b3f3-298154ee2aeb} introduced quantile regression via empirical risk minimization of the so-called pinball loss, an asymmetrically weighted version of the absolute error:
\begin{equation}
\rho_\tau(y ,\hat y) =
\begin{cases}
\tau (y - \hat y) & \text{if } y -\hat y \ge 0, \\
(\tau-1)(y-\hat y) & \text{if } y -\hat y < 0,
\end{cases}
\end{equation}
where $\hat y$ is the predicted $\tau$-quantile and $y$ the observed data. One can easily show that the $\tau$-quantile minimizes the expected pinball loss, and that this loss is a proper scoring rule in the sense of \citet{Schervish2018}. \citet{YU2001437} later introduced a Bayesian quantile regression model via the asymmetric Laplace likelihood defined as
\begin{equation}\label{al_likelihood}
p(y \mid \mu, \lambda) = \frac{\tau(1-\tau)}{\lambda} \exp \left(-\frac{1}{\lambda}\rho_\tau(y-\mu) \right),
\end{equation}
where $\mu\in\mathbb{R}$ and $\lambda>0$ are location and scale parameters, and the $\tau$-quantile of $Y$ is $\mu$. 
\begin{figure}[ht!]
    \centering
    \includegraphics[width=1\linewidth]{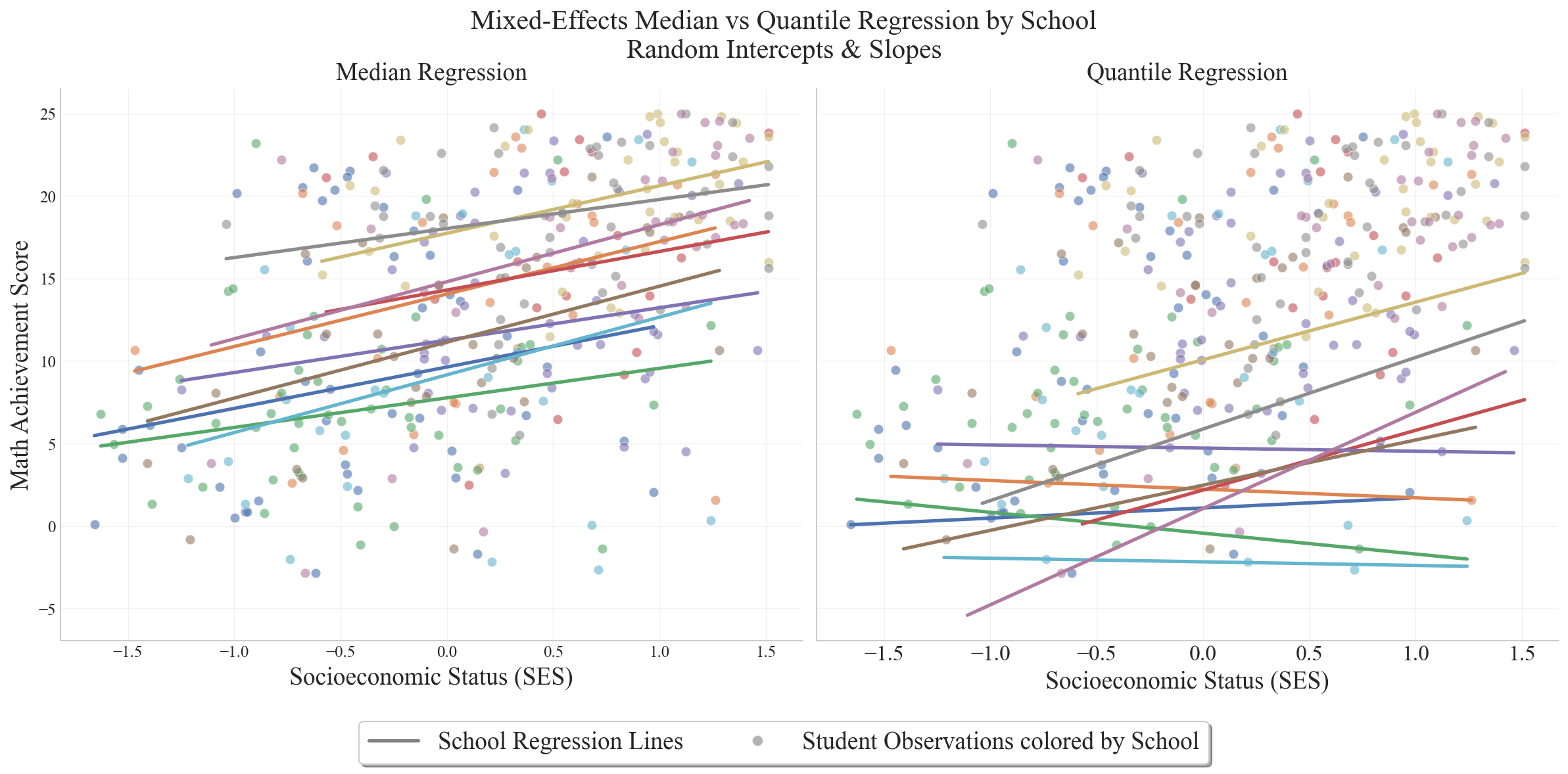}
    \caption{\textbf{A motivating example: High school math scores.} A school-level random intercept and random slope quantile regression model ($\tau=0.5$ and $\tau=0.05$) fitted with our proposed Laplace approximation for the Math Score of students from 160 different high-schools from the \textit{High School and Beyond (HSB) dataset}. The plots show the fitted function for 10 randomly selected schools. The left plot shows that the socioeconomic status has a positive association with the median scores. The right plot shows the fitted function for the 5\% quantile, which suggests that socioeconomic status is less strongly associated with math scores among lower-performing students.}
    \label{fig:mm_hsb}
\end{figure}

In this work, we introduce novel methods for quantile regression with latent Gaussian models including mixed-effects \citep{mcculloch2004generalized} and Gaussian process \citep{williams2006gaussian} models using the asymmetric Laplace likelihood. Laplace approximations are computationally very efficient \citep{nickisch2008approximations} and converge asymptotically to the correct quantity, depending on the asymptotic setting. A traditional Laplace approximation for Bayesian models \citep{tierney1986accurate} approximates the sum of the log-likelihood and log-prior by a second-order Taylor expansion centered at the posterior mode. This leads to an approximate multivariate normal posterior distribution and an approximate marginal likelihood. However, for the asymmetric Laplace likelihood, the observed Hessian of the log-likelihood is zero almost everywhere, and the Laplace approximation can thus not be applied. For this reason, widely used mixed-effects and latent Gaussian model libraries, such as \texttt{lme4}, \texttt{glmmTMB}, and \texttt{R-INLA}, which are based on a Laplace approximation, do not support quantile regression with the asymmetric Laplace likelihood. Resolving this limitation thus not only provides a theoretical basis for Laplace approximations in Bayesian quantile regression, but also leads to a scalable, open-source software implementation for practitioners. 

Le Cam's theory of local asymptotic normality (LAN) \citep{le2000asymptotics} shows that a local quadratic expansion of the log-likelihood does not require the likelihood function to be twice differentiable. Instead, differentiability in quadratic mean is sufficient. Under this weaker condition, the Fisher information, defined as the variance of the score function, characterizes the local curvature. Since the asymmetric Laplace likelihood satisfies this condition, this yields an asymptotically justified Laplace approximation, under the regularity conditions stated below, even though the observed Hessian is degenerate. Furthermore, in misspecified settings, we can rely on the theory of non-smooth M-estimators \citep{van2000asymptotic} to show that a meaningful curvature exists and equals the true data-generating density evaluated at the target quantile. We use these two curvature notions to construct Laplace approximations for latent Gaussian quantile regression and to derive practical plug-in estimators.  

This paper makes three main contributions. First, we construct Laplace approximations for latent Gaussian quantile regression by replacing the degenerate observed Hessian with curvature arising from the local quadratic behavior of the expected loss. Second, we prove asymptotic validity under both correct specification, where the relevant curvature is the Fisher information, and misspecification, where it is governed by the population curvature of the expected pinball loss. Third, we propose practical plug-in curvature estimators, including the triangular kernel curvature estimator, and show that the resulting methodology provides scalable and numerically stable inference for mixed-effects and Gaussian process quantile regression.

The proposed framework also suggests a possible route to Laplace approximations for non-smooth generalized Bayesian models beyond the asymmetric Laplace likelihood, in which curvature is understood through the local behavior of the expected loss rather than through the observed Hessian alone. Other examples where the observed Hessian is degenerate almost everywhere but the expected criterion can still exhibit local curvature include piecewise-linear losses such as the $\varepsilon$-insensitive loss in support vector regression and the interval score for interval regression.

\subsection{Related Work}
For linear mixed-effects models, \cite{geraci2014linear} introduced linear quantile mixed models (LQMMs), where the marginal likelihood is approximated using numerical quadrature. In Markov chain Monte Carlo (MCMC) methods, the asymmetric Laplace likelihood is often handled via its location-scale mixture representation, which facilitates the construction of a Gibbs sampler \citep{kozumi2011gibbs,yue2011bayesian}. Standard gradient-based MCMC methods, such as Hamiltonian Monte Carlo, may perform poorly since the gradients of the asymmetric Laplace log-likelihood are piecewise constant. For this reason, \cite{liu2025bayesian} recently introduced an MCMC approach based on a smoothed likelihood, obtained by convolving the asymmetric Laplace likelihood with density kernels. This approach resembles the work of \cite{he2023smoothed}, but in a sampling setting rather than an optimization one. 



Despite these developments, scalable and accurate quantile regression for latent Gaussian models remains challenging. \cite{geraci2017nonlinear} proposed a Laplace approximation for nonlinear quantile mixed-models, based on a smoothed version of the asymmetric Laplace likelihood. However, smoothing the likelihood introduces an additional bandwidth tuning parameter and, critically, it alters the exact connection between the asymmetric Laplace log-likelihood and the pinball loss that justifies Bayesian quantile regression as described below. \cite{ehm2016quantiles} characterize the pinball loss as the unique proper scoring rule (up to monotone transformation) for eliciting quantiles. This has profound implications for inference. \cite{10.1214/13-BA817} establish posterior consistency precisely because the log-likelihood of an asymmetric Laplace equals the pinball loss (up to constants). Remarkably, this property is special to the asymmetric Laplace likelihood. For instance, a Gaussian likelihood reparametrized to target quantiles, and any other smooth likelihood, do not enjoy the same consistency guarantees, because they lack the connection to the unique proper scoring rule for quantiles.

\cite{hartmann2019laplace} propose a Laplace approximation for the Student-t likelihood in which the Hessian is replaced by the expected Hessian to improve numerical conditioning as the Student-t likelihood is not log-concave. However, their approach differs from ours in several ways. First, we use the original definition of the Fisher information directly since the expected Hessian is not available for the asymmetric Laplace log-likelihood. Second, while their motivation is primarily to improve conditioning, our use of the Fisher information is essential to enable the Laplace approximation in a setting where the standard method breaks down entirely. Third, they do not address the issue that the log-likelihood is typically misspecified.

While Bayesian asymptotics under misspecification have been studied in the literature \citep{Miss-BvM, chernozhukov2003mcmc}, the connection between these theoretical results and the justification of Laplace approximations for non–twice-differentiable likelihoods under misspecification has not been made explicitly.

The remainder of this paper is organized as follows. Section~\ref{Laplace} introduces the  models we consider,  introduces our novel Laplace approximations, and derives theoretical results. Sections~\ref{Experiments} and~\ref{real_world_sec} present simulated and real-world data experiments. Section~\ref{Discussion} concludes with limitations and directions for future research.

\section{Bayesian Quantile Regression and Laplace Approximations}\label{Laplace}
We model the conditional $\tau$-quantiles $\boldsymbol{\mu}=(\mu_1, ..., \mu_n)^T \in \mathbb{R}^n$ of a response variable $\mathbf{y}=(y_1, ..., y_n)^T \in \mathbb{R}^n$ using a mixed-effects model:\footnote{From here on, we use lowercase symbols for both random variables and their realizations for notational simplicity.}
\begin{equation}\label{me_model_def}
\boldsymbol{\mu} = F(\mathbf{X}) + \mathbf{Zb}, \quad \mathbf{b} \sim \mathcal{N}(0, \boldsymbol{K_{\theta}}),
\end{equation}
where $F(\mathbf{X}) \in \mathbb{R}^n$ is a fixed effect, represented as the row-wise evaluation of a function $F(\cdot): \mathbb{R}^p \rightarrow \mathbb{R}$ on a fixed design matrix $\mathbf{X} \in \mathbb{R}^{n \times p}$. $F(\mathbf{X})$ can be modeled as a linear function, $F(\mathbf{X})=\mathbf{X\boldsymbol{\beta}}$, or a nonlinear function, such as boosted trees \citep{sigrist2022gaussian} and neural networks \citep{simchoni2023integrating}. We denote the parameters of this function $F(\cdot)$ by $\beta$, regardless of whether $F(\cdot)$ is linear or not. The vector $\mathbf{b} \in \mathbb{R}^m$ is a latent Gaussian random effect vector and can consist of grouped random effects \citep{mcculloch2004generalized} and/or a finite-dimensional version of a Gaussian process \citep{williams2006gaussian}. We assume that its covariance matrix $\boldsymbol{K_{\theta}}$ is governed by some hyperparameters $\boldsymbol{\theta} \in \mathbb{R}^k$. Further, $\mathbf{Z} \in \mathbb{R}^{n \times m}$ is a deterministic random effects design matrix that maps the random effects $\mathbf{b}$ to the corresponding observations. Often, this is a binary incidence matrix, but it can also contain predictor variables. Note that modeling the $\tau$-quantiles $\boldsymbol{\mu}$ via \eqref{me_model_def} is understood in the usual quantile regression sense and is equivalent to assuming $y_i=\mu_i+\varepsilon_i,$
where $Q_\tau(\varepsilon_i\mid \mu_i)=0$. 

We follow \citet{YU2001437} and adopt a generalized Bayes approach \citep{bissiri2016general} using the asymmetric Laplace likelihood $p(y \mid \mu, \lambda)$ given in \eqref{al_likelihood}. The Gibbs posterior of such a model is then proportional to $p(\mathbf{y}|\mathbf{b}, \boldsymbol{\beta},\lambda)\pi(\mathbf{b}|\boldsymbol{\theta})$, where $p(\mathbf{y}|\mathbf{b}, \boldsymbol{\beta},\lambda)=\prod_{i=1}^n p(y_i \mid \mu_i, \lambda)$ and $\pi(\mathbf{b}|\boldsymbol{\theta})$ is the prior density. Neither the posterior nor the marginal likelihood of the hyperparameters is available in closed form. A classical Laplace approximation \citep{tierney1986accurate} approximates the sum of the log-likelihood $ \log p(\mathbf{y}|\mathbf{b}, \boldsymbol{\beta},\lambda)$ and the quadratic log-prior $\log\pi(\mathbf{b}|\boldsymbol{\theta})$ by a second-order Taylor expansion centered at the posterior mode
\begin{equation}\label{mode_def}
\hat{\mathbf{b}} = \arg\max_{\mathbf{b}} \log p(\mathbf{y}|\mathbf{b}, \boldsymbol{\beta},\lambda) + \log \pi(\mathbf{b}|\boldsymbol{\theta}).
\end{equation}
 This leads to an approximate multivariate normal posterior distribution $\mathcal{N}(\mathbf{\hat{b}}, \boldsymbol{\hat{\Sigma}})$,
where the covariance matrix $\boldsymbol{\hat{\Sigma}} \in \mathbb{R}^{m\times m}$ is the inverse of the negative Hessian of $\log p(\mathbf{y}|\mathbf{b}, \boldsymbol{\beta},\lambda) + \log \pi(\mathbf{b}|\boldsymbol{\theta})$ evaluated at $\mathbf{\hat{b}}$:
\begin{equation}
\boldsymbol{\hat{\Sigma}} = (\boldsymbol{K_{\theta}^{-1}} + \mathbf{W})^{-1},\boldsymbol{K_{\theta}^{-1}}=-\nabla^2_{\mathbf{b}} \text{log}\pi(\mathbf{b})|_{\mathbf{b} = \mathbf{\hat b}}, \mathbf{W} := -\nabla^2_{\mathbf{b}} \log p(\mathbf{y}|\mathbf{b}, \boldsymbol{\beta},\lambda)|_{\mathbf{b} = \mathbf{\hat{b}}}.
\end{equation}
From the Laplace approximation, one also obtains an approximate marginal likelihood $p(\mathbf{y}|\boldsymbol{\beta},\lambda, \boldsymbol{\theta}) = \int p(\mathbf{y}|\mathbf{b}, \boldsymbol{\beta},\lambda)\pi(\mathbf{b}|\boldsymbol{\theta})d\mathbf{b}$:
\begin{equation}\label{mar_like_laplace}
   p(\mathbf{y}|\boldsymbol{\beta},\lambda, \boldsymbol{\theta}) \approx p(\mathbf{y}|\hat{\mathbf{b}}, \boldsymbol{\beta},\lambda)\pi(\hat{\mathbf{b}}|\boldsymbol{\theta})|\det(\hat{\boldsymbol{\Sigma}})|^{1/2} \big(2\pi\big)^{m/2},
\end{equation}
which can be used to estimate hyperparameters. However, for the asymmetric Laplace likelihood, the standard Laplace approximation is not applicable since the second derivative of the log-likelihood is zero almost everywhere. 

In the following, we write $p(y \mid\mathbf{b}, \mathbf{\psi})$ for a generic likelihood, where $\psi$ denotes all parameters other than the latent random effects $\mathbf{b}$. Depending on the context, $p(y \mid\mathbf{b}, \mathbf{\psi})$ is the asymmetric Laplace likelihood or another loss-induced likelihood, such as one associated with the $\varepsilon$-insensitive loss or the interval score. For the asymmetric Laplace model, $\mathbf{\psi}=(\mathbf{\beta},\lambda)$ with $\tau$ fixed, and we write $(\mathbf{\beta},\lambda)$ explicitly when convenient.

\subsection{Fisher-Laplace Approximation}\label{FL_approx}
The Fisher information measures the amount of information that a random variable carries about an unknown parameter and describes the geometry of the likelihood.
\begin{definition}[Fisher information]
\label{def:fisher_info}
For a parametric model $p(y|\mathbf{b}, \mathbf{\psi})$ with true parameter $\mathbf{b}_0 \in \mathbb{R}^m$ and nuisance parameters $\psi$, the Fisher information is 
\begin{equation}
    \mathbf{I}_{\mathbf{b}} \coloneqq \mathbb{E}_{y|\mathbf{b}_0}\left[\nabla_{\mathbf{b}} \log p(y|\mathbf{b},\mathbf{\psi})\nabla_{\mathbf{b}} \log p(y|\mathbf{b},\mathbf{\psi})^\top  \right]
\end{equation}
where $\nabla_{\mathbf{b}} \log p(y|\mathbf{b}, \mathbf{\psi})$ is the score function.
\end{definition}
Under regularity conditions, $\mathbf{I}_{\mathbf{b}} = -\mathbb{E}_{y|\mathbf{b}_0}\left[\nabla_{\mathbf{b}}^2 \log p(y|\mathbf{b},\mathbf{\psi})\right]$ holds. This latter representation is often used in the asymptotic justification of the Laplace approximation \citep{van2000asymptotic}. However, for the asymmetric Laplace likelihood, the dual representation breaks down. But while the Hessian is zero, the Fisher information for the model in \eqref{me_model_def} with the asymmetric Laplace likelihood is 
\begin{equation}\label{eq:fisher_info}
\mathbf I_{n, \mathbf b}=\mathbb{E}_{\mathbf{y}|\mathbf{b}_0}\left[\nabla_{\mathbf{b}} \log p(\mathbf{y}|\mathbf{b},\mathbf{\psi})\nabla_{\mathbf{b}} \log p(\mathbf{y}|\mathbf{b},\mathbf{\psi})^\top  \right]=\mathbf{Z}^T\text{diag}\left(\frac{\tau(1-\tau)}{\lambda^2}\right)\mathbf{Z},
\end{equation}
where $\mathbf D=\text{diag}\left(\tau(1-\tau)/\lambda^2\right)\in \mathbb{R}^{n\times n}$ is a diagonal matrix with the constant $\tau(1-\tau)/\lambda^2$ on the diagonal. This suggests performing a quadratic approximation for the log-likelihood using the Fisher information instead of the Hessian. 

The formal justification for such an approximation is given by the local asymptotic normality (LAN) property of likelihoods that are differentiable in quadratic mean (DQM). 
Intuitively, DQM at $\mathbf{b}$ means that the likelihood is differentiable at $\mathbf{b}$ for most realizations of the random variable $Y$\footnote{Formally, a likelihood $p(y|\mathbf{b})$ is DQM at $\mathbf{b}$ if there exists a random variable $\ell_\mathbf{b}(y)$ with $\mathbb{E}_{y\sim\mathbf{b}}[\ell_{\mathbf{b}}(y)] = 0$ and $\mathbb{E}_{y \sim\mathbf{b}}[\ell_{\mathbf{b}}(y)^2] < \infty$ such that 
$$\int \left(\sqrt{p(y| {\mathbf{b} + \mathbf{t}}}) - \sqrt{p(y|\mathbf{b})} - \frac{\mathbf{t}^\top}{2}\ell_{\mathbf{b}}(y)\sqrt{p(y|\mathbf{b})}\right)^2 \mu(dy) = o(\|\mathbf{t}\|^2)$$
as $\mathbf{t} \to 0$, where $\mu$ is a dominating measure and $\ell_{\mathbf{b}}(\cdot)$ is the score function in quadratic mean.}. This property is sufficient to allow for an asymptotic quadratic expansion of the likelihood ratio process, using the Fisher information, known as LAN expansion  \citep{le2000asymptotics}. In the i.i.d. setting, the LAN property implies the following quadratic representation for every compact set $K\subset\mathbb R^m$ and every $\mathbf{t}\in K$:
\begin{equation}
    \log\frac{p(\mathbf{y}|\mathbf{b} + \mathbf{t}/\sqrt{n}, \mathbf{\psi})}{p(\mathbf{y}|\mathbf{b}, \mathbf{\psi})} = \mathbf{t}^\top\frac{1}{\sqrt{n}} \sum_{i=1}^n \nabla_{\mathbf{b}} \log p(y_i| \mathbf{b}, \mathbf{\psi}) -\frac{1}{2}\mathbf{t}^\top\mathbf{I}_{\mathbf{b}}\mathbf{t} + o_p(1)
\end{equation}
 with non-singular Fisher information $\mathbf{I}_\mathbf{b}$, and where $o_p(1)$ indicates that the remainder goes to zero in $p$-probability as the sample size $n \rightarrow \infty$. The root-$n$ scaling is natural in the i.i.d.\ case because then the Fisher information is additive, i.e., the information in $n$ observations satisfies $\mathbf I_{n,\mathbf{b}}=n\,\mathbf I_{\mathbf b}$.
 
In our setting, the Fisher information no longer necessarily scales as a multiple of $n$. Instead, the full-sample Fisher information takes the form
$\mathbf I_{n, \mathbf{b}}=\mathbf Z^\top \mathbf D\,\mathbf Z$. In this case, the natural local scaling is determined by the inverse square root of the information matrix, that is $
\mathbf b + \mathbf I_{n, \mathbf{b}}^{-1/2}\mathbf{t}
$, with the requirement that $\mathbf I_{n, \mathbf{b}}^{-1/2}$ converges to the zero matrix as $n \rightarrow \infty$. The LAN property, expressed using information-scaled neighborhoods, is then for every compact set $K\subset\mathbb R^m$ and every $\mathbf{t}\in K$,
\begin{equation}
\label{eq:lan_info_scaled}
\log\frac{p(\mathbf y \mid \mathbf b + \mathbf I_{n, \mathbf b}^{-1/2}\mathbf{t},\mathbf{\psi})}
{p(\mathbf y \mid \mathbf b,\mathbf{\psi})}
=
\mathbf{t}^\top \mathbf{I}_{n, \mathbf b}^{-1/2}
\sum_{i=1}^n \nabla_{\mathbf b}\log p(y_i\mid\mathbf b,\mathbf{\psi})
-\frac12 \mathbf{t}^\top \mathbf{t}
+ o_p(1).
\end{equation}
For example, for a single-level grouped random effects model, the Fisher information is
$\mathbf I_{n, \mathbf b} =\frac{\tau(1-\tau)}{\lambda^2}\,\mathbf Z^\top \mathbf Z$ for the asymmetric Laplace likelihood. Since $\mathbf Z^\top \mathbf Z$ is diagonal with entries equal to the numbers of random effect occurrences $n_j$ for random effects $j$, $j=1,\dots,m$, the information-scaled neighborhood is proportional to $(\mathbf Z^\top \mathbf Z)^{-1/2}=\mathrm{diag}\!\left(1/\sqrt{n_1},\dots,1/\sqrt{n_m}\right)$.
Thus, the classical root-$n$ scaling is recovered group-wise as $1/\sqrt{n_j}$, and the local curvature remains $\tau(1-\tau)/\lambda^2$.

This provides a principled curvature estimate and enables a Laplace approximation, even in the absence of a second derivative. In summary, the Fisher-Laplace approximation uses the expected Fisher information to approximate the local curvature of the log-likelihood to obtain 
 $
 \boldsymbol{\hat{\Sigma}} = (\boldsymbol{K_{\theta}^{-1}} + \mathbf{I}_{n, \mathbf{b}})^{-1},
 $
which is used as approximate posterior covariance matrix and in \eqref{mar_like_laplace} for the approximate marginal likelihood.

We formalize this approach by proving a consistency result for the Fisher-Laplace approximation to the marginal likelihood. In particular, we show that the relative error between the Laplace approximation and the true marginal likelihood vanishes in probability. Note that classical consistency proofs for Laplace approximations (e.g., \cite{tierney1986accurate}) rely on stringent regularity conditions, requiring the log-likelihood to be four times continuously differentiable. These conditions fail to hold in our setting due to the non-smooth, piecewise-linear nature of the asymmetric Laplace likelihood. As a result, we develop a different approach.
\begin{theorem}[Consistency of the Fisher-Laplace approximation]\label{thm:FL_consistency}
Assume the model is correctly specified, and let $\mathbf{b}_0$ denote the true parameter and $\mathbb{P}_0$ the corresponding data generating distribution. 
Let $\ell_n(\mathbf{b}) := \sum_{i=1}^n \log p(y_i \mid \mathbf{b}, \mathbf{\psi})$ denote the log-likelihood, $z_n := \int_{\mathbb{R}^m} p(\mathbf{y}\mid\mathbf{b}, \mathbf{\psi}) \pi(\mathbf{b})\,d\mathbf{b}$ the marginal likelihood, and $\hat{\mathbf b}_n$ the posterior mode. Define the Laplace approximation
\begin{equation}
    z_n^{LA}  := p(\mathbf y \mid \hat{\mathbf b}_n, \mathbf{\psi})\,  \pi(\hat{\mathbf b}_n)\, \big(2\pi\big)^{m/2} \big|\det(\mathbf I_{n, \mathbf b_0} + \mathbf{K}^{-1})\big|^{-1/2},
\end{equation}
where $\mathbf I_{n,\mathbf b_0}$ defined in \eqref{eq:fisher_info} is the Fisher information at $\mathbf b_0$ and $\mathbf{K}^{-1}:= - \nabla_{\mathbf{b}}^2 \log \pi(\mathbf{b})\big|_{\mathbf b=\mathbf b_0}$. Assume:
\begin{enumerate}

\item \textbf{(LAN representation)}  
The log-likelihood admits a LAN expansion around $\mathbf b_0$ with symmetric positive definite Fisher information $\mathbf I_{n, \mathbf b_0}$, in the sense that for every compact set $K \subset \mathbb{R}^m$,
\begin{equation}
\sup_{\mathbf{t} \in K}
\Big|
\log\frac{p(\mathbf y \mid \mathbf b_0 + \mathbf{I}_{n,\mathbf{b_0}}^{-1/2}\mathbf{t}, \mathbf{\psi})}{p(\mathbf y \mid \mathbf b_0, \mathbf{\psi})}
-
\mathbf{t}^\top\mathbf{I}_{n,\mathbf{b_0}}^{-1/2}  \sum_{i=1}^n \nabla_{\mathbf b}\log p(y_i \mid \mathbf b_0, \mathbf{\psi})
+
\tfrac12 \mathbf{t}^\top \mathbf{t}
\Big|
\;\xrightarrow{\mathbb{P}_0}\; 0 .
\end{equation}
and the normalized score $
\Delta_n := \mathbf{I}_{n,\mathbf{b_0}}^{-1/2}\sum_{i=1}^n \nabla_{\mathbf b}\log p(y_i \mid \mathbf b_0, \mathbf{\psi})
$
is tight, i.e.\ $\Delta_n=O_{\mathbb P_0}(1)$.
\item \textbf{(Asymptotic linearity of the mode)}  
The posterior mode $\hat{\mathbf b}_n$ exists and satisfies
\[
\mathbf{I}_{n,\mathbf{b_0}}^{1/2}(\hat{\mathbf b}_n - \mathbf b_0)
=
\mathbf I_{n, \mathbf b_0}^{-1/2}
\sum_{i=1}^n \nabla_{\mathbf b}\log p(y_i \mid \mathbf b_0, \mathbf{\psi})
+
o_{\mathbb{P}_0}(1).
\]

\item \textbf{(Separation and information growth)}  
Let $\|\mathbf A\|_{\mathrm{op}} := \sup_{\|x\|_2=1}\|\mathbf A x\|_2$ denote the operator norm induced by the Euclidean norm. There exists $\delta>0$ and a sequence $M_n \to \infty$ with $M_n\| \mathbf{I}_{n, \mathbf{b}_0}^{-1/2}\|_{\mathrm{op}} \rightarrow 0$ and
$ \log\det(\mathbf I_{n,\mathbf b_0}) = o(n)$ such that
\[
\mathbb P_{0}\!\left(
\inf_{\|\mathbf{I}_{n,\mathbf{b_0}}^{1/2}(\mathbf b-\hat{\mathbf b}_n)\| \ge M_n}
\frac{1}{n}\big[\ell_n(\mathbf b)-\ell_n(\hat{\mathbf b}_n)\big]
\le -\delta
\right)
\;\longrightarrow\; 1 .
\]

\item \textbf{(Prior regularity)}  
The prior density $\pi(\mathbf b)$ is positive, continuous and twice continuously differentiable in a neighborhood of $\mathbf b_0$.
\end{enumerate}

Then, for every $\varepsilon>0$,
\[
\mathbb P_0\!\left(\left|\frac{z_n^{LA}}{z_n}-1\right|>\varepsilon\right)\longrightarrow 0
\qquad\text{as } n\to\infty.
\]
\end{theorem}

\paragraph{\textbf{Discussion of assumptions.}}
Assumption~(i) is a standard LAN condition written on the natural information scale. It states that, after rescaling by $\mathbf I_{n,\mathbf b_0}^{-1/2}$, the log-likelihood ratio is asymptotically quadratic with unit curvature; in the regular i.i.d.\ case, such expansions follow from differentiability in quadratic mean \citep[Theorem~7.2]{van2000asymptotic}, which holds for the asymmetric Laplace likelihood. In regression settings, verifying the same expansion additionally requires standard regularity assumptions on the design and a central limit theorem for the normalized score $\Delta_n$ \citep{koenker2005quantile}. Assumption~(ii) is an asymptotic linearity condition for the posterior mode, ensuring that $\hat{\mathbf b}_n$ lies on the same local scale as the LAN expansion and allowing the quadratic approximation to be recentered at $\hat{\mathbf b}_n$. For the asymmetric Laplace log-likelihood, Assumption~(ii) holds under standard quantile-regression regularity conditions \citep{koenker2005quantile}, namely identifiability and suitable design regularity (e.g.\ full-rank conditions) together with the usual Bahadur-type asymptotic linearity of regression quantiles. Since the asymmetric Laplace criterion coincides with the pinball loss up to constants, the posterior mode has the same leading first-order expansion, and the prior contributes only a lower-order term under information growth. Assumption~(iii) is a separation condition ensuring that the contribution to the marginal likelihood from outside the local quadratic region is negligible. The condition $M_n\to\infty$ together with $M_n\|\mathbf I_{n,\mathbf b_0}^{-1/2}\|_{\mathrm{op}}\to 0$ means that the neighborhood expands in the local $\mathbf t$-coordinates while still shrinking in the original $\mathbf b$-coordinates. Our formulation also makes the required information growth explicit, which is useful in settings where the information need not accumulate simply as a scalar multiple of $n$. For the asymmetric Laplace log-likelihood, Assumption~(iii) holds provided that information accumulates in all identifiable random-effect directions, so that $
\|I_{n,b_0}^{-1/2}\|_{\mathrm{op}} \to 0$ and $\log \det(I_{n,b_0}) = o(n)$, and that the empirical pinball-loss objective is uniformly separated from its maximizer outside the local information-scaled neighborhood. The latter conditions are standard quantile-regression assumptions. For single-level grouped random effects, the information-growth requirement is typically guaranteed by $
\min_j n_j \to \infty$. For crossed random effects, the analogous requirement is that the design is identifiable and sufficiently replicated so that the smallest eigenvalue of \(Z^\top D Z\) diverges in all identifiable directions.  Assumption~(iv) is the usual prior regularity condition ensuring that the prior behaves smoothly and does not vanish near the truth. This clearly holds for Gaussian priors. These assumptions are discussed in more detail in Appendix~\ref{app:discussion_assumptions_misspec}. A proof of Theorem~\ref{thm:FL_consistency} is deferred to Appendix~\ref{appendix:bvm_theorem}.

While Theorem~\ref{thm:FL_consistency} shows that the Fisher-Laplace approximation is asymptotically valid under correct model specification, in practice the asymmetric Laplace likelihood serves only as a pseudo-likelihood. In the following section, we address this. However, estimating the scale parameter $\lambda$ of the asymmetric Laplace likelihood helps to mitigate the extent of misspecification. Figure \ref{fig:estimated_scale} in Appendix \ref {fisher_laplace_missp} illustrates this empirically: using an estimated $\hat\lambda$ yields a Fisher curvature that tracks the correct asymptotic curvature much more closely than fixing $\lambda=1$. 

\subsection{Laplace Approximation under Model Misspecification}
\label{sec:misspecified_laplace}

In general parametric modeling, one posits a family of distributions $\mathcal{P} = \{\mathbb{P}_{\mathbf{b}} : \mathbf{b} \in B\}$ indexed by a parameter $\mathbf{b} \in B \subseteq \mathbb  {R}^m$. The model is \emph{misspecified} if the true data-generating distribution $\mathbb{P}^*$ does not belong to $\mathcal{P}$. Under standard regularity conditions, the maximum likelihood estimator (MLE) $\hat{\mathbf{b}}_n := \argmax_{\mathbf{b} \in B} \log p(\mathbf{y}|\mathbf{b})$ converges to the \emph{pseudo-true parameter} defined as:
\[
 \mathbf{b}^* = \arg\min_{\mathbf{b} \in B} \text{KL}(\mathbb{P}^* \| \mathbb{P}_{\mathbf{b}})  = \arg\min_{\mathbf{b} \in B}  -\mathbb{E}_{\mathbb{P}^*}\frac{\log \text{d}\mathbb{P}_{\mathbf{b}}}{\log \text{d}\mathbb{P}^*}.
\]
For quantile regression, the asymmetric Laplace log-likelihood coincides (up to constants) with the pinball loss, and hence minimizing the KL divergence is equivalent to minimizing the expected pinball loss under the true data-generating process.
Consequently, even under misspecification, the pseudo-true parameter remains the true conditional $\tau$-quantile, since the pinball loss is a proper scoring rule for quantiles \citep{ehm2016quantiles}, assuming the model class is capable of representing the true quantile and the parameters are identifiable.

A Bayesian analogue of this perspective is provided by the \emph{generalized Bayes} framework of \citet{bissiri2016general}, in which a loss function replaces the log-likelihood and posterior beliefs are updated according to $\pi_n(\mathbf{b}  \mid y_{1:n}) \propto \pi(\mathbf{b} )\,\exp\!\Big(
    -\sum_{i=1}^n \ell(y_i, \mathbf{b} )
\Big),$
for any suitably chosen loss function $\ell(\cdot,\cdot)$. \cite{bissiri2016general} show that this
construction yields coherent belief updating. 
In such misspecified settings, we can analyze the model using tools from M-estimation theory. Under regularity conditions \citep{van2000asymptotic}, the asymptotic behavior of M-estimators can be studied by considering the population objective function $\mathbf{b}  \mapsto \mathbb{E}_{y \sim \mathbb{P}^*}[-\log p(y|\mathbf{b}, \psi )]$ rather than the sample objective $\mathbf{b}  \mapsto -\log p(\mathbf{y}|\mathbf{b}, \mathbf{\psi})$. This approach allows us to compute derivatives of the expected loss, thus overcoming the degenerate Hessian of the pinball loss, and yielding a meaningful notion of curvature: 
\begin{equation}
\mathbf{H}_{\mathbf{b}} :=  \nabla^2_{\mathbf b}\, \mathbb{E}_{y \sim \mathbb{P}^*}\!\left[ -\log p(y \mid \mathbf{b}, \mathbf{\psi}) \right].
\label{eq:H_b_general}
\end{equation}
For the asymmetric Laplace likelihood and the model in \eqref{me_model_def}, we have
\begin{equation}
\mathbf{H}_{n, \mathbf{b}} =  \nabla^2_{\mathbf b}\, \mathbb{E}_{\mathbf{y} \sim \mathbb{P}^*}\!\left[ -\log p(\mathbf{y} \mid \boldsymbol{\mu}, \lambda) \right] =\mathbf{Z}^T\text{diag}\left(\frac{ \boldsymbol{f}^*(\boldsymbol{\mu}) }{\lambda}\right)\mathbf{Z}
\label{eq:H_b_ald}
\end{equation}
where $\boldsymbol{f}^*(\boldsymbol{\mu})$ is the density of the data-generating $\mathbb{P}^*$ evaluated at $\boldsymbol{\mu} = F(\mathbf{X}) + \mathbf{Zb}$. When the model is correctly specified, $\mathbb{P}^* = \text{AL}(\tau, \lambda)$, $\mathbf{H}_{n, \mathbf{b}}$ equals the Fisher information. The quantity $\mathbf{H}_{n, \mathbf{b}}$ also provides an intuitive interpretation of why the curvature relates to an estimator's asymptotic variance: a larger density at the target parameter implies more available information and consequently smaller variance. In addition, this highlights why estimation of extreme quantiles is intrinsically challenging, as the density at extreme quantiles is typically small. 

Next, we present a result that makes the above arguments rigorous. Theorem~\ref{thm:misspec_laplace} provides a consistency result for a Laplace approximation to the marginal likelihood based on $\mathbf H_{n, \mathbf b^*}$ in \eqref{eq:H_b_ald} in the misspecified case.

\begin{theorem}[Consistency of the Laplace approximation under misspecification]\label{thm:misspec_laplace}
Let $\mathbb{P^*}$ denote the true data generating distribution, and let $\mathbf b^*$ denote the pseudo-true parameter minimizing the expected negative log-likelihood with respect to $\mathbb{P^*}$. Denote by 
$
\ell_n(\mathbf b) := \sum_{i=1}^n \log p(y_i \mid \mathbf b, \mathbf{\psi})
$
the log-likelihood, $z_n := \int_{\mathbb R^m} p(\mathbf y \mid \mathbf b, \mathbf{\psi})\,\pi(\mathbf b)\,d\mathbf b$ the marginal likelihood, and $\hat{\mathbf b}_n$ the posterior mode. Define the Laplace approximation
\begin{equation}
    z_n^{LA} :=  p(\mathbf y \mid \hat{\mathbf b}_n, \mathbf{\psi})\, \pi(\hat{\mathbf b}_n)\, \big(2\pi\big)^{m/2} \big|\det(\mathbf H_{n, \mathbf b^*} + \mathbf{K}^{-1})\big|^{-1/2},
\end{equation}
where $\mathbf H_{n, \mathbf b^*}$ is the expected Hessian defined in \eqref{eq:H_b_ald} evaluated at $\mathbf{b^*}$ and $\mathbf{K}^{-1}:= - \nabla_{\mathbf{b}}^2 \log \pi(\mathbf{b})\big|_{\mathbf b=\mathbf b^*}$. Assume that Assumptions (i) - (iv) of Theorem~\ref{thm:FL_consistency} hold with $\mathbf{H}_{n,\mathbf{b^*}}$ replacing $\mathbf I_{n,\mathbf b_0}$ throughout. Then, for every $\varepsilon>0$,
\[
\mathbb P^*\!\left(\left|\frac{z_n^{LA}}{z_n}-1\right|>\varepsilon\right)\longrightarrow 0
\qquad\text{as } n\to\infty.
\]
\end{theorem}
 
Compared to Theorem~\ref{thm:FL_consistency}, the key difference is that the local curvature is now governed by the expected Hessian $\mathbf H_{n,\mathbf b^\ast}$ of the population criterion rather than by the Fisher information. The assumptions are the direct misspecified analogues of those in Theorem~\ref{thm:FL_consistency}. Assumption~(i) requires a local asymptotic quadratic expansion on the $\mathbf H_{n,\mathbf b^\ast}$-scale, and Assumption~(ii) requires asymptotic linearity of the posterior mode on the same scale. These two assumptions are standard quantile regression and M-estimation assumptions \citep{koenker2005quantile, Miss-BvM}. Assumption~(iii) guarantees separation and information growth with $\mathbf H_{n,\mathbf b^\ast}$, and Assumption~(iv) is the same prior regularity condition as before. The assumptions are discussed in detail after Theorem~\ref{thm:FL_consistency} and in Appendix~\ref{app:discussion_assumptions_misspec}. A proof is given in Appendix~\ref{appendix:bvm_theorem}.

Next, we show that the Laplace-approximated posterior under misspecification converges to the limiting generalized Bayes posterior.
\begin{corollary}[Bernstein-von Mises under misspecification]\label{cor:bvm_noprior}
Define the local parameter $\mathbf t := \mathbf H_{n,\mathbf b^*}^{1/2}(\mathbf b-\hat{\mathbf b}_n)$,
and let $q_n(\mathbf t)$ denote the posterior density of $\mathbf t$ induced by the posterior $\pi_n(\mathbf b\mid \mathbf y)$.
Under the assumptions of Theorem~\ref{thm:misspec_laplace}, it holds that
\[
\int_{\mathbb R^m}
\left|
q_n(\mathbf t)
-
\mathcal N\!\left(\mathbf t;\,0,\mathbf I_m\right)
\right|\,d\mathbf t
\;\xrightarrow{\mathbb P^*}\; 0.
\]
\end{corollary}
A proof for Corollary~\ref{cor:bvm_noprior} can be found in Appendix~\ref{app_proof_BVM}. The key insight from Theorem~\ref{thm:misspec_laplace} and Corollary~\ref{cor:bvm_noprior} is that, also under model misspecification, the log-likelihood exhibits concentration and a quadratic behavior in a local neighborhood of the parameter of interest. Therefore, a Laplace approximation performed with the expected curvature can capture the dominant behavior of the asymptotic posterior.

\subsection{Triangular Kernel Curvature Estimator}\label{sec:cc}
In the following, we propose a way to estimate the population curvature from data. Our triangular kernel curvature (TKC) estimator defined below approximates $\mathbf{H}_{\mathbf{b}^*}=\mathbf{Z}^T\text{diag}\left(\boldsymbol{f}^*(\boldsymbol{\mu^*})/\lambda\right)\mathbf{Z}$
by estimating the scalar density term $\boldsymbol{f}^*(\boldsymbol{\mu^*})/\lambda$ from local log-likelihood differences around the mode. Specifically, we first compute the decrease in log-likelihood (DLL) when moving $\pm\Delta \mu$, $\Delta \mu > 0$, away from $\hat{\boldsymbol{\mu}} =  F(\mathbf{X}) + \mathbf{Z\hat b}$, where $\hat{\mathbf{b}}$ is the mode:
$$
    DLL^U_{\Delta \mu} (\hat{\boldsymbol{\mu}}) = \text{log} p(\mathbf{y}|\hat{\boldsymbol{\mu}}) - \text{log} p(\mathbf{y}|\hat{\boldsymbol{\mu}} + \Delta \mu), ~~
    DLL^L_{\Delta \mu} (\hat{\boldsymbol{\mu}}) = \text{log} p(\mathbf{y}|\hat{\boldsymbol{\mu}}) - \text{log} p(\mathbf{y}|\hat{\boldsymbol{\mu}} - \Delta \mu).
$$
Then, we use a centered second difference approach to obtain the approximation
\begin{equation}\label{eq:cc}
   \hat{C}_{\Delta \mu}(\hat{\boldsymbol{\mu}}) = \frac{DLL^U_{\Delta \mu} (\hat{\boldsymbol{\mu}}) + DLL^L_{\Delta \mu} (\hat{\boldsymbol{\mu}})}{n\Delta \mu^2}\in \mathbb{R}
\end{equation}
and
\begin{equation}\label{cc_approx}
    \hat{\mathbf{H}} = \mathbf{Z}^T\text{diag}\left(\hat{C}_{\Delta \mu}(\hat{\boldsymbol{\mu}})\right)\mathbf{Z}\in \mathbb{R}^{n\times n}.
\end{equation}

In practice, we choose $\Delta_\mu$ to balance locality and stability. If $\Delta_\mu$ is too small, the piecewise-linear structure of the asymmetric Laplace log-likelihood dominates; if it is too large, the log-likelihood is no longer well approximated by a quadratic function since it becomes increasingly asymmetric. We therefore search over candidate values for $\Delta_\mu$ exceeding a minimum likelihood-drop threshold and select the one that gives the best quadratic fit, measured by the R-squared between the true log-likelihood and its quadratic approximation based on $\hat{C}_{\Delta \mu}(\hat{\boldsymbol{\mu}})$ at values $\hat{\boldsymbol{\mu}} \pm 0.5 \Delta \mu$ and $\hat{\boldsymbol{\mu}} \pm \Delta \mu$. A minimum likelihood-drop threshold implies smaller values of $\Delta \mu$ as the sample size $n$ grows, since the posterior mass becomes concentrated in an increasingly smaller neighborhood around the mode for larger sample sizes. This approach thus yields a local approximation window that adapts to the increasing concentration of the likelihood, and the curvature estimator remains in a regime where the local quadratic approximation is accurate as $n$ increases.  We study the sensitivity of the TKC estimator to the choice of the minimum likelihood-drop threshold in Appendix~\ref{app:sensitivity_tkc}. We find that quantile prediction and hyperparameter estimation accuracy are insensitive to this tuning parameter and remain highly stable across threshold magnitudes ranging from $10^{-2}$ to $10^2$.

Figure~\ref{fig:tkc_curvature} compares our proposed approximations to the log-likelihood for a single-level random effects model for a correctly specified and a misspecified likelihood. We consider a quantile level of $\tau=0.8$, $n=200$ observations, and one Gaussian random effect $b_0\sim\mathcal N(0,1)$. In the correctly specified setting (left), the observations are generated from the asymmetric Laplace model (with scale $\lambda=0.1$).
In the misspecified setting (right), the observations are generated by adding Gaussian noise with standard deviation $0.1$ (shifted so that $b_0$ corresponds to the $\tau$-quantile).
Under misspecification, the Fisher-Laplace curvature (with an estimated scale parameter) deviates from the shape of the log-likelihood, whereas the TKC approximates it accurately in a neighborhood of the mode.
\begin{figure}[ht!]
    \centering
    \includegraphics[width=1\linewidth]{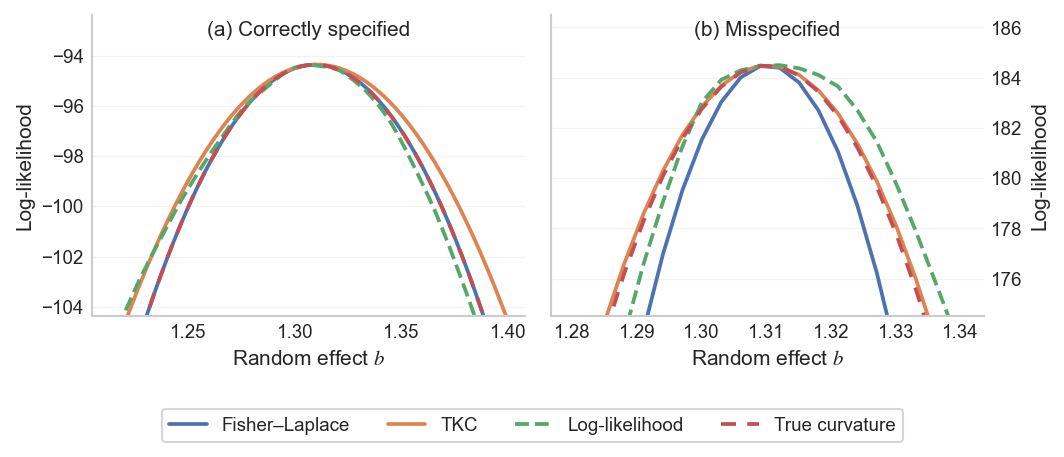}
    \caption{ Comparison of TKC and Fisher-Laplace approximations to the log-likelihood for a single-level random effects model for a correctly specified likelihood (left) and a misspecified likelihood (right). The curve labeled ``true curvature'' corresponds to the asymptotic curvature at the true $b_0$.}
    \label{fig:tkc_curvature}
\end{figure}

\subsubsection{Connection to Kernel Density Estimation}

The triangular kernel curvature estimator~\eqref{eq:cc} can be interpreted as a nonparametric kernel density estimator of the residuals $\mathbf y - \hat{\boldsymbol{\mu}}$ evaluated at zero. A direct calculation (see Appendix~\ref{app:proof_cc}) shows that the TKC estimator equals a kernel density estimator with triangular kernel and bandwidth $h=\Delta \mu$:
\begin{equation}
\label{eq:cc_triangular_kernel}
\hat C_{h}(\hat{\boldsymbol{\mu}}) = \frac{1}{n} \sum_{i=1}^n h^{-1} \left(1-\frac{|( y_i-\hat{{\mu}}_i|}{h}\right) \mathbbm{1}\{|( y_i-\hat{{\mu}}_i|<h\}.
\end{equation}
This is not coincidental, as the asymptotic curvature is equal to the data generating density evaluated at the quantile. This connection allows us to derive asymptotic rates for the bandwidth $\Delta \mu$ for consistent curvature estimation. For this, we impose the following homoscedasticity assumption at the target quantile level.
\begin{assumption}[Constant curvature (CC) at the $\tau$-quantile]
\hypertarget{ass:cc}{}
For $\boldsymbol{\mu}=F(\mathbf X)+\mathbf Z\mathbf b$, the conditional density of $Y_i$ under $\mathbb P^*$ evaluated at $\mu_i$ is constant across observations: $f_1^*(\mu_1)=\cdots=f_n^*(\mu_n)=:c^*(\boldsymbol{\mu})>0.$
\end{assumption}
This condition is satisfied, for example, in the standard setup $y_i=\mu_i+\varepsilon_i$ with $Q_\tau(\varepsilon_i\mid \mu_i)=0$, whenever there is a common conditional error density at zero, $f_{\varepsilon_i\mid \mu_i}(0)=c$ for all $i$. Under this assumption, the diagonal matrix in \eqref{eq:H_b_ald} reduces to a scalar multiple of the identity, and the population curvature simplifies to
\[
\mathbf H_{n, \mathbf b}
=
\frac{c^*(\boldsymbol{\mu})}{\lambda}\,\mathbf Z^\top \mathbf Z.
\]
Note that Theorem~\ref{thm:misspec_laplace} does not require the constant-curvature assumption. Assumption \hyperlink{ass:cc}{(CC)} is introduced only for the following consistency result of the TKC estimator and the plug-in consistency result in Propositions~\ref{prop:cc_consistency} and~\ref{prop:laplace_plugin_cc}, respectively. Empirically, we find the method remains robust even when this assumption is violated (e.g., under the heteroscedastic noise scenarios evaluated in Section~\ref{sec:simulated_exp}). This robustness may partly be explained by the fact that the scalar TKC estimator captures an average local curvature across observations. Although the estimator does not recover observation-specific curvature under heteroscedasticity, this scalar approximation may still be effective for the determinant term in the Laplace marginal likelihood when the variation in individual curvatures is moderate.

\begin{proposition}[Consistency of the triangular kernel curvature estimator]
\label{prop:cc_consistency}
Assume the true data density $\boldsymbol{f}^*(\boldsymbol{\mu}^*)$ at the true quantile $\boldsymbol{\mu}^*$ is positive and continuous, and that the constant curvature assumption~\hyperlink{ass:cc}{(CC)} holds. If the bandwidth $\Delta \mu$ satisfies $n\Delta \mu \to \infty$ and $\Delta \mu \to 0$ as $n \to \infty$, then the triangular kernel curvature estimator $\hat{C}_{\Delta \mu}(\boldsymbol{\mu}^*)$ at $\boldsymbol{\mu}^*$ (\ref{eq:cc_triangular_kernel}) satisfies:
\begin{equation}
\hat{C}_{\Delta \mu}(\boldsymbol{\mu}^*) \xrightarrow{\mathbb{P}^*} c^*( \boldsymbol{\mu}^*)/\lambda
\end{equation}
\end{proposition}
The proof follows from connecting our triangular kernel curvature approach to a triangular kernel density estimator in combination with standard arguments for consistency of kernel density estimators (see Appendix~\ref{app:proof_cc}).
The analysis presented here assumes that the true conditional quantile is known when constructing the estimator. In practice, both the density and the quantile must be estimated simultaneously. This plug-in setting requires slower bandwidth decay rates than in the oracle case where $\boldsymbol{\mu}^*$ is known \citep{POWELL1984303, Kato2012}. To derive admissible bandwidth rates in the plug-in setting, one would need model-specific convergence rates for the quantile estimator $\hat{\boldsymbol{\mu}}$, which is beyond the scope of this work.

The following result shows that the TKC Laplace approximation is consistent when replacing $H_{n, \mathbf b^*}$ in Theorem~\ref{thm:misspec_laplace} with a plug-in curvature matrix.
\begin{proposition}[Laplace-approximated marginal log-likelihood consistency with plug-in curvature]
\label{prop:laplace_plugin_cc}
Assume condition~\hyperlink{ass:cc}{(CC)} holds.
Let $z_n$ be the marginal likelihood, and let $z_n^{LA}$ be the Laplace approximation of Theorem~\ref{thm:misspec_laplace} based on the population curvature
$
\mathbf H_{n,\mathbf b^*}=\frac{c^*(\boldsymbol{\mu}^*)}{\lambda}\,\mathbf Z^\top\mathbf Z.
$
Let $\hat c_n(\hat{\boldsymbol{\mu}}_n)$ be a curvature estimator satisfying
$
\hat c_n(\hat{\boldsymbol{\mu}}_n)\xrightarrow{\mathbb P^*} c^*(\boldsymbol{\mu}^*)
$ and $
\mathbb P^*\!\big(\hat c_n(\hat{\boldsymbol{\mu}}_n)>0\big)\to 1.
$
Define the plug-in curvature matrix
$
\hat{\mathbf H}_{n,\hat{\mathbf b}}:=\frac{\hat c_n(\hat{\boldsymbol{\mu}}_n)}{\lambda}\,\mathbf Z^\top\mathbf Z
$
and let $$
   \hat z_n^{LA} :=  p(\mathbf y \mid \hat{\mathbf b}_n, \mathbf{\psi})\, \pi(\hat{\mathbf b}_n)\, \big(2\pi\big)^{m/2} \big|\det(\hat{\mathbf H}_{n,\hat{\mathbf b}}+ \mathbf{K}^{-1})\big|^{-1/2},
   $$
denote the corresponding plug-in Laplace approximation.
If $z_n^{LA}/z_n\xrightarrow{\mathbb P^*}1,$
then also
$
\frac{\hat z_n^{LA}}{z_n}\xrightarrow{\mathbb P^*}1.
$
\end{proposition}
The proof, deferred to Appendix~\ref{appendix:proof_plugin}, is based on the fact that the ratio between the plug-in and population Laplace approximations differs only through the determinant term of the curvature matrix. Under the \hyperlink{ass:cc}{(CC)} assumption, $\mathbf H_{n, \mathbf b^*}$ and $\hat{\mathbf H}_{n, \mathbf b^*}$ are scalar multiples of $\mathbf Z^\top\mathbf Z$, so this determinant difference depends only on $\hat c_n(\hat{\boldsymbol{\mu}}_n)$ and $c^*(\boldsymbol{\mu}^*)$. Consistency follows by a continuity and Slutsky argument.

While the connection between triangular kernel curvature and kernel density estimation provides theoretical insight into achieving consistency, in practice, asymptotic arguments alone are insufficient to determine an exact value of $\Delta \mu$ for finite samples. We therefore retain the theoretical intuition that $\Delta \mu$ should vanish asymptotically, but supplement this with a local search procedure described in Section~\ref{sec:cc} based on identifying well-conditioned quadratic approximations. This hybrid approach balances theoretical guidance with practical performance.

\subsubsection{Comparison to Smoothing Approaches}\label{sec:comparison_smoothing}

Several authors have proposed smoothing the asymmetric Laplace likelihood to facilitate optimization and sampling \citep{he2023smoothed, liu2025bayesian}. These approaches replace the asymmetric Laplace likelihood with a kernel-smoothed version to create curvature. Our approach leaves the likelihood unchanged and instead estimates the curvature of the expected loss, which is the object that governs the asymptotic quadratic approximation. As the smoothing bandwidth tends to zero, the Hessian of the smoothed likelihood converges to essentially the same density-based object, but our approach avoids smoothing bias in the posterior mode. An additional advantage of our framework is compatibility with higher-order kernel density estimators, which can improve convergence rates of the curvature estimator. In contrast, applying higher-order kernels with negative parts to smooth the likelihood itself would destroy log-concavity, rendering optimization non-convex and sampling intractable. Since we estimate the curvature of the {expected} loss rather than smoothing the loss itself, we would maintain the log-concave structure.

\subsection{Posterior Calibration under Misspecification}\label{sec:calibration}
Our Laplace approximation targets the generalized Bayes posterior induced by the asymmetric Laplace loss. Under model misspecification, this posterior need not be \emph{frequentist-calibrated}. In particular, the generalized posterior concentrates around the mode $\hat{\mathbf{b}}_n$ with asymptotic covariance $\mathbf{H}_{n, \mathbf{b}^*}^{-1}$ (see Corollary~\ref{cor:bvm_noprior}) rather than the frequentist asymptotic sandwich covariance $\mathbf{H}_{n, \mathbf{b}^*}^{-1} \mathbf{I}_{n, \mathbf{b}^*} \mathbf{H}_{n, \mathbf{b}^*}^{-1}$. Consequently, posterior credible intervals need not have correct asymptotic frequentist coverage. A remedy for this frequentist miscalibration is to apply a {sandwich variance correction} replacing $\mathbf{H}_{n, \mathbf{b}^*}^{-1}$ with $\mathbf{H}_{n, \mathbf{b}^*}^{-1} \mathbf{I}_{n, \mathbf{b}^*} \mathbf{H}_{n, \mathbf{b}^*}^{-1}$. 
In Section~\ref{sec:empirical_sandwich}, we demonstrate empirically that the TKC estimator enables effective sandwich covariance corrections and yields correct coverage for single-level grouped random effects models.

An alternative route to frequentist calibration is to temper the generalized posterior \citep{bissiri2016general}, $\pi_n(b \mid y) \propto p(y \mid b)^\alpha \pi(b),$ $\alpha \in (0,1]$. This inflates posterior uncertainty and can, in principle, be tuned to improve coverage under misspecification \citep{10.1093/biomet/asy054}. Ideally, we select $\alpha$ such that the tempered posterior covariance matches the sandwich covariance: $ \mathbf{H}_{n, \mathbf{b}^*}(\alpha)^{-1} \approx  \mathbf{H}_{n, \mathbf{b}^*}^{-1} \mathbf{I}_{n, \mathbf{b}^*} \mathbf{H}_{n, \mathbf{b}^*}^{-1}$.
In our setting, selecting $\alpha$ in a reliable data-driven way is difficult. We have attempted to select $\alpha$ using a data-splitting approach in which the random effects estimated on a validation dataset were treated as ground truth. However, this procedure proved ineffective, likely since the estimation error on the second dataset is of the same order as the confidence intervals we aim to calibrate.

\subsection{Uncertainty-Aware Prediction Intervals}
Quantile regression methods enable the construction of $(1-\alpha)\%$ prediction intervals $C_{1-\alpha}$ for new {responses} $y_{n+1}$ by estimating upper $(1-\alpha/2)$ and lower $(\alpha/2)$ quantiles of the distribution of $y_{n+1}|x_{n+1}$. To achieve calibrated {marginal} coverage, \citet{romano2019conformalized} propose conformalized quantile regression (CQR), which leverages the framework of conformal prediction to learn a scalar adjustment on a held-out calibration set. Furthermore \citet{rossellini2024integrating} extend this by multiplying the adjustment by a notion of {local epistemic uncertainty} to improve {conditional coverage}.
We can apply our GP quantile estimates within this framework, using the predicted quantiles as the starting point for the CQR procedure, and the {predictive variance as uncertainty measure} for the uncertainty-aware CQR. In Appendix~\ref{appendix:cqr}, we show that this approach improves conditional coverage, suggesting that even when our predictive distributions lack exact frequentist asymptotic calibration, they capture meaningful uncertainty information useful for downstream tasks.

\subsection{Computational Aspects}\label{Computational}
We employ the empirical Bayes (EB) method to estimate the (hyper-)parameters $\boldsymbol{\theta}$, $\boldsymbol{\beta}$, and $\lambda$ by maximizing the marginal likelihood. Since we work with a misspecified model, the interpretation of the marginal likelihood as the data density under a hierarchical model no longer holds. In this setting, we adopt the pragmatic justification provided by \cite{fong2020marginal}: the log-marginal likelihood corresponds to an {exhaustive leave-p-out cross-validation} criterion.

A Laplace approximation yields the following approximate marginal likelihood:
\begin{equation}
\label{eq:laplace}
\log p_{LA}(\mathbf{y}|\boldsymbol{\beta},\lambda, \boldsymbol{\theta}) = \log p(\boldsymbol{y}|\hat{\mathbf{b}},\boldsymbol{\beta}, \lambda) + \log \pi(\hat{\mathbf{b}}|\boldsymbol{\theta}) - \frac{1}{2}\log\det(\boldsymbol{\hat{\Sigma}^{-1}}) + \frac{m}{2}\log(2\pi),
\end{equation} 
where either
$$\boldsymbol{\hat{\Sigma}^{-1}} = \boldsymbol{K_{\theta}^{-1}} + \mathbf I_{n,\mathbf b} ~~ \text{or} ~~ \boldsymbol{\hat{\Sigma}^{-1}} = \boldsymbol{K_{\theta}^{-1}} + \hat{\mathbf{H}}_{n,b}$$
in the Fisher-Laplace in \eqref{eq:fisher_info} or the triangular kernel curvature approximation in \eqref{cc_approx}, respectively. The log-marginal likelihood $\log p_{LA}(\mathbf{y}|\boldsymbol{\beta},\lambda, \boldsymbol{\theta})$ is maximized with respect to $\boldsymbol{\theta}$, $\boldsymbol{\beta}$, and $\lambda$. In our experiments, this is done using a limited-memory BFGS (L-BFGS) algorithm which requires calculating gradients of $\log p_{LA}(\mathbf{y}|\boldsymbol{\beta},\lambda, \boldsymbol{\theta})$. Every time, the marginal likelihood or a gradient of it is calculated, the mode defined in \eqref{mode_def} needs to be (re-)determined. This optimization problem is challenging since the gradient of the asymmetric Laplace log-likelihood is piecewise constant and its second derivatives are zero almost everywhere. Despite this, the objective function is strongly log-concave, ensuring the existence of a unique global mode. In our experiments, we find $\hat{\mathbf{b}}$ with a quasi-Newton method using the Fisher information instead of a Hessian. For prediction, we use the Laplace-approximated posterior predictive distributions \citep[see, e.g., ][]{sigrist2022latent}.

For computational efficiency for Gaussian processes, our software implementation allows for using Vecchia approximations \citep{vecchia1988estimation, datta2016hierarchical}. Vecchia approximations often yield state-of-the-art approximation accuracy in spatial statistics \citep{guinness2019gaussian, rambelli2025accuracy}. To further speed up computations, we use the iterative methods presented in \citet{kundig2025scalable} for crossed random effects models. Besides, our software implementation also allows for using other GP approximations such as inducing points and Vecchia-inducing-points full-scale (VIF) approximations which can be more accurate for higher-dimensional inputs compared to Vecchia approximations \citep{gyger2025vecchia}.

\subsubsection{Software Implementation}
The methods presented in this article are implemented in the \texttt{GPBoost} library written in C++ with high-level Python and R interface packages\if1\blind{, see \url{https://github.com/fabsig/GPBoost}}\fi.

\section{Simulated Experiments}\label{Experiments}\label{sec:simulated_exp}
We compare our proposed Fisher and TKC Laplace approximations in simulation studies and real-world applications for both grouped random-effects and Gaussian process models. For each model type, we compare against state-of-the-art benchmark methods specifically designed for the respective model structure. In the simulation experiments, we assess the prediction accuracy using the root mean squared error (RMSE) with respect to the true latent quantile, and in the real-world experiments, we use the quantile loss (QL) since the true quantile is unknown:
\begin{align}
\text{RMSE} = \sqrt{\frac{1}{n_{test}}\sum_{i=1}^{n_{test}}\big(\hat{q}_\tau(x_i) - q_\tau(x_i)\big)^2}\quad,\quad
    \text{QL} = \frac{1}{n_{test}}\sum_{i=1}^{n_{test}} \rho_\tau\big(y_i - \hat{q}_\tau(x_i)\big),
\end{align}
where $q_\tau(x_i)$ is the true conditional $\tau$-th quantile, $\hat{q}_\tau(x_i)$ is the predicted conditional quantile, and $\rho_\tau(u) = u(\tau - \mathbbm{1}_{u < 0})$ is the pinball loss function. We present results for $\tau = 0.8$ in the main text and additionally report results for $\tau=0.95$ in Appendix~\ref{app:tau095}. Unless stated otherwise, all experiments employ out-of-sample validation using 10 replications with 75-25 train-test splits. We also evaluate the hyperparameter estimation accuracy using the RMSE between estimated and true parameters. 

For grouped random effects models, we compare against linear mixed models fitted using the \texttt{lqmm} package \citep{geraci2014linear} (version~1.5.8), Bayesian quantile regression using MCMC sampling via \texttt{brms} \citep{JSSv080i01} (version~2.22.0), and \texttt{bayesQR} \citep{benoit2017bayesqr} (version~2.4), which use a Metropolis sampler and a Gibbs sampler, respectively. We use the package-default priors for all Bayesian baselines. Note that \texttt{lqmm} does not support crossed random effects. Unless stated otherwise, we use 3,000 MCMC samples and 1,000 burn-in iterations. For the fully Bayesian methods, we use posterior means as point predictions and hyperparameter estimates. To keep the computational complexity within a reasonable range across all methods, we impose a wall-clock runtime limit of one hour for estimating a model and generating predictions on the corresponding test set; methods exceeding this limit are terminated and treated as not converged. 

For Gaussian process models, we compare against an inducing points variational GP (SVGP) approximation \citep{hensman2013gaussian, hensman2015scalable} in \texttt{GPyTorch} \citep{gardner2021gpytorchblackboxmatrixmatrixgaussian} (version~1.14) with $1{,}000$  inducing points. The inducing points are optimized jointly with the variational and kernel hyperparameters. We also consider the Vecchia approximation variational inference (VIVA) method \citep{cao2023variational}. Finally, we use a spline-based quantile regression method implemented in the \texttt{qgam} package \citep{Fasiolo03072021} (version~1.3.4) with default optimization settings in two configurations: one with only univariate smooth terms and 40 basis functions, and one where we additionally include interaction effects and 45 basis functions. For $n > 1{,}000$, we employ a Vecchia approximation with 20 nearest neighbors for the Fisher-Laplace and TKC approximations to maintain computational feasibility. 

For the Fisher-Laplace and TKC Laplace approximations, we use the \texttt{GPBoost} library version 1.6.7. The TKC Laplace approximation is performed with a minimum likelihood-drop threshold set to $0.1$ for mixed-effects models and $10$ for Gaussian process models. As discussed in Section~\ref{sec:cc}, this threshold is used to avoid numerical issues associated with estimating curvature in an excessively small neighborhood around the mode. As shown in Appendix~\ref{app:sensitivity_tkc}, the results are robust to this choice. All calculations were carried out using an Intel Xeon E3-1284L v4 processor (2.90\,GHz) and 31 GB of RAM. Code for reproducing all experiments is available at \if1\blind{\url{https://github.com/AndreaThomNava/LaGP}}\else{BLINDED PLACEHOLDER}\fi.

\subsection{Grouped Random Effects} 
\label{sec:single_level_re}
We generate data from a single-level grouped random effects model with $m=100$ groups and $n_j \in \{10, 100, 500\}$ occurrences per random effect, and the random effects follow a Gaussian distribution with variance $\sigma^2_u = 1$. We consider two noise scenarios: observations generated from the asymmetric Laplace distribution and from a Gaussian distribution, with the latter representing model misspecification (see additionally Appendix~\ref{app:student_t} for a Student-t noise scenario). The noise variance is chosen to maintain a signal-to-noise ratio of 5 across both scenarios. 

\begin{table}[ht!]
\centering
\begin{tabular}{lccccccc}
\toprule
\multicolumn{3}{c}{} & \multicolumn{5}{c}{RMSE} \\
\cmidrule(lr){4-8}
Noise & m & $n_j$ & \multicolumn{1}{c}{TKC} & \multicolumn{1}{c}{Fisher} & \multicolumn{1}{c}{BayesQR} & \multicolumn{1}{c}{BRMS} & \multicolumn{1}{c}{LQMM} \\
\midrule
ALD & 100 & 10 & 0.12$_\pm{\text{\tiny 0.005}}$ & 0.12$_\pm{\text{\tiny 0.005}}$ & \textbf{0.11$_\pm{\text{\tiny 0.0037}}$} & 0.34$_\pm{\text{\tiny 0.0048}}$ & 0.17$_\pm{\text{\tiny 0.004}}$ \\
ALD & 100 & 100 & \textbf{0.028$_\pm{\text{\tiny 7.0e-04}}$} & \textbf{0.028$_\pm{\text{\tiny 7.2e-04}}$} & \textbf{0.027$_\pm{\text{\tiny 6.9e-04}}$} & 0.33$_\pm{\text{\tiny 8.5e-04}}$ & 0.49$_\pm{\text{\tiny 0.18}}$ \\
ALD & 100 & 500 & \textbf{0.012$_\pm{\text{\tiny 2.9e-04}}$} & \textbf{0.012$_\pm{\text{\tiny 2.9e-04}}$} & \textbf{0.012$_\pm{\text{\tiny 3.0e-04}}$} & --- & 0.62$_\pm{\text{\tiny 0.2}}$ \\
N & 100 & 10 & 0.22$_\pm{\text{\tiny 0.0065}}$ & 0.22$_\pm{\text{\tiny 0.0062}}$ & \textbf{0.21$_\pm{\text{\tiny 0.0081}}$} & 0.48$_\pm{\text{\tiny 0.0069}}$ & \textbf{0.2$_\pm{\text{\tiny 0.0067}}$} \\
N & 100 & 100 & \textbf{0.072$_\pm{\text{\tiny 0.0018}}$} & \textbf{0.072$_\pm{\text{\tiny 0.0017}}$} & \textbf{0.072$_\pm{\text{\tiny 0.002}}$} & 0.47$_\pm{\text{\tiny 0.0014}}$ & 0.37$_\pm{\text{\tiny 0.14}}$ \\
N & 100 & 500 & \textbf{0.032$_\pm{\text{\tiny 0.0011}}$} & \textbf{0.032$_\pm{\text{\tiny 0.0011}}$} & \textbf{0.032$_\pm{\text{\tiny 0.0011}}$} & --- & 0.58$_\pm{\text{\tiny 0.16}}$ \\
\bottomrule
\end{tabular}
\caption{Single-level grouped random effects: Average quantile test RMSE $\pm$ standard error.
Boldface indicates methods whose RMSE is within two standard errors of the minimum.
Dashes (---) indicate non-convergence or computational failure.}
\label{tab:one_grouped_sim}
\end{table}

The accuracy and runtime results are presented in Tables~\ref{tab:one_grouped_sim} and~\ref{tab:time_groupsim}. Overall, the proposed TKC Laplace and Fisher-Laplace approximations and \texttt{BayesQR} are the most accurate methods. While these three methods yield comparable accuracy, the novel Laplace approximations have computation times that are several orders of magnitude lower than \texttt{BayesQR}. The fully Bayesian method in \texttt{brms} is less accurate and, for large group sizes, often fails to finish within the imposed time limit. \texttt{lqmm} is overall also less accurate and exhibits numerical instability for larger group sizes, but it yields high accuracy for very small groups in the misspecified setting. In Figure~\ref{fig:hypers_one_group}, we report estimated hyperparameters for the misspecified Gaussian setting (see Appendix~\ref{app:single_hyper} for the complete results). Note that \texttt{BayesQR} does not allow for estimating variance components. We find that our proposed Laplace approximations achieve the highest accuracy across all sample sizes, even under misspecification. 
\begin{table}[ht!]
\centering
\begin{tabular}{lccccccc}
\toprule
\multicolumn{3}{c}{} & \multicolumn{5}{c}{Runtime} \\
\cmidrule(lr){4-8}
Noise & m & $n_j$ & \multicolumn{1}{c}{TKC} & \multicolumn{1}{c}{Fisher} & \multicolumn{1}{c}{BayesQR} & \multicolumn{1}{c}{BRMS} & \multicolumn{1}{c}{LQMM} \\
\midrule
ALD & 100 & 10 & 0.19$_\pm{\text{\tiny 0.019}}$ & 0.15$_\pm{\text{\tiny 0.012}}$ & 48.5$_\pm{\text{\tiny 0.34}}$ & 479$_\pm{\text{\tiny 10}}$ & \textbf{0.043$_\pm{\text{\tiny 0.0028}}$} \\
ALD & 100 & 100 & 0.5$_\pm{\text{\tiny 0.037}}$ & \textbf{0.23$_\pm{\text{\tiny 0.028}}$} & 463$_\pm{\text{\tiny 7.2}}$ & 3163$_\pm{\text{\tiny 89.7}}$ & 0.43$_\pm{\text{\tiny 0.081}}$ \\
ALD & 100 & 500 & 1.8$_\pm{\text{\tiny 0.12}}$ & \textbf{0.74$_\pm{\text{\tiny 0.071}}$} & 2618$_\pm{\text{\tiny 23.9}}$ & --- & 3$_\pm{\text{\tiny 0.87}}$ \\
N & 100 & 10 & 0.18$_\pm{\text{\tiny 0.02}}$ & \textbf{0.097$_\pm{\text{\tiny 0.011}}$} & 49.8$_\pm{\text{\tiny 0.46}}$ & 452$_\pm{\text{\tiny 10.6}}$ & 0.15$_\pm{\text{\tiny 0.11}}$ \\
N & 100 & 100 & 0.48$_\pm{\text{\tiny 0.033}}$ & \textbf{0.26$_\pm{\text{\tiny 0.033}}$} & 464$_\pm{\text{\tiny 4.5}}$ & 2787$_\pm{\text{\tiny 116}}$ & \textbf{0.3$_\pm{\text{\tiny 0.027}}$} \\
N & 100 & 500 & 1.4$_\pm{\text{\tiny 0.14}}$ & \textbf{0.88$_\pm{\text{\tiny 0.16}}$} & 2362$_\pm{\text{\tiny 5.8}}$ & --- & 2.1$_\pm{\text{\tiny 0.27}}$ \\
\bottomrule
\end{tabular}
\caption{Single-level grouped random effects results for simulated data: Average runtime (in seconds) for estimation and prediction }
\label{tab:time_groupsim}
\end{table}

\begin{figure}[ht!]
    \centering
    \includegraphics[width=0.8\linewidth]{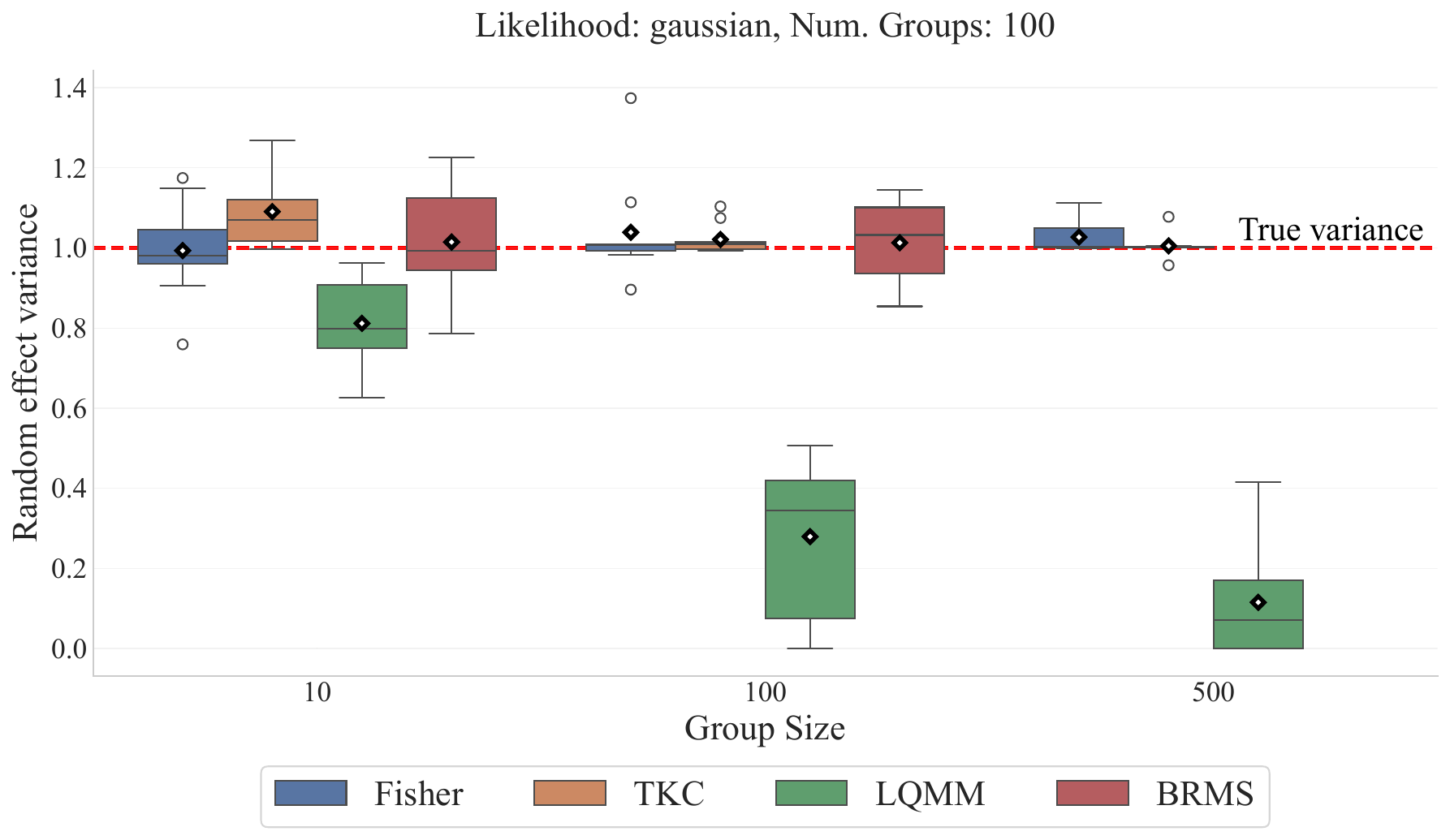}
    \caption{Single-level grouped random effects: Estimates of the random-effects variance $\sigma_u^2$ for varying group sizes $n_j$ in the Gaussian noise setting.}
    \label{fig:hypers_one_group}
\end{figure}


In Appendix \ref{app_crossed}, we also consider a two-factor crossed random effects model. In brief, we find that the TKC and Fisher-Laplace approximations match \texttt{bayesQR} in accuracy whenever \texttt{bayesQR} does not fail to converge, while having lower runtimes by orders of magnitude. And \texttt{brms} did not finish within the one-hour time limit in any crossed-effects configuration.

In Appendix \ref{marg_ll_accuracy}, we additionally analyze the accuracy of marginal likelihood approximations. In misspecified settings, the Fisher-Laplace approximation is less accurate than the TKC Laplace approximation, which converges to the correct quantity.

\subsection{Gaussian Process} 
We sample data from Gaussian processes using a Matérn 1.5 kernel with variance $\sigma^2 = 1$, input dimensions $d \in \{2, 5\}$ with length scale $\ell = 0.25$ and $ \ell = 0.25 \sqrt{5/2}$ for $d=2$ and $d=5$, respectively, and sample sizes $n \in \{1{,}000, 10{,}000\}$. We consider the following noise models: iid Gaussian, heteroskedastic Gaussian, and Student-t with 2 degrees of freedom. In the heteroskedastic Gaussian setting, the log-standard deviation is generated by an independent second Gaussian process with the same kernel specifications as the signal process. The quantile prediction accuracy results are summarized in Table~\ref{tab:gp_results} for the heteroskedastic Gaussian noise. The full results and the estimation accuracy and runtime results are reported in Appendix~\ref{app:gp_sim_runtime}. Overall, the proposed Fisher-Laplace (FL) and TKC Laplace approximations and the variational approximation (VI) yield the most accurate results. Interestingly, the TKC approximation, whose asymptotic justification depends on the homoscedastic assumption \hyperlink{ass:cc}{(CC)}, also yields very accurate results in the misspecified heteroscedastic setting. Both \texttt{qgam} baselines give consistently higher RMSEs; however, they are also the fastest methods considered.

\begin{table}[ht!]
\centering
\begin{tabular}{lllcccccc}
\toprule
\multicolumn{3}{c}{} & \multicolumn{6}{c}{RMSE} \\
\cmidrule(lr){4-9}
Noise & $n$ & d & TKC & FL & VI & VIVA & QGAM & QGAM Int. \\
\midrule
HetN & 1{,}000 & 2 & \textbf{0.24$_\pm{\text{\tiny 0.025}}$} & \textbf{0.25$_\pm{\text{\tiny 0.026}}$} & \textbf{0.24$_\pm{\text{\tiny 0.025}}$} & 0.52$_\pm{\text{\tiny 0.073}}$ & 0.75$_\pm{\text{\tiny 0.069}}$ & 0.36$_\pm{\text{\tiny 0.019}}$ \\
HetN & 1{,}000 & 5 & \textbf{0.66$_\pm{\text{\tiny 0.037}}$} & \textbf{0.65$_\pm{\text{\tiny 0.032}}$} & \textbf{0.63$_\pm{\text{\tiny 0.032}}$} & 0.71$_\pm{\text{\tiny 0.04}}$ & 0.99$_\pm{\text{\tiny 0.03}}$ & 0.99$_\pm{\text{\tiny 0.029}}$ \\
HetN & 10{,}000 & 2 & \textbf{0.28$_\pm{\text{\tiny 0.058}}$} & \textbf{0.28$_\pm{\text{\tiny 0.057}}$} & \textbf{0.31$_\pm{\text{\tiny 0.066}}$} & 0.52$_\pm{\text{\tiny 0.16}}$ & 1$_\pm{\text{\tiny 0.12}}$ & 0.75$_\pm{\text{\tiny 0.12}}$ \\
HetN & 10{,}000 & 5 & \textbf{0.81$_\pm{\text{\tiny 0.13}}$} & \textbf{0.79$_\pm{\text{\tiny 0.12}}$} & \textbf{0.85$_\pm{\text{\tiny 0.14}}$} & \textbf{0.83$_\pm{\text{\tiny 0.12}}$} & 1.1$_\pm{\text{\tiny 0.059}}$ & 1.1$_\pm{\text{\tiny 0.065}}$ \\
\bottomrule
\end{tabular}
\caption{GP results for for heteroscedastic Gaussian noise: Average quantile test RMSE $\pm$ standard error.}
\label{tab:gp_results}
\end{table}
\if1\jasa{\vspace{-0.7cm}}\fi

\subsection{Sandwich Correction using the TKC Estimator}\label{sec:empirical_sandwich}
In the following, we analyze the properties of calibrated confidence intervals using the sandwich variance correction described in Section~\ref {sec:calibration} with the TKC approximation given in \eqref{eq:cc}.
We generate data from a single-level grouped random effects model with $m=100$ groups and $n_j = 100$ observations each, under both Gaussian and Student's $t$ noise distributions, such that there is misspecification. For each random effect $b_j$, we construct $(1-\alpha)100\%$ confidence intervals centered at the posterior mode $\hat{b}_j$ with the sizes determined by the sandwich variance:
\begin{equation}
\text{CI}_{1-\alpha}(\hat{b_j}) = \hat{b}_j \pm z_{1 - \alpha/2} \cdot \text{SE}_{\text{sandwich}}(\hat{b}_j),
\end{equation}
where $z_{1 - \alpha/2}$ is the $1 - \alpha/2$ quantile of the standard normal distribution and $\text{SE}_{\text{sandwich}}(\hat{b}_j)$ denotes the standard error obtained from the sandwich variance estimator. For quantile regression at level $\tau$ and $\boldsymbol{\beta}=0$ (no fixed effect), the sandwich variance is given by
\begin{equation}
\text{Var}_{\text{sandwich}}(\hat{b}_j) = \frac{\tau(1-\tau)}{n_j\hat{C}_{\Delta b}(\hat b_j)^2},
\end{equation}
where $\hat{C}_{\Delta b}(\hat b_j)$ is our triangular kernel curvature estimate of the curvature.
We evaluate the empirical coverage, defined as the proportion of confidence intervals containing the true random effect:
\begin{equation}
\text{Coverage} = \frac{1}{m_j}\sum_{j=1}^{m_j} \mathbbm{1}\{b_j^{(\text{true})} \in \text{CI}_{1-\alpha}(\hat{b_j})\},
\end{equation}
and repeat this experiment for $K=10$ independent replications.
As shown in Figure~\ref{fig:coverage_sand}, the sandwich variance correction using the TKC estimate achieves correct nominal coverage across both noise models, whereas the other methods yield miscalibrated confidence intervals.
\begin{figure}[ht!]
    \centering
    \includegraphics[width=1\linewidth]{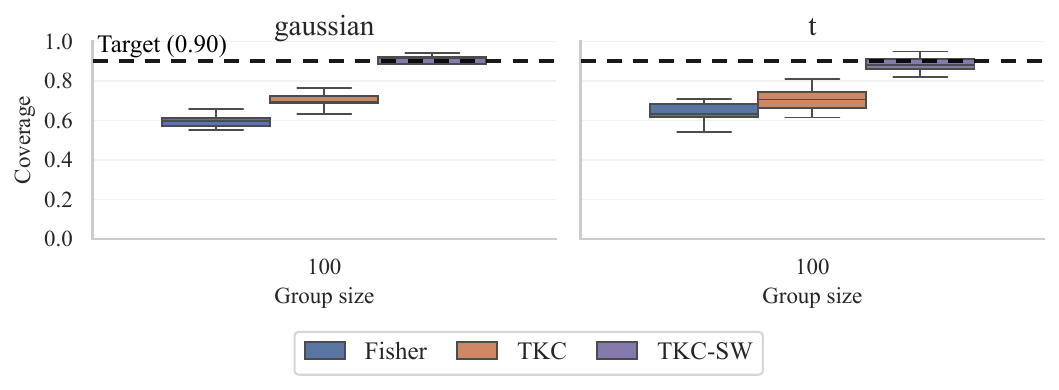}
   \caption{Empirical coverage of $90\%$ Wald-type confidence intervals for the group random effects $b_j$, using the TKC-based sandwich standard errors (TKC--SW) for Gaussian and Student's $t$ noise.}
\label{fig:coverage_sand}
\end{figure}
\if1\jasa{\vspace{-0.7cm}}\fi

\section{Real-World Applications}\label{real_world_sec}
In the following, we evaluate the proposed Fisher and TKC Laplace approximations for grouped random effects and Gaussian process models on several real-world datasets. We compare to the same benchmark methods as in the simulated experiments in Section \ref{sec:simulated_exp} using the same settings.

\subsection{Grouped Data Sets for Random Effects} 
We consider three single-level grouped random effects datasets: the Orthodont ($n = 108$, $m = 27$, $p=2$), Labor ($n = 358$, $m = 83$, $p=2$) and Cars ($n = 97,296$, $m = 15,226$, $p=2$) datasets. These datasets are described in Appendix \ref{desc_data}. They include small and large sample sizes, allowing us to assess both the accuracy and computational scalability of the methods. The comparison of the methods (TKC, Fisher-Laplace, \texttt{BayesQR}, \texttt{lqmm}, \texttt{brms}) is done using 5-fold cross-validation. Table~\ref{tab:grouped_real} reports the quantile loss. For the smaller datasets (Orthodont and Labor), all methods yield comparable accuracy, with \texttt{brms} exhibiting a higher loss on the Labor dataset. On the large dataset (Cars), the proposed TKC and FL approximations are the only methods that run reliably at this scale, whereas the remaining benchmarks fail to converge or do not finish within the computational budget. The runtimes are reported in Appendix~\ref{app:single_real_runtime}.

\begin{table}[htbp]
\centering
\begin{tabular}{lccccc}
\toprule
\multicolumn{1}{c}{} & \multicolumn{5}{c}{Quantile Loss} \\
\cmidrule(lr){2-6}
Dataset & \multicolumn{1}{c}{TKC} & \multicolumn{1}{c}{Fisher} & \multicolumn{1}{c}{BayesQR} & \multicolumn{1}{c}{BRMS} & \multicolumn{1}{c}{LQMM} \\
\midrule
Orthodont & \textbf{0.17$_\pm{\text{\tiny 0.077}}$} & \textbf{0.17$_\pm{\text{\tiny 0.071}}$} & \textbf{0.17$_\pm{\text{\tiny 0.057}}$} & \textbf{0.25$_\pm{\text{\tiny 0.11}}$} & \textbf{0.16$_\pm{\text{\tiny 0.074}}$} \\
Labor & \textbf{0.18$_\pm{\text{\tiny 0.034}}$} & \textbf{0.17$_\pm{\text{\tiny 0.028}}$} & \textbf{0.17$_\pm{\text{\tiny 0.04}}$} & 0.27$_\pm{\text{\tiny 0.051}}$ & \textbf{0.18$_\pm{\text{\tiny 0.035}}$} \\
Cars & \textbf{0.13$_\pm{\text{\tiny 0.0055}}$} & \textbf{0.13$_\pm{\text{\tiny 0.0029}}$} & --- & --- & --- \\
\bottomrule
\end{tabular}
\caption{Real-world data grouped random effects results: Average quantile loss $\pm$ standard error.
Boldface indicates methods within two standard errors of the minimum loss for each dataset.
Dashes (---) indicate non-convergence or failure to finish within the time budget.}
\label{tab:grouped_real}
\end{table}

In Appendix~\ref{app:real_crossed_runtime}, we additionally consider two crossed random effects datasets. We find that the TKC and Fisher-Laplace Laplace approximations yield similar accuracy with the TKC approximation being slower than the Fisher-Laplace approximation. Both \texttt{brms} and \texttt{bayesQR} do not finish within the allocated budget for both datasets.

\subsection{Spatial Datasets for Gaussian Processes}
We consider four spatial datasets with continuous responses and two-dimensional coordinate inputs: the Laegern dataset, two MODIS satellite data sets denoted as Heaton and MODIS, and the Lucas county House dataset. These datasets are described in Appendix \ref{desc_data}. For all datasets, we subsample $n_{train} = 9{,}000$ data points for training and $n_{test} = 1{,}000$ data points for testing, and repeat the procedure 5 times. The following preprocessing is applied to all datasets: the input variables are scaled to the unit interval $[0,1]$ using min-max normalization, and the response variables are standardized to zero mean and unit variance. This standardization is done based on the training data and facilitates comparison of quantile loss values across datasets with different scales. We use a Matérn 1.5 kernel, and adopt a Vecchia approximation using 20 neighbors for the Fisher and TKC Laplace approximations.

Table~\ref{tab:real_data_comparison} reports the test quantile loss for the proposed methods and the same benchmark methods as in the simulated experiments (see Section \ref{sec:simulated_exp}). We observe that the proposed TKC and Fisher-Laplace methods are consistently among the best-performing approaches across all datasets, while spline-based \texttt{qgam} baselines are noticeably less accurate. Moreover, the \texttt{qgam} variants are the fastest methods in terms of runtime (see Table~\ref{table:gp_runtime_real} in Appendix \ref{app:real_GP_runtime}), and TKC and Fisher-Laplace are faster than both VI and VIVA.
\begin{table}[ht!]
\centering
\begin{tabular}{lcccccc}
\toprule
\multicolumn{1}{c}{} & \multicolumn{6}{c}{Quantile Loss} \\
\cmidrule(lr){2-7}
Dataset & TKC & FL & VI & VIVA & QGAM & QGAM Int. \\
\midrule
Heaton & \textbf{0.084$_\pm$\text{\tiny 0.0024}} & \textbf{0.084$_\pm$\text{\tiny 7.7e-04}} & 0.094$_\pm$\text{\tiny 0.0022} & \textbf{0.084$_\pm$\text{\tiny 0.0023}} & 0.14$_\pm$\text{\tiny 0.0045} & 0.11$_\pm$\text{\tiny 0.0028} \\
House & \textbf{0.11$_\pm$\text{\tiny 0.007}} & \textbf{0.11$_\pm$\text{\tiny 0.0068}} & 0.13$_\pm$\text{\tiny 0.0059} & 0.15$_\pm$\text{\tiny 0.0055} & 0.18$_\pm$\text{\tiny 0.0036} & 0.16$_\pm$\text{\tiny 0.004} \\
Laegern & \textbf{0.18$_\pm$\text{\tiny 0.019}} & \textbf{0.17$_\pm$\text{\tiny 0.0065}} & 0.19$_\pm$\text{\tiny 0.0041} & \textbf{0.18$_\pm$\text{\tiny 0.0083}} & 0.23$_\pm$\text{\tiny 0.0042} & 0.22$_\pm$\text{\tiny 0.0067} \\
MODIS & \textbf{0.079$_\pm$\text{\tiny 0.0027}} & \textbf{0.078$_\pm$\text{\tiny 0.0031}} & 0.089$_\pm$\text{\tiny 0.0044} & \textbf{0.08$_\pm$\text{\tiny 0.0028}} & 0.17$_\pm$\text{\tiny 0.0068} & 0.12$_\pm$\text{\tiny 0.0044} \\
\bottomrule
\end{tabular}
\caption{Real-world data Gaussian process results: average quantile loss $\pm$ standard error.
Boldface indicates methods within two standard errors of the minimum loss for each dataset.}
\label{tab:real_data_comparison}
\end{table}
\if1\jasa{\vspace{-0.7cm}}\fi

\section{Conclusion and Future Work}
\label{Discussion}
In this work, we develop novel Laplace approximations for quantile regression with latent Gaussian models.  A key contribution of this work is to make explicit the connection between Bayesian asymptotic results under misspecification and classical M-estimation theory, and to show how this connection can be exploited to justify practical Laplace approximations even when the likelihood is not twice differentiable and potentially misspecified. This link has not previously been articulated in the context of Bayesian quantile regression. We provide theoretical justifications for the proposed Laplace approximations and demonstrate their competitive performance against state-of-the-art methods through comprehensive simulations and real data applications. The resulting methodology is computationally robust, deterministic, and scalable to large datasets. 

Despite these contributions, several limitations remain. A fundamental challenge for quantile regression \citep{koenker2017handbook} is that performance may deteriorate for extreme quantiles ($\tau \rightarrow 0$ or $\tau \rightarrow 1$) due to data sparsity. Future research directions include the use of higher-order kernels in curvature estimation, relaxing the constant curvature assumption by making the curvature estimation adaptive, the development of more refined posterior calibration techniques, and applying the same curvature-based Laplace principle to other non-smooth generalized Bayesian models including the $\varepsilon$-insensitive loss in support vector regression and the interval score for interval regression.

\bibliographystyle{abbrvnat}
\bibliography{bib_name}

\begin{thebibliography}{45}
\providecommand{\natexlab}[1]{#1}
\providecommand{\url}[1]{\texttt{#1}}
\expandafter\ifx\csname urlstyle\endcsname\relax
  \providecommand{\doi}[1]{doi: #1}\else
  \providecommand{\doi}{doi: \begingroup \urlstyle{rm}\Url}\fi

\bibitem[Benoit and Van~den Poel(2017)]{benoit2017bayesqr}
D.~F. Benoit and D.~Van~den Poel.
\newblock bayesqr: A bayesian approach to quantile regression.
\newblock \emph{Journal of Statistical Software}, 76:\penalty0 1--32, 2017.

\bibitem[Bissiri et~al.(2016)Bissiri, Holmes, and Walker]{bissiri2016general}
P.~G. Bissiri, C.~C. Holmes, and S.~G. Walker.
\newblock A general framework for updating belief distributions.
\newblock \emph{Journal of the Royal Statistical Society Series B: Statistical Methodology}, 78\penalty0 (5):\penalty0 1103--1130, 2016.

\bibitem[Bürkner(2017)]{JSSv080i01}
P.-C. Bürkner.
\newblock brms: An r package for bayesian multilevel models using stan.
\newblock \emph{Journal of Statistical Software}, 80\penalty0 (1):\penalty0 1–28, 2017.

\bibitem[Cao et~al.(2023)Cao, Kang, Jimenez, Sang, Schaefer, and Katzfuss]{cao2023variational}
J.~Cao, M.~Kang, F.~Jimenez, H.~Sang, F.~T. Schaefer, and M.~Katzfuss.
\newblock Variational sparse inverse cholesky approximation for latent gaussian processes via double kullback-leibler minimization.
\newblock In \emph{International Conference on Machine Learning}, pages 3559--3576. PMLR, 2023.

\bibitem[Chernozhukov and Hong(2003)]{chernozhukov2003mcmc}
V.~Chernozhukov and H.~Hong.
\newblock An mcmc approach to classical estimation.
\newblock \emph{Journal of econometrics}, 115\penalty0 (2):\penalty0 293--346, 2003.

\bibitem[Datta et~al.(2016)Datta, Banerjee, Finley, and Gelfand]{datta2016hierarchical}
A.~Datta, S.~Banerjee, A.~O. Finley, and A.~E. Gelfand.
\newblock Hierarchical nearest-neighbor {G}aussian process models for large geostatistical datasets.
\newblock \emph{Journal of the American Statistical Association}, 111\penalty0 (514):\penalty0 800--812, 2016.

\bibitem[Ehm et~al.(2016)Ehm, Gneiting, Jordan, and Kr{\"u}ger]{ehm2016quantiles}
W.~Ehm, T.~Gneiting, A.~Jordan, and F.~Kr{\"u}ger.
\newblock Of quantiles and expectiles: consistent scoring functions, choquet representations and forecast rankings.
\newblock \emph{Journal of the Royal Statistical Society Series B: Statistical Methodology}, 78\penalty0 (3):\penalty0 505--562, 2016.

\bibitem[Fasiolo et~al.(2021)Fasiolo, Wood, Zaffran, Nedellec, and and]{Fasiolo03072021}
M.~Fasiolo, S.~N. Wood, M.~Zaffran, R.~Nedellec, and Y.~G. and.
\newblock Fast calibrated additive quantile regression.
\newblock \emph{Journal of the American Statistical Association}, 116\penalty0 (535):\penalty0 1402--1412, 2021.

\bibitem[Fong and Holmes(2020)]{fong2020marginal}
E.~Fong and C.~C. Holmes.
\newblock On the marginal likelihood and cross-validation.
\newblock \emph{Biometrika}, 107\penalty0 (2):\penalty0 489--496, 2020.

\bibitem[Gardner et~al.(2021)Gardner, Pleiss, Bindel, Weinberger, and Wilson]{gardner2021gpytorchblackboxmatrixmatrixgaussian}
J.~R. Gardner, G.~Pleiss, D.~Bindel, K.~Q. Weinberger, and A.~G. Wilson.
\newblock Gpytorch: Blackbox matrix-matrix gaussian process inference with gpu acceleration, 2021.

\bibitem[Geraci(2017)]{geraci2017nonlinear}
M.~Geraci.
\newblock Nonlinear quantile mixed models.
\newblock \emph{arXiv preprint arXiv:1712.09981}, 2017.

\bibitem[Geraci and Bottai(2014)]{geraci2014linear}
M.~Geraci and M.~Bottai.
\newblock Linear quantile mixed models.
\newblock \emph{Statistics and computing}, 24:\penalty0 461--479, 2014.

\bibitem[Guinness(2021)]{guinness2019gaussian}
J.~Guinness.
\newblock {{G}aussian process learning via {F}isher scoring of {V}ecchia’s approximation}.
\newblock \emph{Statistics and Computing}, 31\penalty0 (3):\penalty0 1--8, 2021.

\bibitem[Gyger et~al.(2025)Gyger, Furrer, and Sigrist]{gyger2025vecchia}
T.~Gyger, R.~Furrer, and F.~Sigrist.
\newblock {Vecchia-Inducing-Points Full-Scale Approximations for Gaussian Processes}.
\newblock \emph{arXiv preprint arXiv:2507.05064}, 2025.

\bibitem[Hartmann and Vanhatalo(2019)]{hartmann2019laplace}
M.~Hartmann and J.~Vanhatalo.
\newblock Laplace approximation and natural gradient for gaussian process regression with heteroscedastic student-t model.
\newblock \emph{Statistics and Computing}, 29\penalty0 (4):\penalty0 753--773, 2019.

\bibitem[He et~al.(2023)He, Pan, Tan, and Zhou]{he2023smoothed}
X.~He, X.~Pan, K.~M. Tan, and W.-X. Zhou.
\newblock Smoothed quantile regression with large-scale inference.
\newblock \emph{Journal of Econometrics}, 232\penalty0 (2):\penalty0 367--388, 2023.

\bibitem[Hensman et~al.(2013)Hensman, Fusi, and Lawrence]{hensman2013gaussian}
J.~Hensman, N.~Fusi, and N.~D. Lawrence.
\newblock {G}aussian processes for big data.
\newblock In \emph{Proceedings of the Twenty-Ninth Conference on Uncertainty in Artificial Intelligence}, pages 282--290, 2013.

\bibitem[Hensman et~al.(2015)Hensman, Matthews, and Ghahramani]{hensman2015scalable}
J.~Hensman, A.~Matthews, and Z.~Ghahramani.
\newblock Scalable variational {G}aussian process classification.
\newblock In \emph{Artificial intelligence and statistics}, pages 351--360. PMLR, 2015.

\bibitem[Kato(2012)]{Kato2012}
K.~Kato.
\newblock Asymptotic normality of powell's kernel estimator.
\newblock \emph{Annals of the Institute of Statistical Mathematics}, 64\penalty0 (2):\penalty0 255--273, 2012.

\bibitem[Kleijn and van~der Vaart(2012)]{Miss-BvM}
B.~Kleijn and A.~van~der Vaart.
\newblock {The Bernstein-Von-Mises theorem under misspecification}.
\newblock \emph{Electronic Journal of Statistics}, 6\penalty0 (none):\penalty0 354 -- 381, 2012.

\bibitem[Koenker(2005)]{koenker2005quantile}
R.~Koenker.
\newblock \emph{Quantile regression}, volume~38.
\newblock Cambridge university press, 2005.

\bibitem[Koenker and Bassett(1978)]{1aa6b708-cbcf-320f-b3f3-298154ee2aeb}
R.~Koenker and G.~Bassett.
\newblock Regression quantiles.
\newblock \emph{Econometrica}, 46\penalty0 (1):\penalty0 33--50, 1978.

\bibitem[Koenker et~al.(2017)Koenker, Chernozhukov, He, and Peng]{koenker2017handbook}
R.~Koenker, V.~Chernozhukov, X.~He, and L.~Peng, editors.
\newblock \emph{Handbook of Quantile Regression}.
\newblock Chapman and Hall/CRC, 2017.

\bibitem[Kozumi and Kobayashi(2011)]{kozumi2011gibbs}
H.~Kozumi and G.~Kobayashi.
\newblock Gibbs sampling methods for bayesian quantile regression.
\newblock \emph{Journal of statistical computation and simulation}, 81\penalty0 (11):\penalty0 1565--1578, 2011.

\bibitem[K{\"u}ndig and Sigrist(2025)]{kundig2025scalable}
P.~K{\"u}ndig and F.~Sigrist.
\newblock {Scalable Computations for Generalized Mixed Effects Models with Crossed Random Effects Using Krylov Subspace Methods}.
\newblock \emph{arXiv preprint arXiv:2505.09552}, 2025.

\bibitem[Le~Cam and Yang(2000)]{le2000asymptotics}
L.~M. Le~Cam and G.~L. Yang.
\newblock \emph{Asymptotics in statistics: some basic concepts}.
\newblock Springer Science \& Business Media, 2000.

\bibitem[Liu et~al.(2025)Liu, Li, and Pang]{liu2025bayesian}
B.~Liu, K.~Li, and T.~Pang.
\newblock Bayesian smoothed quantile regression.
\newblock \emph{arXiv preprint arXiv:2508.01738}, 2025.

\bibitem[McCulloch and Searle(2004)]{mcculloch2004generalized}
C.~E. McCulloch and S.~R. Searle.
\newblock \emph{Generalized, linear, and mixed models}.
\newblock John Wiley \& Sons, 2004.

\bibitem[Nickisch et~al.(2008)Nickisch, Rasmussen, et~al.]{nickisch2008approximations}
H.~Nickisch, C.~E. Rasmussen, et~al.
\newblock Approximations for binary gaussian process classification.
\newblock \emph{Journal of Machine Learning Research}, 9\penalty0 (10):\penalty0 2035--2078, 2008.

\bibitem[Powell(1984)]{POWELL1984303}
J.~L. Powell.
\newblock Least absolute deviations estimation for the censored regression model.
\newblock \emph{Journal of Econometrics}, 25\penalty0 (3):\penalty0 303--325, 1984.

\bibitem[Rambelli and Sigrist(2026)]{rambelli2025accuracy}
F.~Rambelli and F.~Sigrist.
\newblock An accuracy-runtime trade-off comparison of scalable {G}aussian process approximations for spatial data.
\newblock \emph{Journal of Agricultural, Biological and Environmental Statistics (in press)}, 2026.

\bibitem[Romano et~al.(2019)Romano, Patterson, and Candes]{romano2019conformalized}
Y.~Romano, E.~Patterson, and E.~Candes.
\newblock Conformalized quantile regression.
\newblock \emph{Advances in neural information processing systems}, 32, 2019.

\bibitem[Rossellini et~al.(2024)Rossellini, Barber, and Willett]{rossellini2024integrating}
R.~Rossellini, R.~F. Barber, and R.~Willett.
\newblock Integrating uncertainty awareness into conformalized quantile regression.
\newblock In \emph{International Conference on Artificial Intelligence and Statistics}, pages 1540--1548. PMLR, 2024.

\bibitem[Schervish et~al.(2012)Schervish, Kadane, and Seidenfeld]{Schervish2018}
M.~J. Schervish, J.~B. Kadane, and T.~Seidenfeld.
\newblock Characterization of proper and strictly proper scoring rules for quantiles.
\newblock \emph{Preprint, Carnegie Mellon University, March}, 18, 2012.

\bibitem[Sigrist(2022)]{sigrist2022gaussian}
F.~Sigrist.
\newblock Gaussian process boosting.
\newblock \emph{Journal of Machine Learning Research}, 23\penalty0 (232):\penalty0 1--46, 2022.

\bibitem[Sigrist(2023)]{sigrist2022latent}
F.~Sigrist.
\newblock {Latent Gaussian model boosting}.
\newblock \emph{IEEE Transactions on Pattern Analysis and Machine Intelligence}, 45\penalty0 (2):\penalty0 1894--1905, 2023.

\bibitem[Simchoni and Rosset(2023)]{simchoni2023integrating}
G.~Simchoni and S.~Rosset.
\newblock Integrating random effects in deep neural networks.
\newblock \emph{Journal of Machine Learning Research}, 24\penalty0 (156):\penalty0 1--57, 2023.

\bibitem[Sriram et~al.(2013)Sriram, Ramamoorthi, and Ghosh]{10.1214/13-BA817}
K.~Sriram, R.~Ramamoorthi, and P.~Ghosh.
\newblock {Posterior Consistency of Bayesian Quantile Regression Based on the Misspecified Asymmetric Laplace Density}.
\newblock \emph{Bayesian Analysis}, 8\penalty0 (2):\penalty0 479 -- 504, 2013.

\bibitem[Syring and Martin(2018)]{10.1093/biomet/asy054}
N.~Syring and R.~Martin.
\newblock Calibrating general posterior credible regions.
\newblock \emph{Biometrika}, 106\penalty0 (2):\penalty0 479--486, 12 2018.

\bibitem[Tierney and Kadane(1986)]{tierney1986accurate}
L.~Tierney and J.~B. Kadane.
\newblock Accurate approximations for posterior moments and marginal densities.
\newblock \emph{Journal of the american statistical association}, 81\penalty0 (393):\penalty0 82--86, 1986.

\bibitem[Van~der Vaart(2000)]{van2000asymptotic}
A.~W. Van~der Vaart.
\newblock \emph{Asymptotic statistics}, volume~3.
\newblock Cambridge university press, 2000.

\bibitem[Vecchia(1988)]{vecchia1988estimation}
A.~V. Vecchia.
\newblock Estimation and model identification for continuous spatial processes.
\newblock \emph{Journal of the Royal Statistical Society Series B: Statistical Methodology}, 50\penalty0 (2):\penalty0 297--312, 1988.

\bibitem[Williams and Rasmussen(2006)]{williams2006gaussian}
C.~K. Williams and C.~E. Rasmussen.
\newblock \emph{{Gaussian processes for machine learning}}.
\newblock MIT Press Cambridge, 2006.

\bibitem[Yu and Moyeed(2001)]{YU2001437}
K.~Yu and R.~A. Moyeed.
\newblock Bayesian quantile regression.
\newblock \emph{Statistics \& Probability Letters}, 54\penalty0 (4):\penalty0 437--447, 2001.

\bibitem[Yue and Rue(2011)]{yue2011bayesian}
Y.~R. Yue and H.~Rue.
\newblock Bayesian inference for additive mixed quantile regression models.
\newblock \emph{Computational Statistics \& Data Analysis}, 55\penalty0 (1):\penalty0 84--96, 2011.

\end{thebibliography}


\begin{thebibliography}{20}
\providecommand{\natexlab}[1]{#1}
\providecommand{\url}[1]{\texttt{#1}}
\expandafter\ifx\csname urlstyle\endcsname\relax
  \providecommand{\doi}[1]{doi: #1}\else
  \providecommand{\doi}{doi: \begingroup \urlstyle{rm}\Url}\fi

\bibitem[Bivand et~al.(2025)Bivand, Nowosad, and Lovelace]{spdata}
R.~Bivand, J.~Nowosad, and R.~Lovelace.
\newblock \emph{spData: Datasets for Spatial Analysis}, 2025.
\newblock R package version 2.3.4.

\bibitem[Chernozhukov and Hong(2003)]{chernozhukov2003mcmc}
V.~Chernozhukov and H.~Hong.
\newblock An mcmc approach to classical estimation.
\newblock \emph{Journal of econometrics}, 115\penalty0 (2):\penalty0 293--346, 2003.

\bibitem[Davis(1991)]{davis1991semi}
C.~S. Davis.
\newblock Semi-parametric and non-parametric methods for the analysis of repeated measurements with applications to clinical trials.
\newblock \emph{Statistics in medicine}, 10\penalty0 (12):\penalty0 1959--1980, 1991.

\bibitem[Gyger et~al.(2026)Gyger, Furrer, and Sigrist]{gyger2025iterativemethodsfullscalegaussian}
T.~Gyger, R.~Furrer, and F.~Sigrist.
\newblock {Iterative Methods for Full-Scale Gaussian Process Approximations for Large Spatial Data}.
\newblock \emph{SIAM/ASA Journal on Uncertainty Quantification}, 14\penalty0 (1):\penalty0 142--167, 2026.

\bibitem[Harper and Konstan(2015)]{harper2015movielens}
F.~M. Harper and J.~A. Konstan.
\newblock The movielens datasets: History and context.
\newblock \emph{Acm transactions on interactive intelligent systems (tiis)}, 5\penalty0 (4):\penalty0 1--19, 2015.

\bibitem[Heaton et~al.(2019)Heaton, Datta, Finley, Furrer, Guinness, Guhaniyogi, Gerber, Gramacy, Hammerling, Katzfuss, et~al.]{heaton2019case}
M.~J. Heaton, A.~Datta, A.~O. Finley, R.~Furrer, J.~Guinness, R.~Guhaniyogi, F.~Gerber, R.~B. Gramacy, D.~Hammerling, M.~Katzfuss, et~al.
\newblock A case study competition among methods for analyzing large spatial data.
\newblock \emph{Journal of agricultural, biological and environmental Statistics}, 24\penalty0 (3):\penalty0 398--425, 2019.

\bibitem[Kleijn and van~der Vaart(2012)]{Miss-BvM}
B.~Kleijn and A.~van~der Vaart.
\newblock {The Bernstein-Von-Mises theorem under misspecification}.
\newblock \emph{Electronic Journal of Statistics}, 6\penalty0 (none):\penalty0 354 -- 381, 2012.

\bibitem[Koenker(2005)]{koenker2005quantile}
R.~Koenker.
\newblock \emph{Quantile regression}, volume~38.
\newblock Cambridge university press, 2005.

\bibitem[Koenker and Bassett(1978)]{1aa6b708-cbcf-320f-b3f3-298154ee2aeb}
R.~Koenker and G.~Bassett.
\newblock Regression quantiles.
\newblock \emph{Econometrica}, 46\penalty0 (1):\penalty0 33--50, 1978.

\bibitem[Le~Cam and Yang(2000)]{le2000asymptotics}
L.~M. Le~Cam and G.~L. Yang.
\newblock \emph{Asymptotics in statistics: some basic concepts}.
\newblock Springer Science \& Business Media, 2000.

\bibitem[Liu and Pierce(1994)]{liu1994note}
Q.~Liu and D.~A. Pierce.
\newblock A note on gauss—hermite quadrature.
\newblock \emph{Biometrika}, 81\penalty0 (3):\penalty0 624--629, 1994.

\bibitem[Miller(2021)]{miller2021asymptotic}
J.~W. Miller.
\newblock Asymptotic normality, concentration, and coverage of generalized posteriors.
\newblock \emph{Journal of Machine Learning Research}, 22\penalty0 (168):\penalty0 1--53, 2021.

\bibitem[Nickisch et~al.(2008)Nickisch, Rasmussen, et~al.]{nickisch2008approximations}
H.~Nickisch, C.~E. Rasmussen, et~al.
\newblock Approximations for binary gaussian process classification.
\newblock \emph{Journal of Machine Learning Research}, 9\penalty0 (10):\penalty0 2035--2078, 2008.

\bibitem[Potthoff and Roy(1964)]{10.1093/biomet/51.3-4.313}
R.~F. Potthoff and S.~N. Roy.
\newblock A generalized multivariate analysis of variance model useful especially for growth curve problems*.
\newblock \emph{Biometrika}, 51\penalty0 (3-4):\penalty0 313--326, 12 1964.

\bibitem[Romano et~al.(2019)Romano, Patterson, and Candes]{romano2019conformalized}
Y.~Romano, E.~Patterson, and E.~Candes.
\newblock Conformalized quantile regression.
\newblock \emph{Advances in neural information processing systems}, 32, 2019.

\bibitem[Rossellini et~al.(2024)Rossellini, Barber, and Willett]{rossellini2024integrating}
R.~Rossellini, R.~F. Barber, and R.~Willett.
\newblock Integrating uncertainty awareness into conformalized quantile regression.
\newblock In \emph{International Conference on Artificial Intelligence and Statistics}, pages 1540--1548. PMLR, 2024.

\bibitem[Schneider et~al.(2017)Schneider, Morsdorf, Schmid, Petchey, Hueni, Schimel, and Schaepman]{schneider2017mapping}
F.~D. Schneider, F.~Morsdorf, B.~Schmid, O.~L. Petchey, A.~Hueni, D.~S. Schimel, and M.~E. Schaepman.
\newblock Mapping functional diversity from remotely sensed morphological and physiological forest traits.
\newblock \emph{Nature communications}, 8\penalty0 (1):\penalty0 1441, 2017.

\bibitem[Simchoni and Rosset(2023)]{simchoni2023integrating}
G.~Simchoni and S.~Rosset.
\newblock Integrating random effects in deep neural networks.
\newblock \emph{Journal of Machine Learning Research}, 24\penalty0 (156):\penalty0 1--57, 2023.

\bibitem[Van~der Vaart(2000)]{van2000asymptotic}
A.~W. Van~der Vaart.
\newblock \emph{Asymptotic statistics}, volume~3.
\newblock Cambridge university press, 2000.

\bibitem[van~der Vaart and Wellner(1996)]{van1996m}
A.~W. van~der Vaart and J.~A. Wellner.
\newblock M-estimators.
\newblock In \emph{Weak Convergence and Empirical Processes: With Applications to Statistics}, pages 284--308. Springer, 1996.

\end{thebibliography}
\clearpage

\appendix

\begin{bibunit}
\setcounter{figure}{0}
\renewcommand{\thefigure}{A\arabic{figure}}
\setcounter{table}{0}
\renewcommand{\thetable}{A\arabic{table}}

\if0\jasa{
\section*{Appendix}\label{app:supp}
}\fi
\section{Proofs}

\subsection{Proofs of Theorems \ref{thm:FL_consistency} and \ref{thm:misspec_laplace}}
\label{appendix:bvm_theorem}

We start by proving Theorem \ref{thm:misspec_laplace} which does not assume a well specified model. The proof of Theorem \ref{thm:FL_consistency} follows as a special case of Theorem \ref{thm:misspec_laplace}. The main difference will be in the tools we use to verify the assumptions of both theorems. For Theorem \ref{thm:misspec_laplace} the main tools stem from the theory of M-estimators \citep{van1996m}, while for Theorem \ref{thm:FL_consistency} we leverage results from the classical LAN theory \citep{le2000asymptotics}. The proof builds on techniques developed by \citet{chernozhukov2003mcmc} and \citet{miller2021asymptotic}, adapted to the present setting.

\begin{proof}[Proof of Theorem~\ref{thm:misspec_laplace}]
We split the proof into 4 steps.

\paragraph{Step 0: Notation.}
Let us rewrite the likelihood $p(\mathbf{y}\mid\mathbf{b}, \psi)$ as $e^{-n f_n(\mathbf{b})}$ where
\[
f_n(\mathbf{b}) := -\frac{1}{n}\sum_{i=1}^n \log p(y_i\mid \mathbf{b},\psi),
\qquad\text{so that}\qquad
-\log p(\mathbf{y}\mid\mathbf{b},\psi)=n f_n(\mathbf{b}).
\]
Furthermore, denote the marginal likelihood by
$
z_n=\int_{\mathbb R^m} e^{-n f_n(\mathbf{b})}\,\pi(\mathbf{b})\,d\mathbf{b}.$
Define now the ``local'' parameter
\[
\mathbf{t} := \mathbf H_{n,\mathbf b^*}^{1/2}(\mathbf b-\hat{\mathbf b}_n),
\]
which zooms into an information-scaled neighborhood of $\hat{\mathbf b}_n$, since
$\mathbf b=\hat{\mathbf b}_n+\mathbf H_{n,\mathbf b^*}^{-1/2}\mathbf{t} \to \hat{\mathbf b}_n$ as $n\to\infty$ whenever $\|\mathbf H_{n,\mathbf b^*}^{-1/2}\|_{\mathrm{op}}\to 0$.

\paragraph{Step 1: Re-centering the LAN expansion at the posterior mode.}
The LAN expansion in Assumption~(i) is centered at the pseudo-true parameter $\mathbf b^*$.
To apply this expansion at the posterior mode, define
\[
\hat{\mathbf{t}}_n := \mathbf H_{n,\mathbf b^*}^{1/2}(\hat{\mathbf b}_n-\mathbf b^*).
\]
By Assumption~(ii),
\[
\hat{\mathbf{t}}_n
=
\mathbf H_{n,\mathbf b^*}^{-1/2}\sum_{i=1}^n \nabla_{\mathbf b}\log p(y_i\mid\mathbf b^*,\psi)
+o_{\mathbb P^*}(1).
\]
By Assumption~(i), the standardized score $\mathbf H_{n,\mathbf b^*}^{-1/2}\sum_{i=1}^n \nabla_{\mathbf b}\log p(y_i\mid\mathbf b^*,\psi)
$ is tight; hence $\hat{\mathbf t}_n=O_{\mathbb P^*}(1)$. Consequently, for any $\eta>0$ there exists a compact set $K_0\subset\mathbb R^m$ such that
\[
\mathbb P^*(\hat{\mathbf t}_n\in K_0)\ge 1-\eta
\]
for all sufficiently large $n$. On this event, and for any fixed compact set $K$, the shifted neighborhood
\[
\left\{\mathbf b^* + \mathbf H_{n,\mathbf b^*}^{-1/2}(\hat{\mathbf t}_n+\mathbf t):\; \mathbf t\in K\right\}
=
\left\{\hat{\mathbf b}_n+\mathbf H_{n,\mathbf b^*}^{-1/2}\mathbf t:\; \mathbf t\in K\right\}
\]
corresponds to local parameters $\hat{\mathbf t}_n+\mathbf t$ ranging over the compact set $K_0+K$. Writing the log-likelihood difference as
\begin{align*}
\ell_n(\hat{\mathbf b}_n + \mathbf H_{n,\mathbf b^*}^{-1/2}\mathbf t) - \ell_n(\hat{\mathbf b}_n)
&=
\big[\ell_n(\mathbf b^* + \mathbf H_{n,\mathbf b^*}^{-1/2}(\hat{\mathbf t}_n+\mathbf t)) - \ell_n(\mathbf b^*)\big]\\
&\quad-
\big[\ell_n(\mathbf b^* + \mathbf H_{n,\mathbf b^*}^{-1/2}\hat{\mathbf t}_n) - \ell_n(\mathbf b^*)\big],
\end{align*}
we may apply the LAN expansion in Assumption~(i) uniformly to the first bracketed term for $\mathbf t\in K$, since $\hat{\mathbf t}_n+\mathbf t\in K_0+K$, and likewise to the second bracketed term since $\hat{\mathbf t}_n\in K_0$. Thus,
\begin{align*}
\ell_n(\hat{\mathbf b}_n + \mathbf H_{n,\mathbf b^*}^{-1/2}\mathbf t) - \ell_n(\hat{\mathbf b}_n)
&=
\Big((\hat{\mathbf t}_n+\mathbf t)^\top \Delta_n -\tfrac12(\hat{\mathbf t}_n+\mathbf t)^\top(\hat{\mathbf t}_n+\mathbf t)\Big)\\
&\quad-
\Big(\hat{\mathbf t}_n^\top \Delta_n -\tfrac12\hat{\mathbf t}_n^\top\hat{\mathbf t}_n\Big)
+o_{\mathbb P^*}(1),
\end{align*}
where
\[
\Delta_n
:=
\mathbf H_{n,\mathbf b^*}^{-1/2}\sum_{i=1}^n \nabla_{\mathbf b}\log p(y_i\mid\mathbf b^*,\psi).
\]
Expanding the quadratic term gives
\[
-\tfrac12 \mathbf{t}^\top \mathbf{t}
+ \mathbf{t}^\top(\Delta_n-\hat{\mathbf{t}}_n)
+o_{\mathbb P^*}(1),
\]
uniformly for $\mathbf{t}\in K$. By Assumption~(ii), $\hat{\mathbf{t}}_n=\Delta_n+o_{\mathbb P^*}(1)$, and therefore $\mathbf{t}^\top(\Delta_n-\hat{\mathbf{t}}_n)=o_{\mathbb P^*}(1)$ uniformly on compact sets. Consequently,
\[
\ell_n(\hat{\mathbf b}_n + \mathbf H_{n,\mathbf b^*}^{-1/2}\mathbf{t}) - \ell_n(\hat{\mathbf b}_n)
=
-\tfrac12 \mathbf{t}^\top \mathbf{t}
+o_{\mathbb P^*}(1),
\]
uniformly for $\mathbf{t}$ ranging over compact sets. Thus, uniformly on compact subsets of the local $\mathbf{t}$-space, the log-likelihood re-centered at the posterior mode is asymptotically quadratic with leading term $-\frac{1}{2}\mathbf{t}^\top\mathbf{t}$.

\paragraph{Step 2: Convergence on compact sets.}
Define the sequence of functions $g_n$ by
\begin{equation}
g_n(\mathbf{t})
:=
\exp\!\Big(-n\big[f_n(\hat{\mathbf b}_n + \mathbf H_{n,\mathbf b^*}^{-1/2}\mathbf{t})-f_n(\hat{\mathbf b}_n)\big]\Big)\,
\pi(\hat{\mathbf b}_n + \mathbf H_{n,\mathbf b^*}^{-1/2}\mathbf{t}).
\end{equation}
Equivalently, $g_n$ is the unnormalized posterior kernel in the local parameterization. Indeed, if
\[
q_n(\mathbf t)
:=
\pi_n(\hat{\mathbf b}_n+\mathbf H_{n,\mathbf b^*}^{-1/2}\mathbf t\mid \mathbf y)\,
\big|\det(\mathbf H_{n,\mathbf b^*}^{-1/2})\big|
\]
denotes the posterior density of $\mathbf t$, then
\begin{equation}
\label{eq:post_kernel}
q_n(\mathbf t)
=
\frac{e^{-n f_n(\hat{\mathbf b}_n)}\,|\det(\mathbf H_{n,\mathbf b^*})|^{-1/2}}{z_n}\,g_n(\mathbf t).
\end{equation}
Define now the corresponding Gaussian approximation
\begin{equation}
g_0^{(n)}(\mathbf{t})
:=
\exp\!\big(-\tfrac12 \mathbf{t}^\top \mathbf{t}\big)\,\pi(\hat{\mathbf b}_n).
\end{equation}
Let $K\subset\mathbb R^m$ be compact, and define the event
\begin{equation}
\label{eq:En_misspec_infoscaled_t}
E_n(\varepsilon)
=
\left\{
\sup_{\mathbf t \in K}
\left|
\ell_n(\hat{\mathbf b}_n + \mathbf H_{n,\mathbf b^*}^{-1/2}\mathbf t)-\ell_n(\hat{\mathbf b}_n)
+\frac12 \mathbf t^\top \mathbf t
\right|
<
\varepsilon
\right\}.
\end{equation}
By Step 1, for every $\varepsilon>0$ we have $\mathbb P^*(E_n(\varepsilon))\to 1$. On $E_n(\varepsilon)$, the exponential factor in $g_n(\mathbf t)$ converges uniformly on $K$ to $
\exp\!\big(-\tfrac12 \mathbf t^\top \mathbf t\big).
$
Moreover, since $\hat{\mathbf b}_n\to \mathbf b^*$ in $\mathbb P^*$-probability and
\[
\sup_{\mathbf t\in K}\|\mathbf H_{n,\mathbf b^*}^{-1/2}\mathbf t\|
\le
\|\mathbf H_{n,\mathbf b^*}^{-1/2}\|_{\mathrm{op}} \sup_{\mathbf t\in K}\|\mathbf t\|
\to 0,
\]
prior continuity implies
\[
\sup_{\mathbf t\in K}
\left|
\pi(\hat{\mathbf b}_n + \mathbf H_{n,\mathbf b^*}^{-1/2}\mathbf t)-\pi(\hat{\mathbf b}_n)
\right|
\xrightarrow{\mathbb P^*} 0.
\]
Therefore,
\[
\sup_{\mathbf t\in K}\bigl|g_n(\mathbf t)-g_0^{(n)}(\mathbf t)\bigr|
\xrightarrow{\mathbb P^*} 0.
\]
It follows that
\[
\int_K g_n(\mathbf t)\,d\mathbf t \;\xrightarrow{\mathbb P^*}\; \int_K g_0^{(n)}(\mathbf t)\,d\mathbf t.
\]
Assume for now that we are allowed to extend the domain from a compact $K$ to $\mathbb R^m$. Then
\begin{align}
\int_{\mathbb R^m} g_0^{(n)}(\mathbf{t})\,d\mathbf{t}
&=
\pi(\hat{\mathbf b}_n)\int_{\mathbb R^m}\exp\!\big(-\tfrac12 \mathbf{t}^\top \mathbf{t}\big)\,d\mathbf{t} \nonumber\\
&=
\pi(\hat{\mathbf b}_n)(2\pi)^{m/2}. \label{eq:gauss_int_infoscaled_t}
\end{align}
On the other hand, by \eqref{eq:post_kernel} and the fact that $\int_{\mathbb{R}^m} q_n(\mathbf t)\, d\mathbf t = 1$, we see that
\begin{equation}
\label{eq:zn_changevar_t}
z_n
=
e^{-n f_n(\hat{\mathbf b}_n)}\,\big|\det(\mathbf H_{n,\mathbf b^*})\big|^{-1/2}\,
\int_{\mathbb R^m} g_n(\mathbf t)\,d\mathbf t.
\end{equation}
Combining \eqref{eq:gauss_int_infoscaled_t}--\eqref{eq:zn_changevar_t} (and extending to $\mathbb R^m$ in Step 3) yields
\[
z_n \sim e^{-n f_n(\hat{\mathbf b}_n)}\,\big|\det(\mathbf H_{n,\mathbf b^*})\big|^{-1/2}
\pi(\hat{\mathbf b}_n)\,(2\pi)^{m/2}
\;=:\; z_n^{LA,0}.
\]
\paragraph{Step 3: Extending to $\mathbb{R}^m$.}
We need to address two concerns: extending the domain for the convergence and for the integration from $K$ to $\mathbb R^m$.
Let
\[
K_n:=\{\mathbf t\in\mathbb R^m:\|\mathbf t\|\le M_n\},
\]
where $M_n\to\infty$ is as in Assumption~(iii). Then $K_n\uparrow\mathbb R^m$, and $K_n^c$ corresponds in the original $\mathbf b$-scale to the set
\[
\left\{\mathbf b:\;\|\mathbf H_{n,\mathbf b^*}^{1/2}(\mathbf b-\hat{\mathbf b}_n)\|\ge M_n\right\}.
\]
Assumption~(iii) enforces exactly a tail decay uniformly on such sets. In particular, on the event
\[
S_n
=
\left\{
\inf_{\|\mathbf H_{n,\mathbf b^*}^{1/2}(\mathbf b-\hat{\mathbf b}_n)\|\ge M_n}
\big[f_n(\mathbf b)-f_n(\hat{\mathbf b}_n)\big]\ge \delta
\right\},
\]
we have $\mathbb P^*(S_n)\to 1$ by Assumption~(iii), and for every $\mathbf t\in K_n^c$,
\[
f_n(\hat{\mathbf b}_n+\mathbf H_{n,\mathbf b^*}^{-1/2}\mathbf t)-f_n(\hat{\mathbf b}_n)\ge \delta.
\]
Therefore, on $S_n$,
\begin{align*}
\int_{K_n^c} g_n(\mathbf{t})\,d\mathbf{t}
&\le
e^{-n\delta}\int_{\mathbb R^m} \pi(\hat{\mathbf b}_n+\mathbf H_{n,\mathbf b^*}^{-1/2}\mathbf{t})\,d\mathbf{t}\\
&=
e^{-n\delta}\,\big|\det(\mathbf H_{n,\mathbf b^*})\big|^{1/2}\int_{\mathbb R^m}\pi(\mathbf b)\,d\mathbf b,
\end{align*}
where we used the same change of variables as in \eqref{eq:zn_changevar_t} in the last equality.
Under the information growth condition $\log\det(\mathbf H_{n,\mathbf b^*})=o(n)$ in Assumption~(iii), the right-hand side vanishes as $n\to\infty$.
Therefore, combining the uniform convergence on $K_n$ from Step~2 with the negligibility of the tail integral over $K_n^c$, we obtain
\[
\int_{\mathbb R^m} g_n(\mathbf{t})\,d\mathbf{t} \;\xrightarrow{\mathbb P^*}\; \int_{\mathbb R^m} g_0^{(n)}(\mathbf{t})\,d\mathbf{t},
\]
and thus $z_n^{LA,0}/z_n\to 1$ in $\mathbb P^*$-probability.
\paragraph{Step 4: Upgrading to Laplace with prior curvature.}
Recall that the Laplace approximation in the theorem is
\[
z_n^{LA}
=
p(\mathbf y\mid \hat{\mathbf b}_n,\psi)\,\pi(\hat{\mathbf b}_n)\,(2\pi)^{m/2}\,
\big|\det(\mathbf H_{n,\mathbf b^*}+\mathbf K^{-1})\big|^{-1/2}.
\]
By construction,
\[
\frac{z_n^{LA}}{z_n^{LA,0}}
=
\left(
\frac{\det(\mathbf H_{n,\mathbf b^*}+\mathbf K^{-1})}{\det(\mathbf H_{n,\mathbf b^*})}
\right)^{-1/2}.
\]
Since $\mathbf H_{n,\mathbf b^*}$ is symmetric positive definite, it admits a symmetric positive definite square root $\mathbf H_{n,\mathbf b^*}^{1/2}$. Hence
\[
\mathbf H_{n,\mathbf b^*}+\mathbf K^{-1}
=
\mathbf H_{n,\mathbf b^*}^{1/2}
\Bigl(
\mathbf I+\mathbf H_{n,\mathbf b^*}^{-1/2}\mathbf K^{-1}\mathbf H_{n,\mathbf b^*}^{-1/2}
\Bigr)
\mathbf H_{n,\mathbf b^*}^{1/2},
\]
and therefore, by multiplicativity of the determinant,
\begin{align*}
\det(\mathbf H_{n,\mathbf b^*}+\mathbf K^{-1})
&=
\det(\mathbf H_{n,\mathbf b^*}^{1/2})
\det\!\Bigl(
\mathbf I+\mathbf H_{n,\mathbf b^*}^{-1/2}\mathbf K^{-1}\mathbf H_{n,\mathbf b^*}^{-1/2}
\Bigr)
\det(\mathbf H_{n,\mathbf b^*}^{1/2}) \\
&=
\det(\mathbf H_{n,\mathbf b^*})
\det\!\Bigl(
\mathbf I+\mathbf H_{n,\mathbf b^*}^{-1/2}\mathbf K^{-1}\mathbf H_{n,\mathbf b^*}^{-1/2}
\Bigr).
\end{align*}
Thus,
\[
\frac{z_n^{LA}}{z_n^{LA,0}}
=
\det\!\left(
\mathbf I+\mathbf H_{n,\mathbf b^*}^{-1/2}\mathbf K^{-1}\mathbf H_{n,\mathbf b^*}^{-1/2}
\right)^{-1/2}.
\]
By Assumption~(iii), $M_n\to\infty$ and $M_n\|\mathbf H_{n,\mathbf b^*}^{-1/2}\|_{\mathrm{op}}\to 0$, hence
\[
\|\mathbf H_{n,\mathbf b^*}^{-1/2}\|_{\mathrm{op}}\to 0.
\]
Since $\mathbf K^{-1}$ is fixed, submultiplicativity of the operator norm gives
\[
\bigl\|
\mathbf H_{n,\mathbf b^*}^{-1/2}\mathbf K^{-1}\mathbf H_{n,\mathbf b^*}^{-1/2}
\bigr\|_{\mathrm{op}}
\le
\|\mathbf H_{n,\mathbf b^*}^{-1/2}\|_{\mathrm{op}}^2
\|\mathbf K^{-1}\|_{\mathrm{op}}
\to 0.
\]
Therefore
\[
\mathbf I+\mathbf H_{n,\mathbf b^*}^{-1/2}\mathbf K^{-1}\mathbf H_{n,\mathbf b^*}^{-1/2}
\to \mathbf I,
\]
and by continuity of the determinant,
\[
\det\!\left(
\mathbf I+\mathbf H_{n,\mathbf b^*}^{-1/2}\mathbf K^{-1}\mathbf H_{n,\mathbf b^*}^{-1/2}
\right)\to 1.
\]
Hence
\[
\frac{z_n^{LA}}{z_n^{LA,0}}\to 1
\qquad\text{in }\mathbb P^*\text{-probability}.
\]
Combining this with $z_n^{LA,0}/z_n\to 1$ yields $z_n^{LA}/z_n\to 1$, which concludes the proof.
\end{proof}

\subsection{Proof of Corollary \ref{cor:bvm_noprior}}\label{app_proof_BVM}
\begin{proof}[Proof]
We recall the objects and limits established in the proof of Theorem~\ref{thm:misspec_laplace}.
Let $\mathbf t=\mathbf H_{n,\mathbf b^*}^{1/2}(\mathbf b-\hat{\mathbf b}_n)$ and write
$\mathbf b=\hat{\mathbf b}_n+\mathbf H_{n,\mathbf b^*}^{-1/2}\mathbf t$.
Define the unnormalised posterior density in the $\mathbf t$-parametrisation by
\begin{equation}
\label{eq:gn_def_cor_t}
g_n(\mathbf t)
:=
\exp\!\Big(-n\big[f_n(\hat{\mathbf b}_n + \mathbf H_{n,\mathbf b^*}^{-1/2}\mathbf t)-f_n(\hat{\mathbf b}_n)\big]\Big)\,
\pi(\hat{\mathbf b}_n+\mathbf H_{n,\mathbf b^*}^{-1/2}\mathbf t).
\end{equation}
Let $q_n(\mathbf t)$ denote the posterior density of $\mathbf t$ induced by $\pi_n(\mathbf b\mid \mathbf y)$.
By the change of variables $\mathbf b\mapsto \mathbf t=\mathbf H_{n,\mathbf b^*}^{1/2}(\mathbf b-\hat{\mathbf b}_n)$,
its Jacobian is $\big|\det(\mathbf H_{n,\mathbf b^*}^{-1/2})\big|=\big|\det(\mathbf H_{n,\mathbf b^*})\big|^{-1/2}$ and thus
\begin{equation}
\label{eq:qn_def_cor_t}
q_n(\mathbf t)
=
\frac{e^{-n f_n(\hat{\mathbf b}_n+\mathbf H_{n,\mathbf b^*}^{-1/2}\mathbf t)}\,\pi(\hat{\mathbf b}_n+\mathbf H_{n,\mathbf b^*}^{-1/2}\mathbf t)}{z_n}\,
\big|\det(\mathbf H_{n,\mathbf b^*})\big|^{-1/2}.
\end{equation}
Combining \eqref{eq:gn_def_cor_t} and \eqref{eq:qn_def_cor_t}, we obtain the exact identity
\begin{equation}
\label{eq:gn_qn_relation_t}
g_n(\mathbf t)
=
q_n(\mathbf t)\; z_n\; e^{n f_n(\hat{\mathbf b}_n)}\; \big|\det(\mathbf H_{n,\mathbf b^*})\big|^{1/2}.
\end{equation}
Next define the Gaussian limit
\begin{equation}
\label{eq:g0_def_cor_t}
g_0(\mathbf t)
:=
\exp\!\Big(-\tfrac12 \mathbf t^\top \mathbf t\Big)\,\pi(\mathbf b^*).
\end{equation}

\medskip
\noindent
\textbf{Step 1: $L^1$ convergence of the unnormalised densities.}
Let
\[
g_0^{(n)}(\mathbf t):=\exp\!\Big(-\tfrac12 \mathbf t^\top \mathbf t\Big)\,\pi(\hat{\mathbf b}_n).
\]
From Steps~2--3 of the proof of Theorem~\ref{thm:misspec_laplace}, we established that
\begin{equation}
\label{eq:l1_unnorm_cor_t_aux}
\int_{\mathbb R^m} \big|g_n(\mathbf t)-g_0^{(n)}(\mathbf t)\big|\,d\mathbf t \xrightarrow{\mathbb P^*} 0.
\end{equation}
Moreover, since $\hat{\mathbf b}_n\to\mathbf b^*$ in $\mathbb P^*$-probability and $\pi$ is continuous at $\mathbf b^*$,
\[
\big|\pi(\hat{\mathbf b}_n)-\pi(\mathbf b^*)\big|\xrightarrow{\mathbb P^*}0.
\]
Therefore
\begin{align*}
\int_{\mathbb R^m} \big|g_0^{(n)}(\mathbf t)-g_0(\mathbf t)\big|\,d\mathbf t
&=
\big|\pi(\hat{\mathbf b}_n)-\pi(\mathbf b^*)\big|
\int_{\mathbb R^m}\exp\!\Big(-\tfrac12 \mathbf t^\top \mathbf t\Big)\,d\mathbf t 
\xrightarrow{\mathbb P^*}0.
\end{align*}
By the triangle inequality, it follows that
\begin{equation}
\label{eq:l1_unnorm_cor_t}
\int_{\mathbb R^m} \big|g_n(\mathbf t)-g_0(\mathbf t)\big|\,d\mathbf t \xrightarrow{\mathbb P^*} 0.
\end{equation}

\medskip
\noindent
\textbf{Step 2: convergence of the normalising constants.}
Define the normalising constants
\begin{equation}
\label{eq:an_def_cor_t}
a_n := \int_{\mathbb R^m} g_n(\mathbf t)\,d\mathbf t,
\qquad
a := \int_{\mathbb R^m} g_0(\mathbf t)\,d\mathbf t.
\end{equation}
By definition of $g_n$ and \eqref{eq:gn_qn_relation_t},
\[
a_n
=
\int_{\mathbb R^m} g_n(\mathbf t)\,d\mathbf t
=
z_n\,e^{n f_n(\hat{\mathbf b}_n)}\,\big|\det(\mathbf H_{n,\mathbf b^*})\big|^{1/2}
\int_{\mathbb R^m} q_n(\mathbf t)\,d\mathbf t
=
z_n\,e^{n f_n(\hat{\mathbf b}_n)}\,\big|\det(\mathbf H_{n,\mathbf b^*})\big|^{1/2},
\]
since $\int q_n(\mathbf t)\,d\mathbf t=1$.
Moreover, by \eqref{eq:g0_def_cor_t} the constant $a$ is finite and equals
\[
a
=
\pi(\mathbf b^*)\int_{\mathbb R^m}\exp\!\Big(-\tfrac12 \mathbf t^\top \mathbf t\Big)\,d\mathbf t
=
\pi(\mathbf b^*)(2\pi)^{m/2}.
\]
In the proof of Theorem~\ref{thm:misspec_laplace} we proved that
\[
z_n\,e^{n f_n(\hat{\mathbf b}_n)}\,\big|\det(\mathbf H_{n,\mathbf b^*})\big|^{1/2}
\xrightarrow{\mathbb P^*}
\pi(\mathbf b^*)(2\pi)^{m/2},
\]
hence
\begin{equation}
\label{eq:an_conv_cor_t}
a_n \xrightarrow{\mathbb P^*} a,
\qquad\text{with } a>0.
\end{equation}

\medskip
\noindent
\textbf{Step 3: $L^1$ convergence of the \emph{normalised} densities.}
Define the normalised densities
\begin{equation}
\label{eq:qn_q_def_cor_t}
\tilde q_n(\mathbf t):=\frac{g_n(\mathbf t)}{a_n},
\qquad
\tilde q(\mathbf t):=\frac{g_0(\mathbf t)}{a}.
\end{equation}
By construction, $\tilde q_n$ is a probability density on $\mathbb R^m$. Note that by
\eqref{eq:gn_qn_relation_t} and the definition of $a_n$ in \eqref{eq:an_def_cor_t},
\[
\tilde q_n(\mathbf t)=\frac{g_n(\mathbf t)}{a_n}
=
\frac{q_n(\mathbf t)\,z_n\,e^{n f_n(\hat{\mathbf b}_n)}\,|\det(\mathbf H_{n,\mathbf b^*})|^{1/2}}
{z_n\,e^{n f_n(\hat{\mathbf b}_n)}\,|\det(\mathbf H_{n,\mathbf b^*})|^{1/2}}
=
q_n(\mathbf t),
\]
so $\tilde q_n$ coincides exactly with the posterior density of $\mathbf t$.
We now bound the $L^1$ distance:
\[
\int_{\mathbb R^m}\big|\tilde q_n(\mathbf t)-\tilde q(\mathbf t)\big|\,d\mathbf t
=
\int_{\mathbb R^m}\left|\frac{g_n(\mathbf t)}{a_n}-\frac{g_0(\mathbf t)}{a}\right|\,d\mathbf t.
\]
Add and subtract $g_0(\mathbf t)/a_n$ and use the triangle inequality:
\begin{align*}
\int \left|\frac{g_n}{a_n}-\frac{g_0}{a}\right|
&\le
\int \left|\frac{g_n-g_0}{a_n}\right|
+
\int \left|g_0\left(\frac{1}{a_n}-\frac{1}{a}\right)\right| \\
&=
\frac{1}{a_n}\int |g_n-g_0|
+
\left|\frac{1}{a_n}-\frac{1}{a}\right|\int g_0.
\end{align*}
Since $a_n\xrightarrow{\mathbb P^*} a>0$ by \eqref{eq:an_conv_cor_t}, we have $1/a_n\xrightarrow{\mathbb P^*}1/a$
and in particular $(1/a_n)$ is bounded in probability. Together with \eqref{eq:l1_unnorm_cor_t} and
$\int g_0 = a <\infty$, this implies
\[
\int_{\mathbb R^m}\big|\tilde q_n(\mathbf t)-\tilde q(\mathbf t)\big|\,d\mathbf t \xrightarrow{\mathbb P^*} 0.
\]
Finally, $\tilde q(\mathbf t)=g_0(\mathbf t)/a$ is exactly the standard Gaussian density
$\mathcal N\!\left(\mathbf t;0,\mathbf I_m\right)$, hence the posterior density of
$\mathbf t=\mathbf H_{n,\mathbf b^*}^{1/2}(\mathbf b-\hat{\mathbf b}_n)$ converges in $L^1$ to this Gaussian density
in $\mathbb P^*$-probability.
\end{proof}

\subsection{Consistency of the TKC estimator}
\label{app:proof_cc}
In this section, we first show that the triangular kernel curvature (TKC) estimator corresponds to a triangular kernel density estimator. We then use this result together with the CC assumption to show that the TKC estimator is consistent for the true curvature when evaluated at $\boldsymbol{\mu}^*$.

\begin{proof}[Proof of proposition \ref{prop:cc_consistency}]
Fix $h=\Delta\mu$. For a single observation $y_i$ and location $\mu_i^*$, define $z_i:=y_i-\mu_i^*$ and consider the symmetric second difference of the pinball loss
\[
D_h(z_i):=\rho_\tau(z_i+h)-2\rho_\tau(z_i)+\rho_\tau(z_i-h),\qquad 
\rho_\tau(u)=u\big(\tau-\mathbbm 1\{u<0\}\big).
\]
A direct case distinction according to whether the interval $(z_i-h,z_i+h)$ crosses the kink at $0$ yields
\begin{equation}
\label{eq:kernel_repr}
D_h(z_i)=(h-|z_i|)\,\mathbbm 1\{|z_i|<h\}.
\end{equation}
Summing \eqref{eq:kernel_repr} over $i$ and scaling gives the TKC estimator at $\boldsymbol\mu^*$:
\[
\hat C_h(\boldsymbol\mu^*)
=\frac{1}{n\lambda h^2}\sum_{i=1}^n D_h(y_i-\mu_i^*)
=\frac{1}{n\lambda h^2}\sum_{i=1}^n (h-|y_i-\mu_i^*|)\,\mathbbm 1\{|y_i-\mu_i^*|<h\}.
\]
Equivalently, with $K(u):=(1-|u|)_+$,
\[
\hat C_h(\boldsymbol\mu^*)=\frac{1}{\lambda}\cdot \frac{1}{nh}\sum_{i=1}^n K\!\Big(\frac{y_i-\mu_i^*}{h}\Big).
\]
We now establish consistency. Write
\[
X_{n,i}:=\frac{1}{h}K\!\Big(\frac{y_i-\mu_i^*}{h}\Big),
\qquad\text{so that}\qquad
\hat C_h(\boldsymbol\mu^*)=\frac{1}{\lambda}\cdot \frac{1}{n}\sum_{i=1}^n X_{n,i}.
\]
For the mean, using the change of variables $u=(y-\mu_i^*)/h$,
\[
\mathbb E_{\mathbb P^*}[X_{n,i}]
=\int K(u)\, f_i^*(\mu_i^*+hu)\,du,
\]
where $f_i^*$ denotes the (true) density of $Y_i$ under $\mathbb P^*$. For brevity, all expectations and probabilities below are taken under $\mathbb P^*$. Since $K$ is supported on $[-1,1]$, bounded, and satisfies $\int K(u)\,du=1$, continuity of $f_i^*$ at $\mu_i^*$ implies
\[
\mathbb E[X_{n,i}] \to f_i^*(\mu_i^*) \quad\text{as }h\to0.
\]
Under \textbf{(CC)}, $f_i^*(\mu_i^*)=c^*(\boldsymbol\mu^*)$ for all $i$, hence
\[
\mathbb E\!\left[\hat C_h(\boldsymbol\mu^*)\right]
=\frac{1}{\lambda}\cdot \frac{1}{n}\sum_{i=1}^n \mathbb E[X_{n,i}]
\longrightarrow \frac{c^*(\boldsymbol\mu^*)}{\lambda}.
\]
For the variance, since $X_{n,i}$ depends only on $Y_i-\mu_i^*$ and the observations $Y_i$ are independent, the variables $X_{n,i}$ are independent and thus
\[
\mathrm{Var}\!\left(\hat C_h(\boldsymbol\mu^*)\right)
=\frac{1}{\lambda^2}\cdot \frac{1}{n^2}\sum_{i=1}^n \mathrm{Var}(X_{n,i})
\le \frac{1}{\lambda^2}\cdot \frac{1}{n^2}\sum_{i=1}^n \mathbb E[X_{n,i}^2].
\]
Since $0\le K\le 1$ and $K(u)=0$ for $|u|>1$,
\[
\mathbb E[X_{n,i}^2]\le \frac{1}{h^2}\,\mathbb P^*\big(|y_i-\mu_i^*|<h\big).
\]
By continuity of $f_i^*$ at $\mu_i^*$, $\mathbb P^*(|y_i-\mu_i^*|<h)=\int_{\mu_i^*-h}^{\mu_i^*+h} f_i^*(y)\,dy=O(h)$ as $h\to0$, hence $\mathbb E[X_{n,i}^2]=O(h^{-1})$ and therefore
\[
\mathrm{Var}\!\left(\hat C_h(\boldsymbol\mu^*)\right)=O\!\left(\frac{1}{nh}\right)\to 0
\qquad\text{whenever } nh\to\infty.
\]
We now combine the mean and variance bounds.
Let $\bar X_n := \frac{1}{n}\sum_{i=1}^n X_{n,i}$ so that $\hat C_h(\boldsymbol\mu^*)=\bar X_n/\lambda$. Decompose
\[
\hat C_h(\boldsymbol\mu^*)-\frac{c^*(\boldsymbol\mu^*)}{\lambda}
=
\underbrace{\frac{1}{\lambda}\big(\bar X_n-\mathbb E[\bar X_n]\big)}_{\text{stochastic term}}
+
\underbrace{\frac{1}{\lambda}\big(\mathbb E[\bar X_n]-c^*(\boldsymbol\mu^*)\big)}_{\text{bias term}}.
\]
The bias term converges to zero since $\mathbb E[\bar X_n]\to c^*(\boldsymbol\mu^*)$ as $h\to0$ under continuity and \textbf{(CC)}. For the stochastic term, Chebyshev’s inequality gives for any $\varepsilon>0$,
\[
\mathbb P^*\!\left(\left|\bar X_n-\mathbb E[\bar X_n]\right|>\lambda\varepsilon\right)
\le
\frac{\mathrm{Var}(\bar X_n)}{\lambda^2\varepsilon^2}
=
\frac{1}{\lambda^2\varepsilon^2}\cdot \frac{1}{n^2}\sum_{i=1}^n \mathrm{Var}(X_{n,i})
=O\!\left(\frac{1}{nh}\right)\xrightarrow[n\to\infty]{}0,
\]
where we used $\mathrm{Var}(X_{n,i})\le \mathbb E[X_{n,i}^2]=O(h^{-1})$ and $nh\to\infty$. Hence $\bar X_n-\mathbb E[\bar X_n]\xrightarrow{\mathbb P^*}0$, and together with the bias term this implies $\hat C_h(\boldsymbol\mu^*)\xrightarrow{\mathbb P^*}c^*(\boldsymbol\mu^*)/\lambda$.
\end{proof}

\subsection{Proof of Laplace consistency with plug-in curvature}
\label{appendix:proof_plugin}
\begin{proof}[Proof of Proposition~\ref{prop:laplace_plugin_cc}]
Write the population-curvature Laplace approximation from Theorem~\ref{thm:misspec_laplace} as
\[
z_n^{LA}
=
p(\mathbf y\mid \hat{\mathbf b}_n,\psi)\,\pi(\hat{\mathbf b}_n)\,(2\pi)^{m/2}\,
\big|\det(\mathbf H_{n,\mathbf b^*}+\mathbf K^{-1})\big|^{-1/2},
\]
and define the plug-in version
\[
\hat z_n^{LA}
=
p(\mathbf y\mid \hat{\mathbf b}_n,\psi)\,\pi(\hat{\mathbf b}_n)\,(2\pi)^{m/2}\,
\big|\det(\hat{\mathbf H}_{n,\mathbf b^*}+\mathbf K^{-1})\big|^{-1/2}.
\]
All terms coincide except for the determinant factor, hence
\begin{equation}
\label{eq:ratio_plugin_laplace}
\frac{\hat z_n^{LA}}{z_n^{LA}}
=
\left(
\frac{\det(\mathbf H_{n,\mathbf b^*}+\mathbf K^{-1})}{\det(\hat{\mathbf H}_{n,\mathbf b^*}+\mathbf K^{-1})}
\right)^{1/2}.
\end{equation}
Under \textbf{(CC)},
\[
\mathbf H_{n,\mathbf b^*}=\alpha\,\mathbf Z^\top\mathbf Z,
\qquad
\hat{\mathbf H}_{n,\mathbf b^*}=\hat\alpha_n\,\mathbf Z^\top\mathbf Z,
\]
with
\[
\alpha:=\frac{c^*(\boldsymbol\mu^*)}{\lambda}>0,
\qquad
\hat\alpha_n:=\frac{\hat c_n(\hat{\boldsymbol\mu}_n)}{\lambda}.
\]
Let $\mathbf A:=\mathbf Z^\top\mathbf Z$, so that $\mathbf A$ is symmetric positive semidefinite.
Since $\mathbb P^*(\hat c_n(\hat{\boldsymbol\mu}_n)>0)\to 1$ and $\mathbf K^{-1}$ is positive semidefinite, both
$\mathbf K^{-1}+\alpha\mathbf A$ and $\mathbf K^{-1}+\hat\alpha_n\mathbf A$ are positive definite with probability tending to one.

Consider the function
\[
\varphi(x):=\log\det(\mathbf K^{-1}+x\,\mathbf A),\qquad x>0.
\]
This function is continuous on $(0,\infty)$: for each $x>0$, the matrix $\mathbf K^{-1}+x\mathbf A$ is symmetric positive definite, its eigenvalues depend continuously on $x$, and $\log\det$ is the sum of the logarithms of these eigenvalues.
Therefore, by the continuous mapping theorem and the assumed consistency
$
\hat\alpha_n\xrightarrow{\mathbb P^*}\alpha,
$
we obtain
\[
\varphi(\hat\alpha_n)-\varphi(\alpha)\xrightarrow{\mathbb P^*}0,
\]
that is,
\[
\log\frac{\det(\mathbf K^{-1}+\hat\alpha_n\mathbf A)}{\det(\mathbf K^{-1}+\alpha\mathbf A)}
\xrightarrow{\mathbb P^*}0.
\]
Exponentiating yields
\[
\frac{\det(\mathbf K^{-1}+\alpha\mathbf A)}{\det(\mathbf K^{-1}+\hat\alpha_n\mathbf A)}
\xrightarrow{\mathbb P^*}1.
\]
Plugging this into \eqref{eq:ratio_plugin_laplace} gives
\[
\frac{\hat z_n^{LA}}{z_n^{LA}}\xrightarrow{\mathbb P^*}1.
\]
Finally, by assumption,
\[
\frac{z_n^{LA}}{z_n}\xrightarrow{\mathbb P^*}1.
\]
Hence Slutsky's theorem implies
\[
\frac{\hat z_n^{LA}}{z_n}
=
\frac{\hat z_n^{LA}}{z_n^{LA}}\cdot \frac{z_n^{LA}}{z_n}
\xrightarrow{\mathbb P^*}1.
\]
\end{proof}

\subsection{Discussion of the assumptions of Theorems ~\ref{thm:FL_consistency} and \ref{thm:misspec_laplace}}
\label{app:discussion_assumptions_misspec}

We briefly discuss the role of each assumption in Theorems ~\ref{thm:FL_consistency} and \ref{thm:misspec_laplace}, how it relates to existing asymptotic theory in the i.i.d.\ case, and what additional conditions are typically needed in regression settings. To unify the discussion, we write generically $\mathbf H_{n, \textbf{b}^*}$ for the relevant local curvature matrix, which equals the Fisher information in the correctly specified case and the expected Hessian of the population criterion in the misspecified case.

\paragraph{\textbf{Assumption (i): LAN around the (pseudo-)true parameter.}}
Its role in the proof is to provide a local quadratic expansion of the log-likelihood on the natural information scale, which is the starting point for both the Gaussian approximation of the posterior and the Laplace approximation of the marginal likelihood. In particular, the proof only uses the expansion locally, after rescaling by $\mathbf H_{n,\mathbf b^*}^{-1/2}$.

In the correctly specified i.i.d.\ case, a standard route to such an expansion is differentiability in quadratic mean. In regular parametric models, DQM implies local asymptotic normality at the classical root-$n$ scale; see, for example, \cite{van2000asymptotic}[Theorem 7.2]. In that setting, one typically writes
\[
\mathbf H_{n,\mathbf b^*}=n\,\mathbf H_{\mathbf b^*},
\]
so that the information-scaled perturbation $\mathbf H_{n,\mathbf b^*}^{-1/2}\mathbf t$ reduces to the familiar $n^{-1/2}$ localization.

Under misspecification, the relevant local expansion is instead centered at the pseudo-true parameter $\mathbf b^*$ and involves the curvature of the population criterion rather than the Fisher information of a correctly specified model. In the i.i.d.\ setting, this is the framework of misspecified Bernstein--von Mises theory; see \cite{Miss-BvM}. Their conditions are formulated in terms of local regularity of the likelihood ratio around the pseudo-true parameter and yield precisely the type of local quadratic expansion required here.

In regression settings, verifying Assumption~(i) requires additional conditions beyond those needed in the i.i.d.\ location case. First, one needs identifiability of the pseudo-true parameter, which typically amounts to full-rank or positive-definiteness conditions on the relevant design matrix. Second, one needs a central limit theorem for the normalized score, which in linear quantile regression corresponds to a multivariate CLT for design-weighted quantile scores. Finally, positivity and local regularity of the true conditional density at the target quantile ensure that the expected Hessian is well defined and positive definite. These ingredients are standard in regression quantile asymptotics; see \cite{1aa6b708-cbcf-320f-b3f3-298154ee2aeb} and \cite{koenker2005quantile}.

\paragraph{\textbf{Assumption (ii): asymptotic linearity of the mode.}}
Its role in the proof is to ensure that the posterior mode lies on the same local scale as the LAN expansion, so that the local quadratic approximation can be recentered at $\hat{\mathbf b}_n$. This is exactly what allows Step~1 of the proof to transfer the expansion from $\mathbf b^*$ to a neighborhood of the mode.

In the i.i.d.\ setting, this is the familiar asymptotic linearity of quantile-type M-estimators. For the sample quantile, it is the classical Bahadur representation; more generally, for M-estimators based on convex loss functions one obtains root-$n$ consistency and asymptotic linearity under standard smoothness and identifiability conditions. 

In regression settings, the same conclusion requires additional design conditions. In particular, unlike the i.i.d.\ location model, identifiability is no longer automatic from the pinball loss alone: it must be induced by the design through full-rank conditions. Once this is in place, asymptotic normality and Bahadur-type representations of regression quantiles provide the natural analogue of asymptotic linearity. Since the asymmetric Laplace criterion coincides with the pinball loss up to constants, the leading first-order expansion of the posterior mode is governed by the same score structure. The prior contributes an additional smooth term, but under the information-growth conditions of the theorem this does not affect the leading stochastic term. See again \cite{1aa6b708-cbcf-320f-b3f3-298154ee2aeb} and \cite{koenker2005quantile}.

\paragraph{\textbf{Assumption (iii): separation and information growth.}}
Its role in the proof is to control the tail contribution to the marginal likelihood outside the local quadratic region. The separation part ensures that, once one leaves an expanding ball in the local $\mathbf t$-coordinates, the log-likelihood is uniformly lower by order $n$. This makes the tail contribution negligible relative to the local Gaussian part. The information-growth part ensures that the Hessian-based reparameterization does not offset this exponential decay through the Jacobian term.

In the classical i.i.d.\ setting, such localization is often hidden inside the usual root-$n$ normalization, since one writes $\mathbf H_{n,\mathbf b^*}=n\mathbf H_{\mathbf b^*}$ and the information growth is then automatic. Our formulation makes this dependence explicit by working directly with $\mathbf H_{n,\mathbf b^*}^{-1/2}$. This is particularly useful in settings such as grouped random effects, where information accumulates anisotropically through matrices such as $\mathbf Z^\top\mathbf Z$ rather than through a single scalar factor $n$.

The separation condition is stronger than what is strictly necessary for posterior concentration alone. In more classical Bernstein--von Mises arguments, related control is often obtained via uniformly consistent tests or equivalent localization conditions. In generalized or misspecified posterior settings, however, such conditions are less transparent, and stronger direct separation assumptions are often used instead. This is also the perspective taken by \cite{miller2021asymptotic}, who emphasizes that such assumptions are stronger than necessary but provide a clean route to concentration and Laplace approximations in generalized posterior settings.

In regression settings, Assumption~(iii) should be read as a nonlocal identifiability condition on the sample criterion, ensuring that the local quadratic region around the mode captures all asymptotically relevant posterior mass. Verifying it typically requires a uniform law of large numbers for the empirical criterion together with strict local convexity of the population risk around the pseudo-true parameter.

\paragraph{\textbf{Assumption (iv): prior regularity.}}
Its role in the proof is twofold. First, local positivity ensures that the prior does not eliminate posterior mass near the pseudo-true parameter. Second, local smoothness ensures that the prior can be treated as approximately constant on the local information scale, and in the refined Laplace approximation it also justifies the inclusion of the prior curvature term.

For first-order Bernstein--von Mises type arguments, local positivity and continuity of the prior are typically sufficient. In our theorem, however, we also cover the practical Laplace approximation based on the curvature of the full log-posterior, which includes the prior curvature term $\mathbf K^{-1}=-\nabla^2\log\pi(\mathbf b)\rvert_{\mathbf b=\mathbf b^*}$. This is why local twice differentiability is required when the prior curvature is retained in the approximation. By contrast, generalized-posterior results such as those of \cite{miller2021asymptotic} focus on likelihood curvature alone and therefore do not need this extra smoothness requirement.

In our setting, Assumption~(iv) is immediate for the Gaussian priors used throughout: they are strictly positive, smooth, and twice continuously differentiable, so both the local positivity and the prior-curvature term are automatic.

\clearpage
\section{Estimating the scale parameter helps mitigate misspecification of the Fisher-Laplace approximation}\label{fisher_laplace_missp}
\begin{figure}[ht!]
    \centering
    \includegraphics[width=1\linewidth]{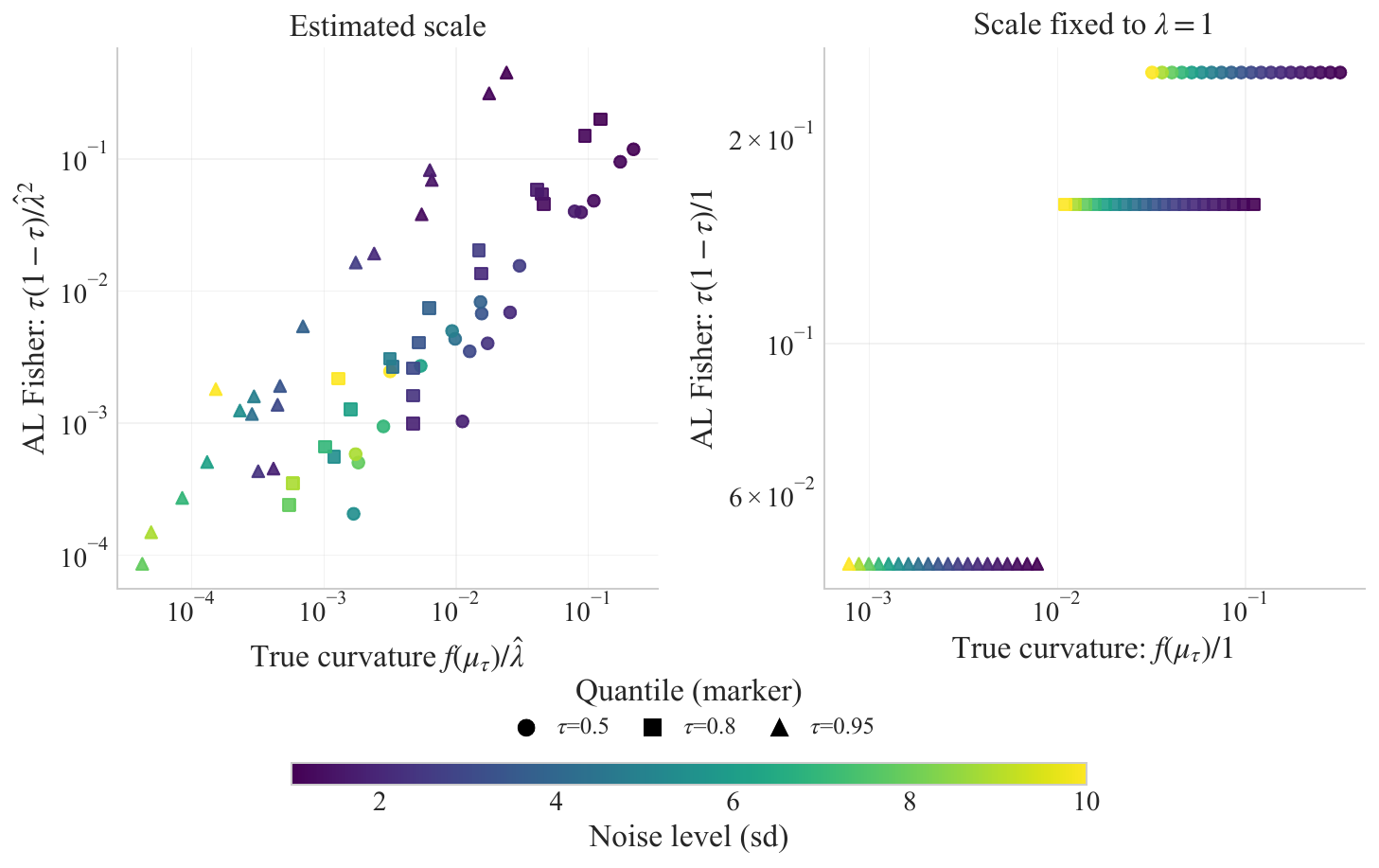}
    \caption{
    \textbf{Comparison of asymmetric Laplace Fisher curvature with the asymptotic curvature}.
    \textit{Left:} Fisher curvature $\tau(1-\tau)/\hat \lambda^2$ using the estimated asymmetric Laplace scale $\hat \lambda$ plotted against the correct asymptotic curvature $f(\mu_\tau)/\hat \lambda$. \textit{Right:} same comparison when the Fisher curvature is computed with fixed scale $\lambda = 1$, i.e. $\tau(1-\tau)$ vs. $f(\mu_\tau)$. Marker shape indicates the target quantile $\tau$; color indicates the noise level (standard deviation). Using an estimated $\hat\lambda$ yields a Fisher curvature $\tau(1-\tau)/\hat\lambda^2$ that tracks the correct asymptotic curvature $f(\mu_\tau)/\hat\lambda$ much more closely (left panel) than fixing $\lambda=1$ (right panel).
    }
    \label{fig:estimated_scale}
\end{figure}

\clearpage
\section{Supplementary results referenced in the main text}
\label{app:supplementary}
Best performance ($\pm$ 2 std. err.) are in bold. A missing value --- indicates that the method is not applicable or convergence failure.
\vspace{-0.25cm}

\subsection{Additional results for the simulation experiments}

\subsubsection{Single-level Grouped Random Effects: Student-t noise distribution}
\label{app:student_t}

\begin{table}[ht!]
\centering
\begin{tabular}{lccccccc}
\toprule
\multicolumn{3}{c}{} & \multicolumn{5}{c}{RMSE} \\
\cmidrule(lr){4-8}
Noise & $m$ & $n_j$ & \multicolumn{1}{c}{TKC} & \multicolumn{1}{c}{Fisher} & \multicolumn{1}{c}{BayesQR} & \multicolumn{1}{c}{BRMS} & \multicolumn{1}{c}{LQMM} \\
\midrule
t & 100 & 10 & \textbf{0.25$_\pm{\text{\tiny 0.012}}$} & \textbf{0.25$_\pm{\text{\tiny 0.011}}$} & \textbf{0.24$_\pm{\text{\tiny 0.013}}$} & 0.48$_\pm{\text{\tiny 0.0062}}$ & \textbf{0.26$_\pm{\text{\tiny 0.014}}$} \\
t & 100 & 100 & \textbf{0.067$_\pm{\text{\tiny 0.0021}}$} & \textbf{0.067$_\pm{\text{\tiny 0.0021}}$} & \textbf{0.067$_\pm{\text{\tiny 0.002}}$} & 0.49$_\pm{\text{\tiny 0.003}}$ & 0.36$_\pm{\text{\tiny 0.11}}$ \\
t & 100 & 500 & \textbf{0.029$_\pm{\text{\tiny 2.8e-04}}$} & \textbf{0.029$_\pm{\text{\tiny 2.8e-04}}$} & \textbf{0.029$_\pm{\text{\tiny 3.0e-04}}$} & --- & 0.6$_\pm{\text{\tiny 0.18}}$ \\
\bottomrule
\end{tabular}
\caption{RMSE for quantile predictions when using a Student-t noise distribution.}
\end{table}
\vspace{-0.25cm}
\begin{table}[htbp]
\centering
\begin{tabular}{lccccccc}
\toprule
\multicolumn{3}{c}{} & \multicolumn{5}{c}{Runtime} \\
\cmidrule(lr){4-8}
Noise & $m$ & $n_j$ & \multicolumn{1}{c}{TKC} & \multicolumn{1}{c}{Fisher} & \multicolumn{1}{c}{BayesQR} & \multicolumn{1}{c}{BRMS} & \multicolumn{1}{c}{LQMM} \\
\midrule
t & 100 & 10 & 0.18$_\pm{\text{\tiny 0.016}}$ & 0.26$_\pm{\text{\tiny 0.15}}$ & 48.5$_\pm{\text{\tiny 0.5}}$ & 467$_\pm{\text{\tiny 10}}$ & \textbf{0.04$_\pm{\text{\tiny 0.0012}}$} \\
t & 100 & 100 & 6.7$_\pm{\text{\tiny 0.57}}$ & 1.4$_\pm{\text{\tiny 0.13}}$ & 461$_\pm{\text{\tiny 2.1}}$ & 2781$_\pm{\text{\tiny 124}}$ & \textbf{0.3$_\pm{\text{\tiny 0.026}}$} \\
t & 100 & 500 & 20.4$_\pm{\text{\tiny 2.2}}$ & \textbf{1.1$_\pm{\text{\tiny 0.14}}$} & 2315$_\pm{\text{\tiny 11.5}}$ & --- & 3.6$_\pm{\text{\tiny 0.77}}$ \\
\bottomrule
\end{tabular}
\caption{Runtime (s) when using a Student-t noise distribution.}
\end{table}

\vspace{-0.25cm}
\begin{table}[htbp]
\centering
\begin{tabular}{lcccccc}
\toprule
\multicolumn{3}{c}{} & \multicolumn{4}{c}{MSE of Random Effect Variance} \\
\cmidrule(lr){4-7}
Noise & $m$ & $n_j$ & \multicolumn{1}{c}{TKC} & \multicolumn{1}{c}{Fisher} & \multicolumn{1}{c}{BRMS} & \multicolumn{1}{c}{LQMM} \\
\midrule
t & 100 & 10 & 0.031 & \textbf{0.011} & 0.028 & 0.038 \\
t & 100 & 100 & \textbf{0.0033} & 0.021 & --- & 0.43 \\
t & 100 & 500 & \textbf{0.013} & 0.015 & --- & 0.78 \\

\bottomrule
\end{tabular}
\caption{Hyperparameter MSE when using a Student-t noise distribution.}
\label{tab:hyperparams_mse}
\end{table}

\subsubsection{Single-level Grouped Random Effects: Hyperparameter MSE}
\label{app:single_hyper}
\begin{table}[ht!]
\centering
\begin{tabular}{lcccccc}
\toprule
\multicolumn{3}{c}{} & \multicolumn{4}{c}{MSE of Random Effect Variance (1)} \\
\cmidrule(lr){4-7}
Noise & $m$ & $n_j$ & \multicolumn{1}{c}{TKC} & \multicolumn{1}{c}{Fisher} & \multicolumn{1}{c}{BRMS} & \multicolumn{1}{c}{LQMM} \\
\midrule
ALD & 100 & 10 & 0.052 & 0.021 & \textbf{0.019} & 0.055 \\
ALD & 100 & 100 & \textbf{6.7e-05} & 0.004 & --- & 0.71 \\
ALD & 100 & 500 & 4.9e-06 & \textbf{4.6e-06} & --- & 0.72 \\
N & 100 & 10 & 0.016 & \textbf{0.013} & 0.017 & 0.046 \\
N & 100 & 100 & \textbf{0.0017} & 0.016 & 0.0097 & 0.56 \\
N & 100 & 500 & \textbf{7.9e-04} & 0.0023 & --- & 0.8 \\

\bottomrule
\end{tabular}
\caption{Hyperparameter MSE.}
\end{table}

\newpage
\subsubsection{Crossed Random Effects}\label{app_crossed}
We consider a two-factor crossed random effects model with $m_1=100$ levels for the first factor and $m_2=50$ levels for the second factor. 
For each level of the first factor, we generate $n_j\in\{100,500\}$ observations, yielding $N=\sum_{j=1}^{m_1} n_j$ total observations. 
The levels of the second factor are then assigned cyclically within each first-factor level, yielding a completely crossed design.  
Both random effects follow Gaussian distributions with variances $\sigma^2_{u_1} = 1$ and $\sigma^2_{u_2} = 2$. We use the same noise generating likelihoods and the same signal-to-noise ratios as before. Since \texttt{lqmm} does not support crossed random effects, we compare our Laplace approximations against \texttt{brms} and \texttt{bayesQR} with identical MCMC specifications as in the single-level experiments. 

Table~\ref{tab:crossed_sim} shows that the TKC and Fisher-Laplace approximations match \texttt{bayesQR} in accuracy whenever \texttt{bayesQR} does not fail to converge, while having lower runtimes by orders of magnitude. For example, for $n_j=100$, the three methods achieve essentially identical RMSE across both noise models, but \texttt{bayesQR} requires more than $1,000$ seconds, whereas TKC and Fisher-Laplace finish in only a few seconds. For $n_j=500$, \texttt{bayesQR} fails to run within the one-hour time limit (---). \texttt{brms} did not finish within the one-hour time limit in any crossed-effects configuration and is therefore omitted.

\begin{table}[ht!]
\centering
\begin{tabular}{@{}lccccc@{\hspace{1.4em}}ccc@{}}
\toprule
\multicolumn{3}{c}{} & \multicolumn{3}{c}{RMSE} & \multicolumn{3}{c}{Runtime} \\
\cmidrule(lr){4-9}
Noise & $m_1$ & $n_j$ & \multicolumn{1}{c}{TKC} & \multicolumn{1}{c}{Fisher} & \multicolumn{1}{c}{BayesQR} & \multicolumn{1}{c}{TKC} & \multicolumn{1}{c}{Fisher} & \multicolumn{1}{c}{BayesQR} \\
\cmidrule(lr){1-6} \cmidrule(l){7-9}
ALD & 100 & 100 & \textbf{0.035$_\pm{\text{\tiny 7.3e-04}}$} & \textbf{0.035$_\pm{\text{\tiny 7.4e-04}}$} & \textbf{0.034$_\pm{\text{\tiny 6.0e-04}}$} & 7.3$_\pm{\text{\tiny 0.2}}$ & \textbf{2.6$_\pm{\text{\tiny 0.22}}$} & 1264$_\pm{\text{\tiny 7.5}}$ \\
ALD & 100 & 500 & \textbf{0.014$_\pm{\text{\tiny 2.4e-04}}$} & \textbf{0.014$_\pm{\text{\tiny 2.4e-04}}$} & --- & 25$_\pm{\text{\tiny 1.1}}$ & \textbf{8.9$_\pm{\text{\tiny 0.99}}$} & --- \\
N & 100 & 100 & \textbf{0.09$_\pm{\text{\tiny 0.0015}}$} & \textbf{0.09$_\pm{\text{\tiny 0.0015}}$} & \textbf{0.088$_\pm{\text{\tiny 0.0013}}$} & 5.7$_\pm{\text{\tiny 0.21}}$ & \textbf{2.5$_\pm{\text{\tiny 0.26}}$} & 1285$_\pm{\text{\tiny 4.4}}$ \\
N & 100 & 500 & \textbf{0.04$_\pm{\text{\tiny 7.6e-04}}$} & \textbf{0.04$_\pm{\text{\tiny 7.6e-04}}$} & --- & 22.2$_\pm{\text{\tiny 0.49}}$ & \textbf{7.3$_\pm{\text{\tiny 0.68}}$} & --- \\
\bottomrule
\end{tabular}
\caption{Crossed Random Effects: Average quantile test RMSE $\pm$ standard error and runtime (in seconds) for estimation and prediction.
Boldface indicates methods whose performance is within two standard errors of the best result.
Dashes (---) indicate that the method did not converge or exceeded the 1-hour time limit.
\texttt{brms} did not finish within the time limit in any configuration and is omitted from the table.}
\label{tab:crossed_sim}
\end{table}

\begin{table}[htbp]
\centering
\begin{tabular}{lcccccccccc}
\toprule
\multicolumn{3}{c}{} & \multicolumn{4}{c}{MSE of Random Effect Variance (1)} & \multicolumn{4}{c}{MSE of Random Effect Variance (2)} \\
\cmidrule(lr){4-11}
Noise & $m$ & $n_j$ & \multicolumn{1}{c}{TKC} & \multicolumn{1}{c}{Fisher} & \multicolumn{1}{c}{BRMS} & \multicolumn{1}{c}{BayesQR} & \multicolumn{1}{c}{TKC} & \multicolumn{1}{c}{Fisher} & \multicolumn{1}{c}{BRMS} & \multicolumn{1}{c}{BayesQR} \\
\midrule
ALD & 100 & 100 & 0.055 & \textbf{0.042} & --- & --- & 0.49 & \textbf{0.096} & --- & --- \\
ALD & 100 & 500 & 0.11 & \textbf{0.026} & --- & --- & 1.1 & \textbf{0.35} & --- & --- \\
N & 100 & 100 & \textbf{0.042} & 0.048 & --- & --- & 0.19 & \textbf{0.15} & --- & --- \\
N & 100 & 500 & 0.065 & \textbf{0.022} & --- & --- & 0.71 & \textbf{0.31} & --- & --- \\

\bottomrule
\end{tabular}
\caption{Hyperparameters MSE.}
\end{table}

\clearpage
\subsubsection{Additional Results for Gaussian Process Simulated Experiments}

\begin{table}[ht!]
\centering
\begin{tabular}{lllcccccc}
\toprule
\multicolumn{3}{c}{} & \multicolumn{6}{c}{RMSE} \\
\cmidrule(lr){4-9}
Noise & $n$ & d & TKC & FL & VI & VIVA & QGAM & QGAM Int. \\
\midrule
N & 1{,}000 & 2 & 0.22$_\pm{\text{\tiny 0.0012}}$ & 0.22$_\pm{\text{\tiny 0.0099}}$ & \textbf{0.2$_\pm{\text{\tiny 0.0036}}$} & 0.44$_\pm{\text{\tiny 0.068}}$ & 0.67$_\pm{\text{\tiny 0.086}}$ & 0.28$_\pm{\text{\tiny 0.012}}$ \\
N & 1{,}000 & 5 & 0.55$_\pm{\text{\tiny 0.016}}$ & 0.56$_\pm{\text{\tiny 0.012}}$ & \textbf{0.5$_\pm{\text{\tiny 0.013}}$} & 0.55$_\pm{\text{\tiny 0.016}}$ & 0.89$_\pm{\text{\tiny 0.016}}$ & 0.89$_\pm{\text{\tiny 0.014}}$ \\
N & 10{,}000 & 2 & \textbf{0.19$_\pm{\text{\tiny 0.025}}$} & \textbf{0.19$_\pm{\text{\tiny 0.028}}$} & \textbf{0.2$_\pm{\text{\tiny 0.035}}$} & 0.3$_\pm{\text{\tiny 0.012}}$ & 0.84$_\pm{\text{\tiny 0.074}}$ & 0.6$_\pm{\text{\tiny 0.11}}$ \\
N & 10{,}000 & 5 & \textbf{0.71$_\pm{\text{\tiny 0.12}}$} & \textbf{0.69$_\pm{\text{\tiny 0.11}}$} & \textbf{0.74$_\pm{\text{\tiny 0.13}}$} & \textbf{0.69$_\pm{\text{\tiny 0.11}}$} & 0.98$_\pm{\text{\tiny 0.053}}$ & 0.97$_\pm{\text{\tiny 0.058}}$ \\
HetN & 1{,}000 & 2 & \textbf{0.24$_\pm{\text{\tiny 0.025}}$} & \textbf{0.25$_\pm{\text{\tiny 0.026}}$} & \textbf{0.24$_\pm{\text{\tiny 0.025}}$} & 0.52$_\pm{\text{\tiny 0.073}}$ & 0.75$_\pm{\text{\tiny 0.069}}$ & 0.36$_\pm{\text{\tiny 0.019}}$ \\
HetN & 1{,}000 & 5 & \textbf{0.66$_\pm{\text{\tiny 0.037}}$} & \textbf{0.65$_\pm{\text{\tiny 0.032}}$} & \textbf{0.63$_\pm{\text{\tiny 0.032}}$} & 0.71$_\pm{\text{\tiny 0.04}}$ & 0.99$_\pm{\text{\tiny 0.03}}$ & 0.99$_\pm{\text{\tiny 0.029}}$ \\
HetN & 10{,}000 & 2 & \textbf{0.28$_\pm{\text{\tiny 0.058}}$} & \textbf{0.28$_\pm{\text{\tiny 0.057}}$} & \textbf{0.31$_\pm{\text{\tiny 0.066}}$} & 0.52$_\pm{\text{\tiny 0.16}}$ & 1$_\pm{\text{\tiny 0.12}}$ & 0.75$_\pm{\text{\tiny 0.12}}$ \\
HetN & 10{,}000 & 5 & \textbf{0.81$_\pm{\text{\tiny 0.13}}$} & \textbf{0.79$_\pm{\text{\tiny 0.12}}$} & \textbf{0.85$_\pm{\text{\tiny 0.14}}$} & \textbf{0.83$_\pm{\text{\tiny 0.12}}$} & 1.1$_\pm{\text{\tiny 0.059}}$ & 1.1$_\pm{\text{\tiny 0.065}}$ \\
t & 1{,}000 & 2 & \textbf{0.21$_\pm{\text{\tiny 0.0028}}$} & \textbf{0.21$_\pm{\text{\tiny 0.0032}}$} & \textbf{0.21$_\pm{\text{\tiny 0.0036}}$} & 0.38$_\pm{\text{\tiny 0.045}}$ & 0.69$_\pm{\text{\tiny 0.067}}$ & 0.29$_\pm{\text{\tiny 0.0085}}$ \\
t & 1{,}000 & 5 & 0.55$_\pm{\text{\tiny 0.015}}$ & \textbf{0.53$_\pm{\text{\tiny 0.0095}}$} & \textbf{0.52$_\pm{\text{\tiny 0.011}}$} & 0.58$_\pm{\text{\tiny 0.015}}$ & 0.94$_\pm{\text{\tiny 0.02}}$ & 0.93$_\pm{\text{\tiny 0.018}}$ \\
t & 10{,}000 & 2 & \textbf{0.18$_\pm{\text{\tiny 0.026}}$} & \textbf{0.18$_\pm{\text{\tiny 0.027}}$} & \textbf{0.21$_\pm{\text{\tiny 0.037}}$} & 0.29$_\pm{\text{\tiny 0.015}}$ & 0.87$_\pm{\text{\tiny 0.076}}$ & 0.63$_\pm{\text{\tiny 0.12}}$ \\
t & 10{,}000 & 5 & \textbf{0.72$_\pm{\text{\tiny 0.12}}$} & \textbf{0.68$_\pm{\text{\tiny 0.11}}$} & \textbf{0.76$_\pm{\text{\tiny 0.13}}$} & \textbf{0.71$_\pm{\text{\tiny 0.11}}$} & 1$_\pm{\text{\tiny 0.054}}$ & 1$_\pm{\text{\tiny 0.059}}$ \\
\bottomrule
\end{tabular}
\caption{GP results for simulated data: Average quantile test RMSE $\pm$ standard error.
Noise models are homoscedastic Gaussian (N), heteroscedastic Gaussian (HetN), and Student-$t$ with 2 degrees of freedom ($t$).
Boldface indicates methods within two standard errors of the minimum RMSE in each row.}
\label{tab:gp_results_full}
\end{table}

\begin{table}[htbp]
\centering
\begin{tabular}{lcccccc@{\hspace{6pt}}|@{\hspace{6pt}}cccc}
\toprule
\multicolumn{3}{c}{} & \multicolumn{4}{c}{Lengthscale MSE} & \multicolumn{4}{c}{Signal Variance MSE} \\
\cmidrule(lr){4-11}
Noise & $n$ & d & \multicolumn{1}{c}{TKC} & \multicolumn{1}{c}{FL} & \multicolumn{1}{c}{VI} & \multicolumn{1}{c}{VIVA} & \multicolumn{1}{c}{TKC} & \multicolumn{1}{c}{FL} & \multicolumn{1}{c}{VI} & \multicolumn{1}{c}{VIVA} \\
\midrule
N & 1{,}000 & 2 & 0.0041 & 0.0043 & \textbf{0.0017} & 0.025 & 0.048 & 0.084 & \textbf{0.014} & 0.11 \\
N & 1{,}000 & 5 & 0.017 & \textbf{0.0043} & 0.031 & 0.024 & 0.11 & 0.047 & \textbf{0.016} & 0.066 \\
N & 10{,}000 & 2 & 0.024 & 0.022 & \textbf{0.011} & 0.046 & \textbf{0.052} & 0.068 & 0.095 & 0.18 \\
N & 10{,}000 & 5 & \textbf{0.005} & 0.0096 & 0.21 & 0.0088 & 0.059 & 0.044 & 0.11 & \textbf{0.035} \\
HetN & 1{,}000 & 2 & 0.0014 & 0.0051 & \textbf{6.9e-04} & 0.026 & 0.044 & 0.33 & \textbf{0.031} & 0.1 \\
HetN & 1{,}000 & 5 & 0.039 & \textbf{0.0054} & 0.039 & 0.028 & 0.13 & 0.037 & \textbf{0.016} & 0.075 \\
HetN & 10{,}000 & 2 & 0.017 & 0.022 & \textbf{0.012} & 0.041 & 0.61 & 1.4 & \textbf{0.096} & 0.22 \\
HetN & 10{,}000 & 5 & 0.029 & 0.012 & 0.21 & \textbf{0.009} & 0.089 & 0.033 & 0.11 & \textbf{0.024} \\
t & 1{,}000 & 2 & \textbf{0.001} & 0.0036 & 0.0012 & 0.027 & 0.07 & 0.096 & \textbf{0.028} & 0.09 \\
t & 1{,}000 & 5 & 0.012 & \textbf{0.0035} & 0.033 & 0.026 & 0.11 & 0.051 & \textbf{0.015} & 0.064 \\
t & 10{,}000 & 2 & 0.019 & 0.021 & \textbf{0.011} & 0.045 & \textbf{0.045} & 0.06 & 0.094 & 0.19 \\
t & 10{,}000 & 5 & 0.032 & 0.01 & 0.21 & \textbf{0.0089} & 0.076 & 0.057 & 0.11 & \textbf{0.033} \\
\bottomrule
\end{tabular}
\caption{\textbf{Gaussian process: Hyperparameter MSE}. Mean squared error for lengthscale and signal variance across noise types, sample sizes, and dimensions. Bold indicates lowest MSE.}
\label{tab:hyperparameter_mse}
\end{table}

\label{app:gp_sim_runtime}
\begin{table}[htbp]
\centering
\begin{tabular}{lllcccccc}
\toprule
\multicolumn{3}{c}{} & \multicolumn{6}{c}{Runtime} \\
\cmidrule(lr){4-9}
Noise & N & d & TKC & FL & VI & VIVA & QGAM & QGAM Int. \\
\midrule
N & 1{,}000 & 2 & 4.8$_\pm{\text{\tiny 0.72}}$ & 3.3$_\pm{\text{\tiny 0.17}}$ & 45.9$_\pm{\text{\tiny 4.1}}$ & 56.3$_\pm{\text{\tiny 2.2}}$ & 0.59$_\pm{\text{\tiny 0.066}}$ & \textbf{0.47$_\pm{\text{\tiny 0.013}}$} \\
N & 1{,}000 & 5 & 5$_\pm{\text{\tiny 0.6}}$ & 4.2$_\pm{\text{\tiny 0.75}}$ & 82.3$_\pm{\text{\tiny 12}}$ & 133$_\pm{\text{\tiny 8.2}}$ & 4.5$_\pm{\text{\tiny 0.31}}$ & \textbf{1.6$_\pm{\text{\tiny 0.14}}$} \\
N & 10{,}000 & 2 & 162$_\pm{\text{\tiny 16.3}}$ & 192$_\pm{\text{\tiny 20.2}}$ & 673$_\pm{\text{\tiny 25.1}}$ & 653$_\pm{\text{\tiny 27.2}}$ & 6.2$_\pm{\text{\tiny 0.97}}$ & \textbf{5.5$_\pm{\text{\tiny 0.3}}$} \\
N & 10{,}000 & 5 & 4622$_\pm{\text{\tiny 566}}$ & 2410$_\pm{\text{\tiny 257}}$ & 545$_\pm{\text{\tiny 31.6}}$ & 1602$_\pm{\text{\tiny 45.2}}$ & 34.4$_\pm{\text{\tiny 3.7}}$ & \textbf{12$_\pm{\text{\tiny 0.85}}$} \\
HetN & 1{,}000 & 2 & 5.5$_\pm{\text{\tiny 0.43}}$ & 5.1$_\pm{\text{\tiny 0.6}}$ & 56.9$_\pm{\text{\tiny 6.2}}$ & 53.6$_\pm{\text{\tiny 2.5}}$ & 1.8$_\pm{\text{\tiny 0.4}}$ & \textbf{0.8$_\pm{\text{\tiny 0.16}}$} \\
HetN & 1{,}000 & 5 & 5$_\pm{\text{\tiny 0.41}}$ & 3.6$_\pm{\text{\tiny 0.46}}$ & 58.2$_\pm{\text{\tiny 3.7}}$ & 103$_\pm{\text{\tiny 3.6}}$ & 14.9$_\pm{\text{\tiny 1.3}}$ & \textbf{1.7$_\pm{\text{\tiny 0.28}}$} \\
HetN & 10{,}000 & 2 & 156$_\pm{\text{\tiny 11.8}}$ & 181$_\pm{\text{\tiny 17.2}}$ & 649$_\pm{\text{\tiny 5.4}}$ & 638$_\pm{\text{\tiny 11.2}}$ & \textbf{4.6$_\pm{\text{\tiny 0.35}}$} & 6.2$_\pm{\text{\tiny 0.37}}$ \\
HetN & 10{,}000 & 5 & 4795$_\pm{\text{\tiny 486}}$ & 2821$_\pm{\text{\tiny 155}}$ & 554$_\pm{\text{\tiny 35.9}}$ & 1595$_\pm{\text{\tiny 41.7}}$ & 26.4$_\pm{\text{\tiny 3.4}}$ & \textbf{13.7$_\pm{\text{\tiny 0.84}}$} \\
t & 1{,}000 & 2 & 5.5$_\pm{\text{\tiny 0.23}}$ & 4.7$_\pm{\text{\tiny 0.32}}$ & 58$_\pm{\text{\tiny 4.8}}$ & 54.7$_\pm{\text{\tiny 2.3}}$ & 0.9$_\pm{\text{\tiny 0.13}}$ & \textbf{0.56$_\pm{\text{\tiny 0.026}}$} \\
t & 1{,}000 & 5 & 5.6$_\pm{\text{\tiny 0.94}}$ & 3.9$_\pm{\text{\tiny 0.53}}$ & 55.9$_\pm{\text{\tiny 5.1}}$ & 100$_\pm{\text{\tiny 3.1}}$ & 7.6$_\pm{\text{\tiny 1.9}}$ & \textbf{1.6$_\pm{\text{\tiny 0.37}}$} \\
t & 10{,}000 & 2 & 152$_\pm{\text{\tiny 14.8}}$ & 178$_\pm{\text{\tiny 19}}$ & 630$_\pm{\text{\tiny 16.3}}$ & 643$_\pm{\text{\tiny 22.8}}$ & \textbf{8.4$_\pm{\text{\tiny 1.5}}$} & \textbf{10.1$_\pm{\text{\tiny 1.7}}$} \\
t & 10{,}000 & 5 & 4179$_\pm{\text{\tiny 556}}$ & 2632$_\pm{\text{\tiny 296}}$ & 547$_\pm{\text{\tiny 39.6}}$ & 1621$_\pm{\text{\tiny 61.5}}$ & 53.1$_\pm{\text{\tiny 7.5}}$ & \textbf{11.5$_\pm{\text{\tiny 0.51}}$} \\
\bottomrule
\end{tabular}
\caption{
\textbf{Gaussian process: Runtimes (s).}
}
\end{table}

\clearpage
\subsubsection{Approximation of the Log-Marginal Likelihood}\label{marg_ll_accuracy}

We next analyze how well the proposed Laplace approximations estimate the log-marginal likelihood $\log p(\mathbf{y}|\boldsymbol{\theta}, \boldsymbol{\beta},\lambda) = \text{log} \int 
p(\mathbf{y}|\mathbf{b}, \boldsymbol{\beta},\lambda)\pi(\mathbf{b}|\boldsymbol{\theta})d \mathbf{b}$. Since this quantity is not analytically available, we consider a single-level grouped random effects model and use adaptive Gauss-Hermite quadrature \citep{liu1994note} to compute the ground truth log-marginal likelihood.\footnote{We do not perform this comparison for Gaussian process models, as the corresponding integrals are over an $n$-dimensional latent space, making exact quadrature infeasible for the sample sizes of interest. While previous work has employed sampling-based techniques such as thermodynamic integration or annealed importance sampling for datasets with $n \approx 500$ \citep{nickisch2008approximations}, these approaches do not scale to larger sample sizes. In our experiments, we observed that sampling from high-dimensional Gaussian process posteriors is computationally prohibitive and Monte-Carlo-based marginal likelihood approximations (thermodynamic integration and annealed importance sampling) have high variances (results not shown).} We compare this to a Fisher-Laplace approximation, a triangular kernel curvature Laplace approximation, and an alternative numerical quadrature method based on adaptive Gauss-Kronrod integration\footnote{As implemented in the \texttt{scipy.integrate.quad} function of the SciPy library.}, applied over fixed integration bounds. Specifically, we consider $m=20$ groups with group sizes $n_j\in\{100,200,500, 1{,}000\}$, set $\tau=0.8$, draw random effects with standard deviation $\sigma=1.0$, and repeat each setting over $K=50$ simulated datasets. We consider a correctly specified asymmetric Laplace likelihood with a scale parameter $\lambda=1.0$ and misspecified Gaussian noise with variance $1$.

Figure~\ref{fig:numerical_combined} shows the differences between the approximations and adaptive-GH quadrature. We find that in the correctly specified scenario, both the Fisher-Laplace and the TKC Laplace approximations are accurate and converge to the correct quantity. In the misspecified scenario, as expected, the Fisher-Laplace approximation performs worse and systematically overestimates the negative log-marginal likelihood, whereas the TKC Laplace approximation converges to the correct quantity. The adaptive Gauss-Kronrod integration is accurate for small sample sizes $n_j$, but for larger sample sizes, the software implementation returns missing values since the integrand becomes sharply peaked around its mode.
\begin{figure}[ht!]
    \centering
    \includegraphics[width=1\linewidth]{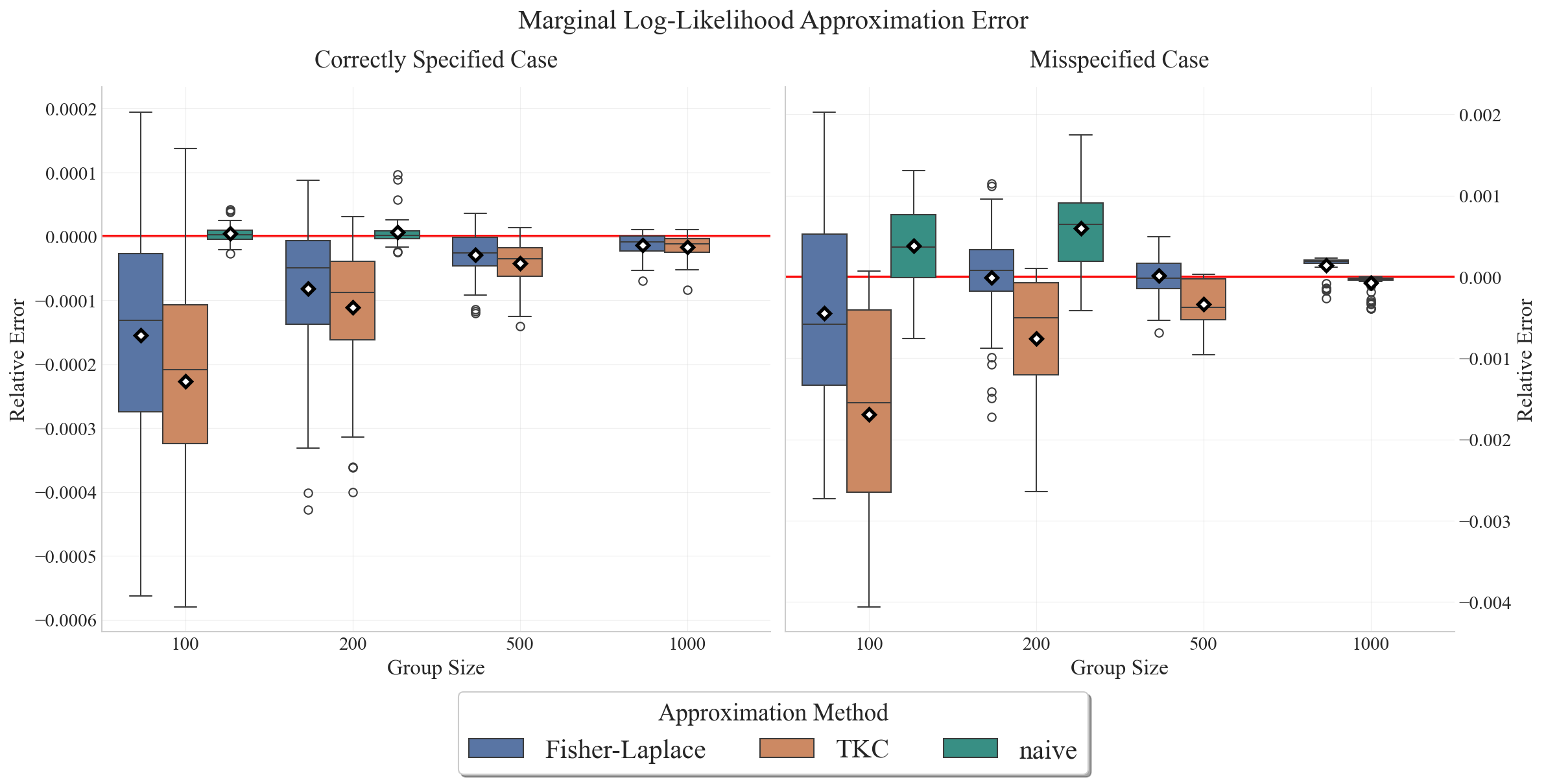}
    \caption{ Accuracy of log-marginal likelihood approximations, measured by the relative error with respect to adaptive Gauss--Hermite quadrature, across different sample sizes.
    Left: correctly specified likelihood. Right: misspecified likelihood.}
    \label{fig:numerical_combined}
\end{figure}

\clearpage
\subsection{Sensitivity Analysis for TKC's minimum likelihood drop parameter}
\label{app:sensitivity_tkc}
We use the same simulation setup as in Section~\ref{sec:simulated_exp} and study the sensitivity of TKC to the minimum likelihood-drop threshold. In particular, we compare the values \(10^{-2}, 3\times 10^{-2}, 10^{-1}, 3\times 10^{-1}, 10^{0}, 10^{1}, 10^{2}, 10^{3}, 10^{4}, 10^{5}\).
\subsubsection{Single-level Grouped Random Effects: $\tau = 0.8$}
\begin{figure}[ht!]
    \centering
    \includegraphics[width=1\linewidth]{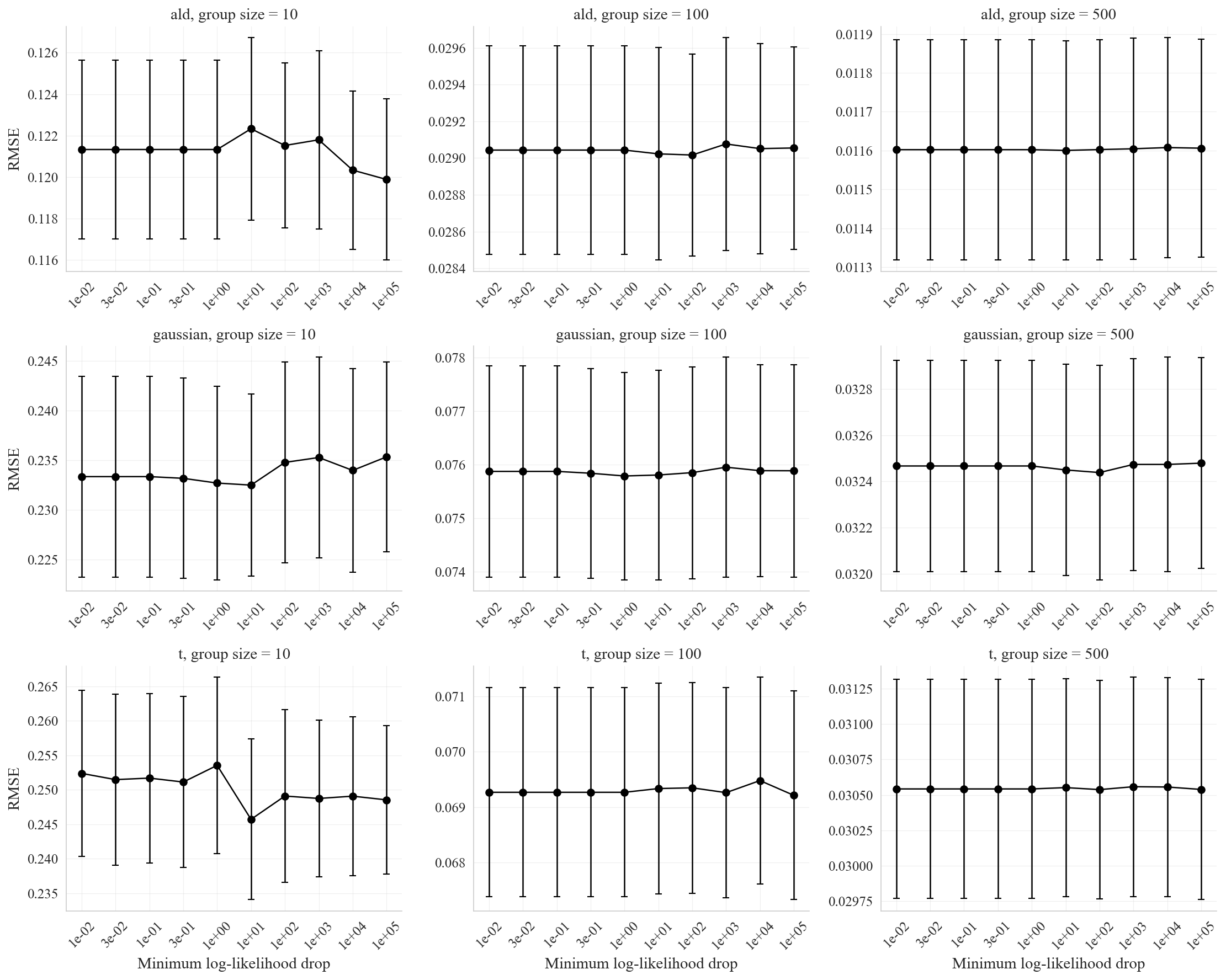}
    \caption{RMSE of the quantile for various minimum likelihood-drop thresholds.}
\end{figure}

\begin{figure}[ht!]
    \centering
    \includegraphics[width=1\linewidth]{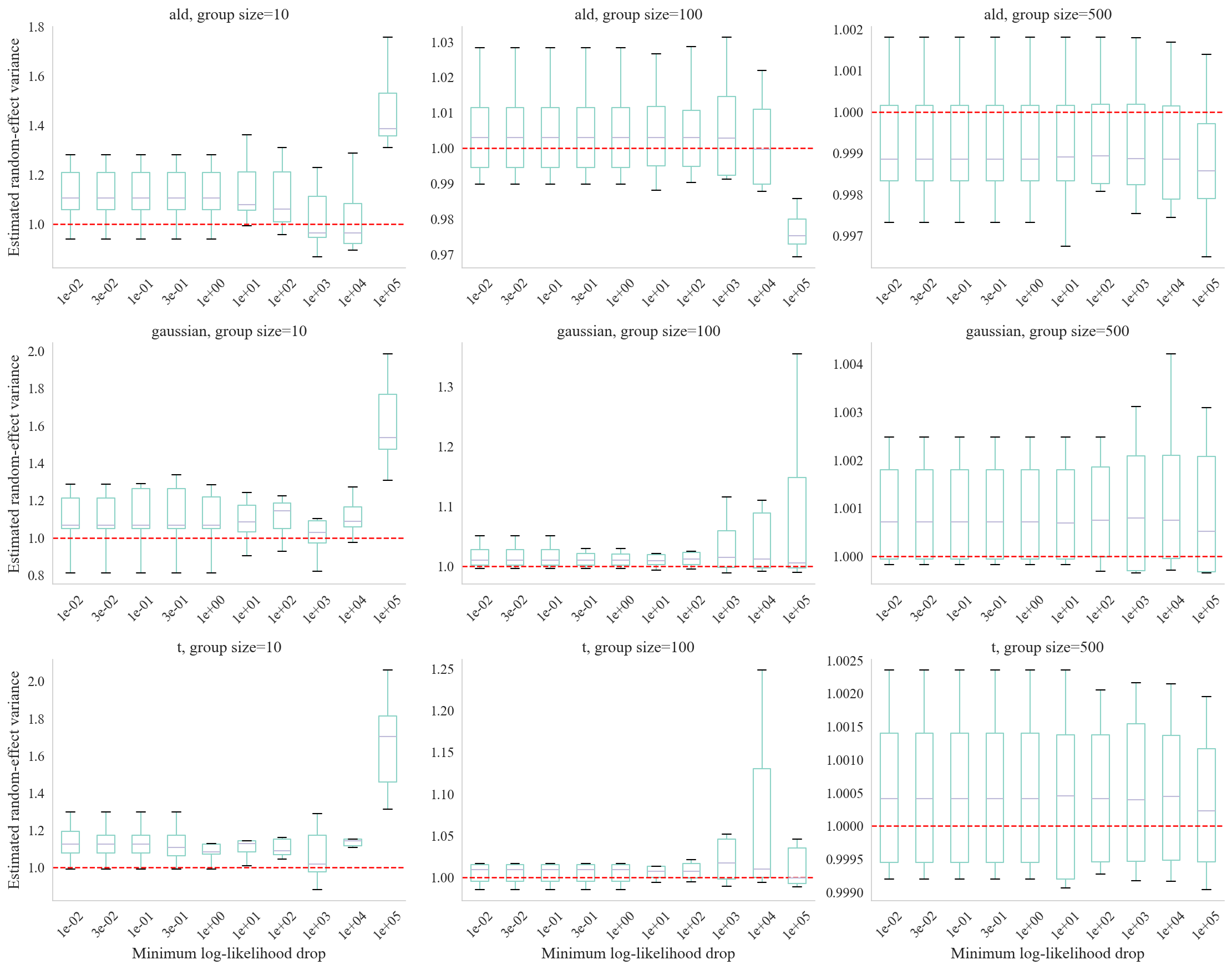}
    \caption{Estimated hyperparameter for various minimum likelihood-drop thresholds.}
\end{figure}

\clearpage
\subsubsection{Crossed Random Effects: $\tau = 0.8$}

\begin{figure}[ht!]
    \centering
    \includegraphics[width=1\linewidth]{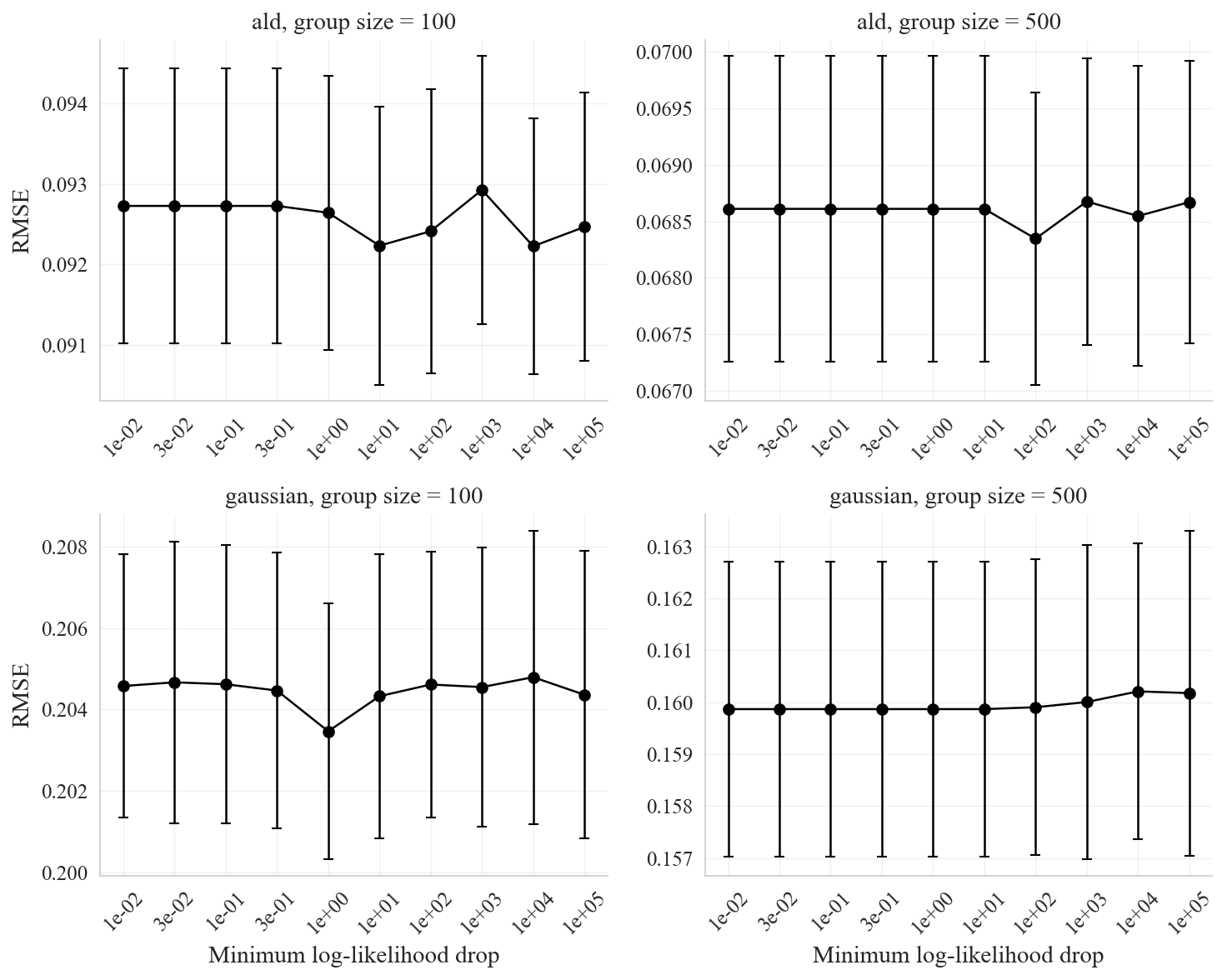}
    \caption{RMSE of the quantile for various minimum likelihood-drop thresholds}
\end{figure}

\begin{figure}[ht!]
    \centering
    \includegraphics[width=1\linewidth]{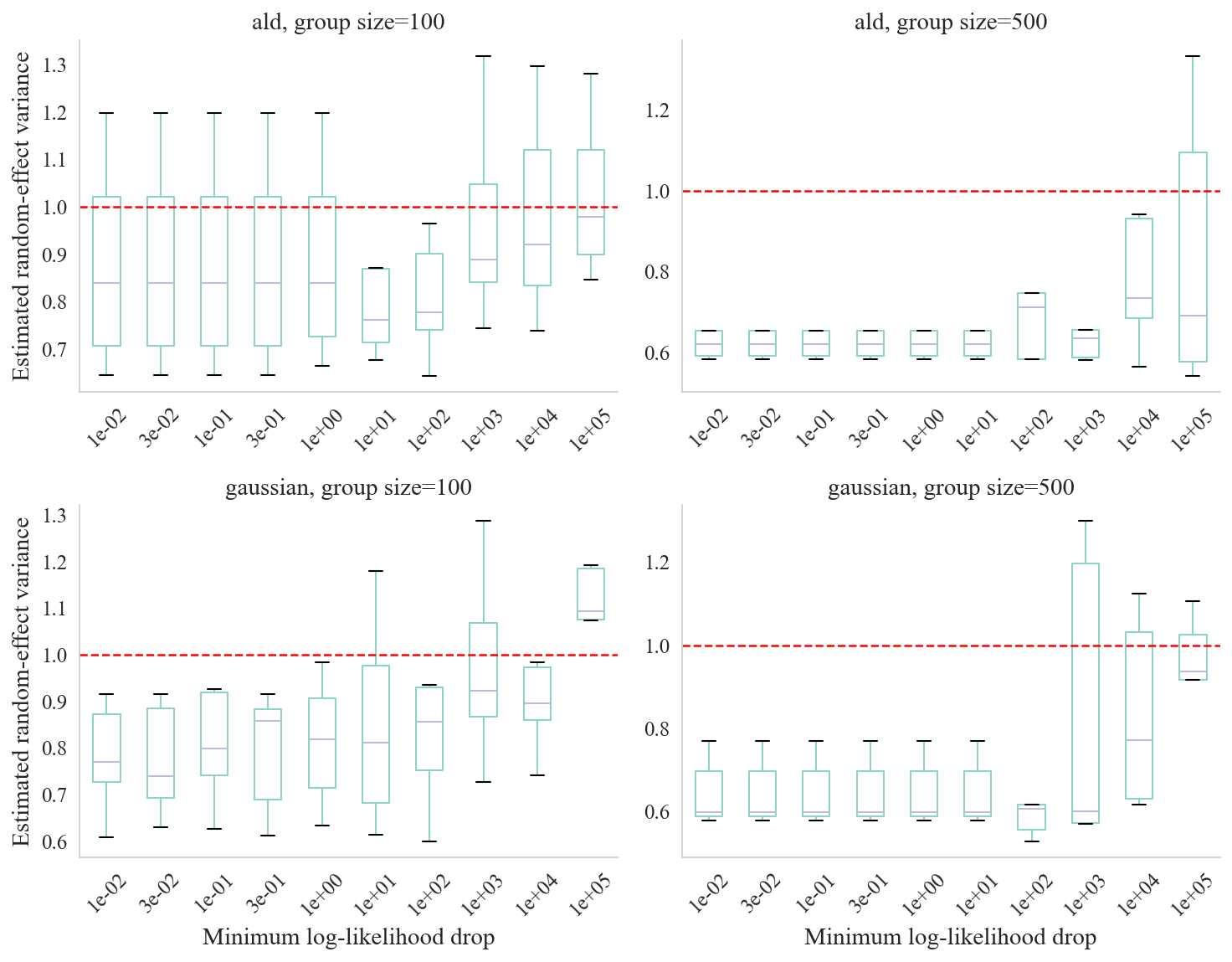}
    \caption{Estimated hyperparameter (variance of first random effect)  for various minimum likelihood-drop thresholds.}
\end{figure}

\begin{figure}[ht!]
    \centering
    \includegraphics[width=1\linewidth]{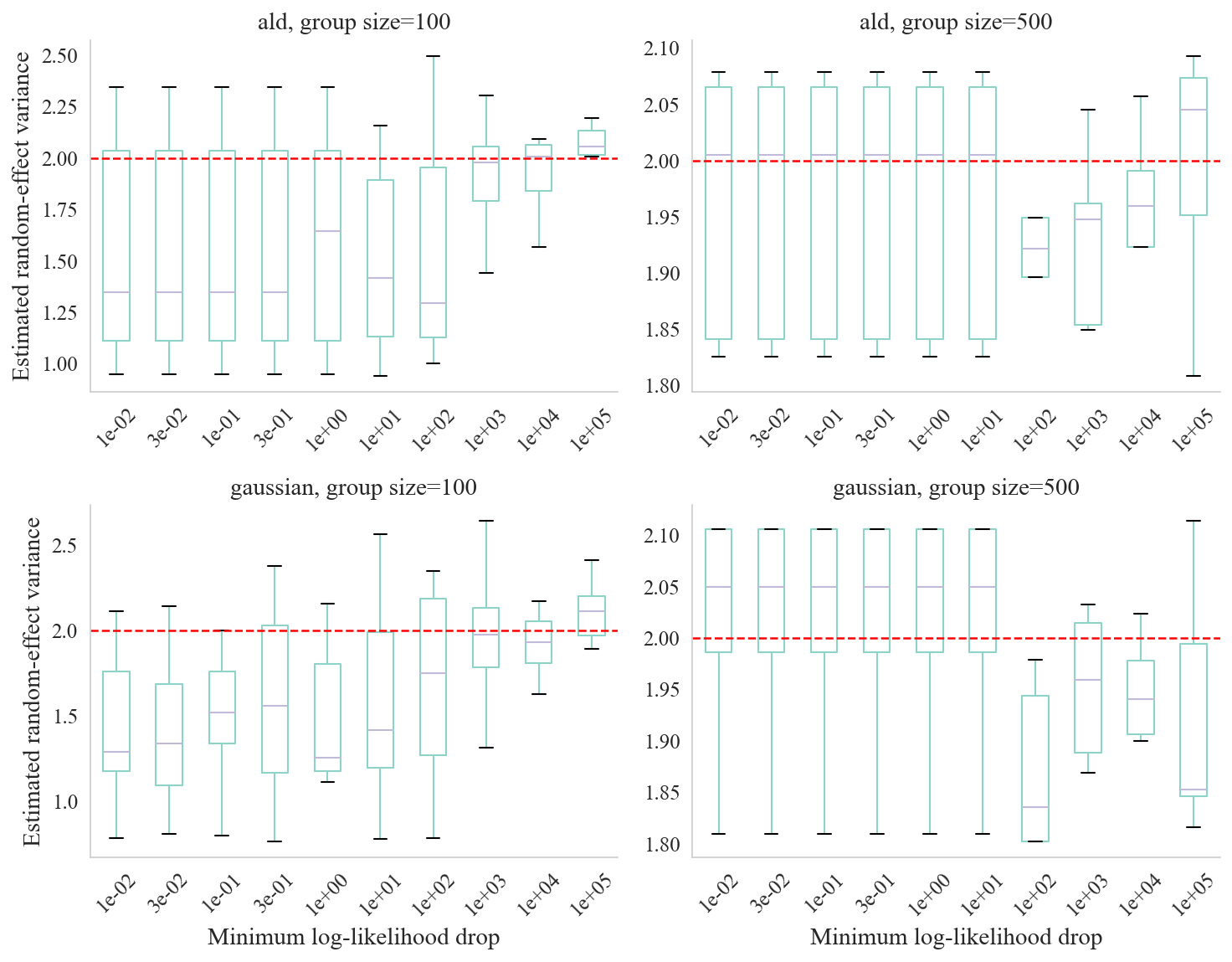}
    \caption{Estimated hyperparameter (variance of second random effect)  for various minimum likelihood-drop thresholds.}
\end{figure}

\clearpage
\subsubsection{Gaussian process: $\tau = 0.8$}
For the Gaussian process experiments, we present the sensitivity results for $n = 1{,}000$. The results are qualitatively similar for $n = 10{,}000$ (results not shown).
\begin{figure}[ht!]
    \centering
    \includegraphics[width=0.8\linewidth]{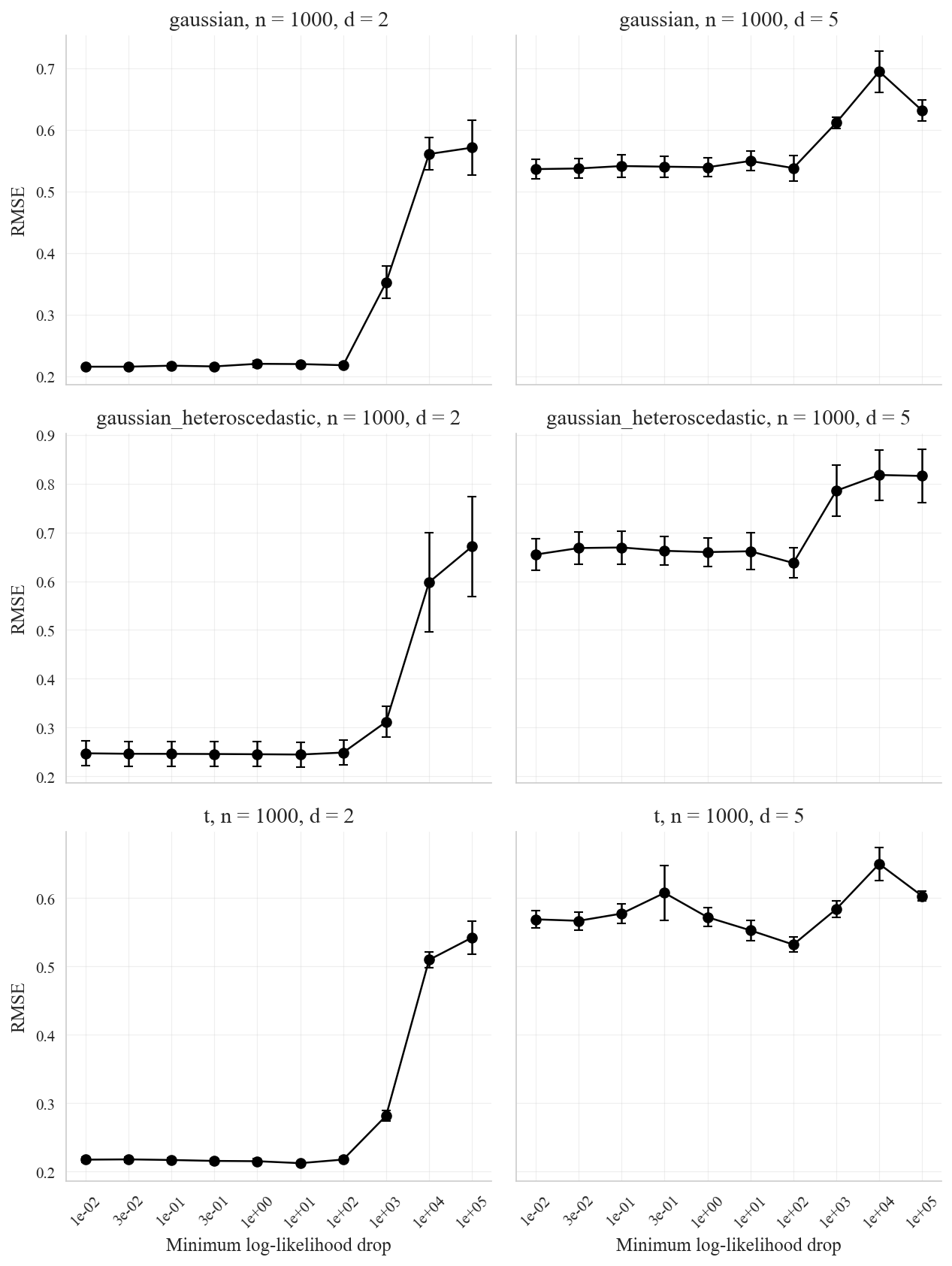}
    \caption{RMSE of the quantile for various minimum likelihood-drop thresholds.}
\end{figure}

\begin{figure}[ht!]
    \centering
    \includegraphics[width=1\linewidth]{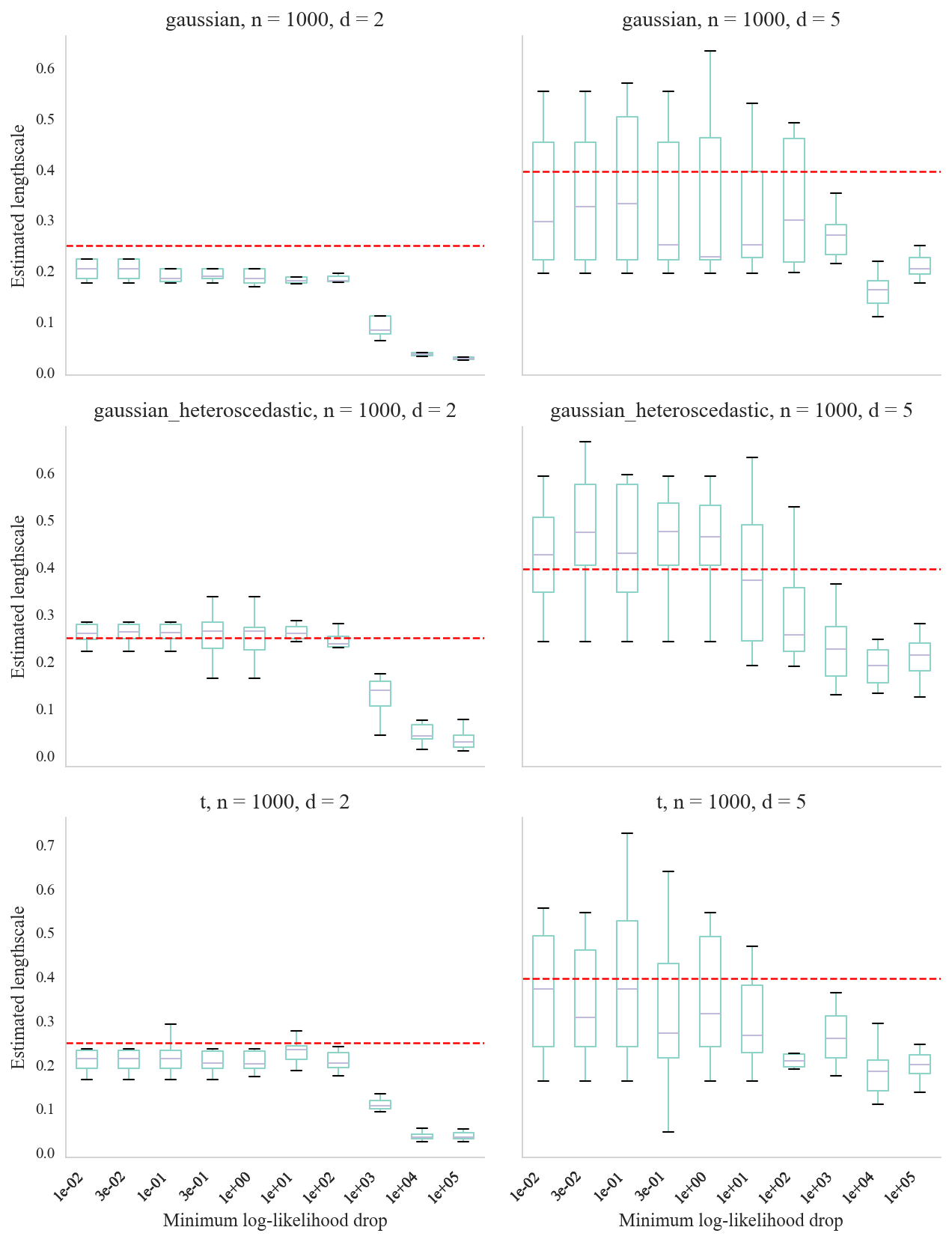}
    \caption{Estimated hyperparameter (length-scale) for various minimum likelihood-drop thresholds.}
\end{figure}

\begin{figure}[ht!]
    \centering
    \includegraphics[width=1\linewidth]{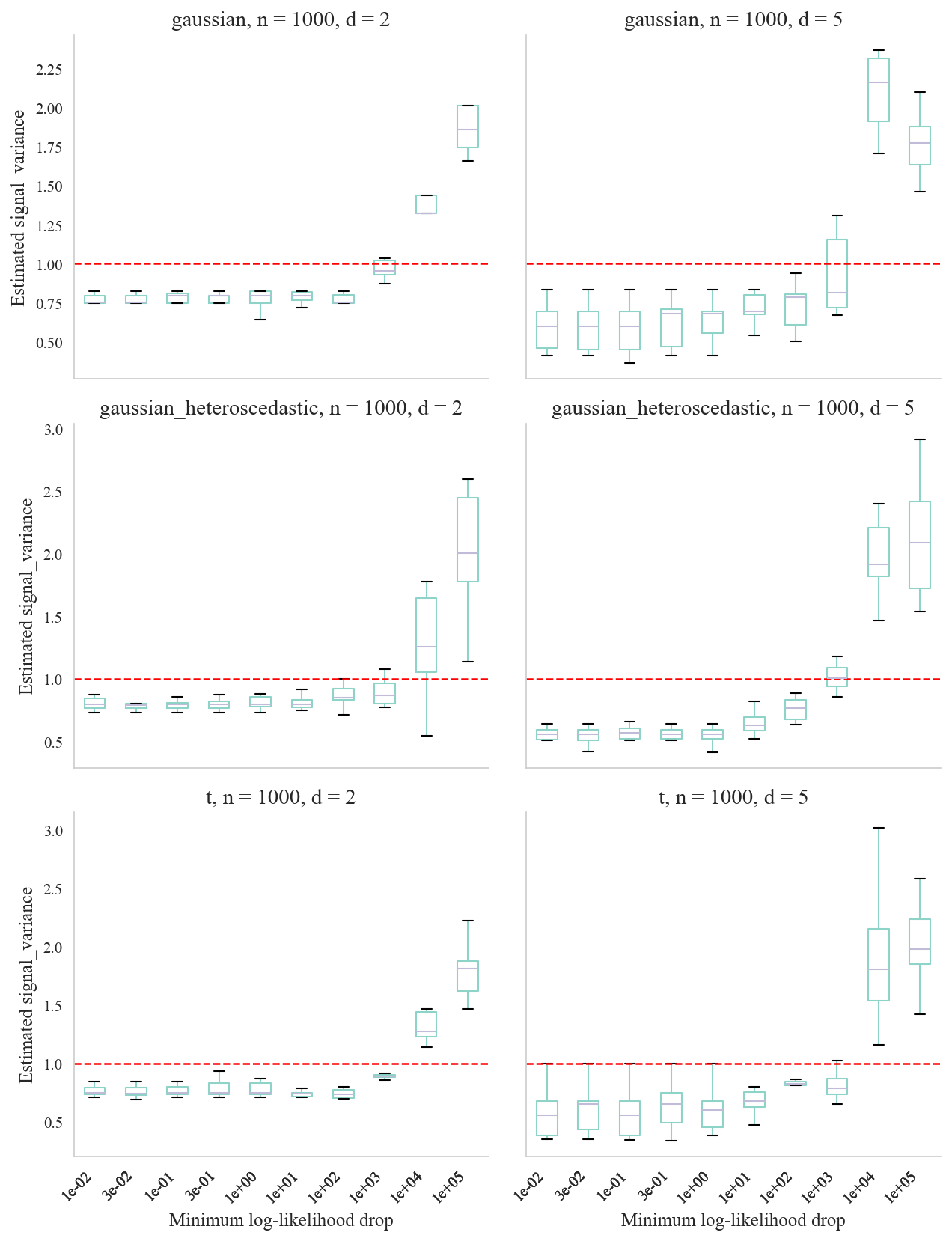}
    \caption{Estimated hyperparameter (signal variance) for various minimum likelihood-drop thresholds.}
\end{figure}

\clearpage
\subsection{Additional information and results for the real world data experiments}
\subsubsection{Description of real world datasets}\label{desc_data}
For single-level grouped random effects, we consider three datasets: the Orthodont dataset \citep{10.1093/biomet/51.3-4.313} with $n = 108$ measurements across $m = 27$ subjects and $p=2$ fixed effects, tracking \textit{dental growth} over time; the Labor dataset \citep{davis1991semi} contains $n = 358$ observations nested within $m = 83$ labor groups, $p=2$ fixed effects, and \textit{pain} is the response variable; the Cars dataset \citep{simchoni2023integrating}, with response \textit{log-price (\$)} of cars, which contains $n = 97,296$ observations with $m = 15,226$ car models as grouping variable and $p=2$ fixed effects.

For Gaussian processes, we consider four spatial datasets with continuous responses and two-dimensional coordinate inputs. The Laegern dataset ($n = 237{,}286$) \citep{schneider2017mapping} encompasses measurements of plant functional traits of the Laegern temperate mixed forest in Switzerland. The response is the \textit{canopy height}. Next, we consider two \textit{land surface temperatures} MODIS satellite data sets. The first is denoted as Heaton ($n = 148{,}309$) since it was used in \cite{heaton2019case}, and the second denoted as MODIS ($n = 600{,}000$) \citep{gyger2025iterativemethodsfullscalegaussian}. Finally, we consider the House dataset ($n = 21{,}554$) available from the R package \texttt{spData} \citep{spdata}, where our response is the \textit{house price}.

\subsubsection{Single-level Grouped Random Effects}

\label{app:single_real_runtime}
\begin{table}[htbp]
\centering
\begin{tabular}{lccccc}
\toprule
\multicolumn{1}{c}{} & \multicolumn{5}{c}{Runtime} \\
\cmidrule(lr){2-6}
Dataset & \multicolumn{1}{c}{TKC} & \multicolumn{1}{c}{Fisher} & \multicolumn{1}{c}{BayesQR} & \multicolumn{1}{c}{BRMS} & \multicolumn{1}{c}{LQMM} \\
\midrule
Orthodont & 0.97$_\pm{\text{\tiny 0.75}}$ & 0.098$_\pm{\text{\tiny 0.06}}$ & 0.77$_\pm{\text{\tiny 0.0053}}$ & 157$_\pm{\text{\tiny 6.8}}$ & \textbf{0.017$_\pm{\text{\tiny 0.0029}}$} \\
Labor & 4.4$_\pm{\text{\tiny 3.7}}$ & 0.62$_\pm{\text{\tiny 0.76}}$ & 15.5$_\pm{\text{\tiny 0.44}}$ & 334$_\pm{\text{\tiny 10.4}}$ & \textbf{0.042$_\pm{\text{\tiny 0.014}}$} \\
Cars & \textbf{32.5$_\pm{\text{\tiny 16}}$} & \textbf{26.1$_\pm{\text{\tiny 6.5}}$} & --- & --- & --- \\
\bottomrule
\end{tabular}
\caption{\textbf{Runtime (s).}}
\end{table}

\subsubsection{Crossed Random Effects}\label{app:real_crossed_runtime}
We evaluate crossed random effects models using two datasets. First, we extend the Cars dataset with an additional grouping variable representing $m_2 = 12,235$ locations, resulting in a crossed random effects structure across car models and locations. Second, we consider the MovieLens 100k dataset \citep{harper2015movielens}, which contains $n=100,000$ user-movie ratings with $m_1= 943$ users and $m_2=1,682$ movies, where the response variable corresponds to the recorded \textit{rating score}.

Our Laplace approximations are compared against \texttt{brms} and \texttt{bayesQR}. \texttt{lqmm} does not support crossed random effects structures. The comparison employs 5-fold cross-validation. Table~\ref{tab:crossed_real} reports the quantile loss for the methods that successfully returned results. We find that both Laplace approximations (TKC and Fisher-Laplace) run reliably within the time limit, whereas \texttt{brms} and \texttt{bayesQR} did not finish within the allocated budget for both crossed-effects datasets. The runtimes are reported in Table \ref{tab:crossed_real_time}.

\begin{table}[htbp]
\centering
\begin{tabular}{lcccc}
\toprule
\multicolumn{1}{c}{} & \multicolumn{4}{c}{Quantile Loss} \\
\cmidrule(lr){2-5}
Dataset & \multicolumn{1}{c}{TKC} & \multicolumn{1}{c}{Fisher} & \multicolumn{1}{c}{BayesQR} & \multicolumn{1}{c}{BRMS} \\
\midrule
Cars & 0.13$_\pm{\text{\tiny 0.0026}}$ & \textbf{0.12$_\pm{\text{\tiny 0.002}}$} & --- & --- \\
Ratings & 0.21$_\pm{\text{\tiny 0.0019}}$ & \textbf{0.21$_\pm{\text{\tiny 9.1e-04}}$} & --- & --- \\
\bottomrule
\end{tabular}
\caption{\textbf{Crossed Random Effects:} Mean quantile loss $\pm$ standard error over 5-fold cross-validation.
Dashes (---) indicate that the method did not converge or failed to finish within the pre-specified time limit. Boldface indicates the minimum loss within two standard errors.}
\label{tab:crossed_real}
\end{table}

\begin{table}[htbp]
\centering
\begin{tabular}{lcccc}
\toprule
\multicolumn{1}{c}{} & \multicolumn{4}{c}{Runtime} \\
\cmidrule(lr){2-5}
Dataset & \multicolumn{1}{c}{TKC} & \multicolumn{1}{c}{Fisher} & \multicolumn{1}{c}{BayesQR} & \multicolumn{1}{c}{BRMS} \\
\midrule
Cars & 885$_\pm{\text{\tiny 0.33}}$ & \textbf{335$_\pm{\text{\tiny 80.8}}$} & --- & --- \\
Ratings & 47.8$_\pm{\text{\tiny 3.6}}$ & \textbf{37.4$_\pm{\text{\tiny 3.8}}$} & --- & --- \\
\bottomrule
\end{tabular}
\caption{\textbf{Runtime (s).}}
\label{tab:crossed_real_time}
\end{table}

\clearpage
\subsubsection{Gaussian process}\label{app:real_GP_runtime}
\begin{table}[htbp]
\centering
\begin{tabular}{lcccccc}
\toprule
\multicolumn{1}{c}{} & \multicolumn{6}{c}{Runtime} \\
\cmidrule(lr){2-7}
Dataset & TKC & FL & VI & VIVA & QGAM & QGAM Int. \\
\midrule
Heaton & 966$_\pm$\text{\tiny 286} & 926$_\pm$\text{\tiny 162} & 2412$_\pm$\text{\tiny 9.7} & 3153$_\pm$\text{\tiny 149} & 4.1$_\pm$\text{\tiny 0.45} & \textbf{2.7$_\pm$\text{\tiny 0.14}} \\
House & 312$_\pm$\text{\tiny 113} & 641$_\pm$\text{\tiny 187} & 2412$_\pm$\text{\tiny 2.6} & 3309$_\pm$\text{\tiny 82.5} & \textbf{5.1$_\pm$\text{\tiny 2.1}} & \textbf{6$_\pm$\text{\tiny 0.86}} \\
Laegern & 837$_\pm$\text{\tiny 417} & 1054$_\pm$\text{\tiny 402} & 2932$_\pm$\text{\tiny 23.6} & 3136$_\pm$\text{\tiny 158} & \textbf{4.3$_\pm$\text{\tiny 1.5}} & \textbf{4.7$_\pm$\text{\tiny 1.2}} \\
MODIS & 829$_\pm$\text{\tiny 65.7} & 1097$_\pm$\text{\tiny 410} & 2402$_\pm$\text{\tiny 19.7} & 3104$_\pm$\text{\tiny 358} & \textbf{4.3$_\pm$\text{\tiny 2}} & \textbf{3.4$_\pm$\text{\tiny 0.81}} \\
\bottomrule
\end{tabular}
\caption{Gaussian process: Runtime (seconds).}
\label{table:gp_runtime_real}
\end{table}

\clearpage
\section{Additional results: experiments for the $\tau = 0.95$ quantile}
\label{app:tau095}
Best performance ($\pm$ 2 std. err.) in bold. --- indicates method not applicable or convergence failure.
\vspace{-0.25cm}
\subsection{Simulation experiments}
\subsubsection{Single-level Grouped Random Effects}
\begin{table}[htbp]
\centering
\begin{tabular}{lccccccc}
\toprule
\multicolumn{3}{c}{} & \multicolumn{5}{c}{RMSE} \\
\cmidrule(lr){4-8}
Noise & $m$ & $n_j$ & \multicolumn{1}{c}{TKC} & \multicolumn{1}{c}{Fisher} & \multicolumn{1}{c}{BayesQR} & \multicolumn{1}{c}{LQMM} & \multicolumn{1}{c}{BRMS} \\
\midrule
ALD & 100 & 10 & 0.23$_\pm{\text{\tiny 0.015}}$ & 0.26$_\pm{\text{\tiny 0.017}}$ & \textbf{0.14$_\pm{\text{\tiny 0.0033}}$} & 0.5$_\pm{\text{\tiny 0.0059}}$ & 0.21$_\pm{\text{\tiny 0.0093}}$ \\
ALD & 100 & 100 & 0.057$_\pm{\text{\tiny 0.0012}}$ & 0.057$_\pm{\text{\tiny 0.0012}}$ & \textbf{0.055$_\pm{\text{\tiny 0.0011}}$} & --- & 0.51$_\pm{\text{\tiny 0.17}}$ \\
ALD & 100 & 500 & \textbf{0.025$_\pm{\text{\tiny 3.2e-04}}$} & \textbf{0.025$_\pm{\text{\tiny 3.2e-04}}$} & \textbf{0.024$_\pm{\text{\tiny 2.7e-04}}$} & --- & --- \\
N & 100 & 10 & 0.35$_\pm{\text{\tiny 0.01}}$ & 0.35$_\pm{\text{\tiny 0.01}}$ & \textbf{0.29$_\pm{\text{\tiny 0.0068}}$} & 0.82$_\pm{\text{\tiny 0.0085}}$ & \textbf{0.3$_\pm{\text{\tiny 0.1}}$} \\
N & 100 & 100 & \textbf{0.11$_\pm{\text{\tiny 0.0021}}$} & \textbf{0.11$_\pm{\text{\tiny 0.0019}}$} & \textbf{0.11$_\pm{\text{\tiny 0.002}}$} & 0.87$_\pm{\text{\tiny 0.0019}}$ & 0.25$_\pm{\text{\tiny 0.1}}$ \\
N & 100 & 500 & \textbf{0.048$_\pm{\text{\tiny 9.3e-04}}$} & \textbf{0.048$_\pm{\text{\tiny 9.3e-04}}$} & \textbf{0.049$_\pm{\text{\tiny 7.1e-04}}$} & --- & --- \\
t & 100 & 10 & 0.5$_\pm{\text{\tiny 0.029}}$ & 0.51$_\pm{\text{\tiny 0.033}}$ & 0.51$_\pm{\text{\tiny 0.047}}$ & 0.94$_\pm{\text{\tiny 0.027}}$ & \textbf{0.24$_\pm{\text{\tiny 0.016}}$} \\
t & 100 & 100 & \textbf{0.18$_\pm{\text{\tiny 0.0069}}$} & \textbf{0.18$_\pm{\text{\tiny 0.0071}}$} & \textbf{0.18$_\pm{\text{\tiny 0.0071}}$} & 1.1$_\pm{\text{\tiny 3.0e-04}}$ & 0.39$_\pm{\text{\tiny 0.0097}}$ \\
t & 100 & 500 & \textbf{0.076$_\pm{\text{\tiny 0.0017}}$} & \textbf{0.077$_\pm{\text{\tiny 0.0019}}$} & --- & --- & --- \\
\bottomrule
\end{tabular}
\caption{\textbf{RMSE for quantile predictions.} }
\end{table}

\begin{table}[ht!]
\centering
\begin{tabular}{lccccccc}
\toprule
\multicolumn{3}{c}{} & \multicolumn{5}{c}{Runtime} \\
\cmidrule(lr){4-8}
Noise & $m$ & $n_j$ & \multicolumn{1}{c}{TKC} & \multicolumn{1}{c}{Fisher} & \multicolumn{1}{c}{BayesQR} & \multicolumn{1}{c}{LQMM} & \multicolumn{1}{c}{BRMS} \\
\midrule
ALD & 100 & 10 & 0.24$_\pm{\text{\tiny 0.03}}$ & 0.17$_\pm{\text{\tiny 0.014}}$ & 56.3$_\pm{\text{\tiny 3.6}}$ & 517$_\pm{\text{\tiny 19.1}}$ & \textbf{0.072$_\pm{\text{\tiny 0.0063}}$} \\
ALD & 100 & 100 & 0.82$_\pm{\text{\tiny 0.069}}$ & \textbf{0.44$_\pm{\text{\tiny 0.056}}$} & 813$_\pm{\text{\tiny 102}}$ & --- & 0.59$_\pm{\text{\tiny 0.093}}$ \\
ALD & 100 & 500 & 1.6$_\pm{\text{\tiny 0.041}}$ & \textbf{0.51$_\pm{\text{\tiny 0.034}}$} & 2949$_\pm{\text{\tiny 41.3}}$ & --- & --- \\
N & 100 & 10 & 0.31$_\pm{\text{\tiny 0.033}}$ & 0.22$_\pm{\text{\tiny 0.022}}$ & 67.5$_\pm{\text{\tiny 0.77}}$ & 459$_\pm{\text{\tiny 10.4}}$ & \textbf{0.077$_\pm{\text{\tiny 0.0033}}$} \\
N & 100 & 100 & 1.1$_\pm{\text{\tiny 0.042}}$ & 0.93$_\pm{\text{\tiny 0.078}}$ & 697$_\pm{\text{\tiny 149}}$ & 2608$_\pm{\text{\tiny 106}}$ & \textbf{0.33$_\pm{\text{\tiny 0.02}}$} \\
N & 100 & 500 & 2.7$_\pm{\text{\tiny 0.13}}$ & \textbf{1.6$_\pm{\text{\tiny 0.12}}$} & 3264$_\pm{\text{\tiny 62.8}}$ & --- & --- \\
t & 100 & 10 & \textbf{0.45$_\pm{\text{\tiny 0.09}}$} & \textbf{0.38$_\pm{\text{\tiny 0.058}}$} & 277$_\pm{\text{\tiny 16.5}}$ & 1662$_\pm{\text{\tiny 70.9}}$ & \textbf{0.41$_\pm{\text{\tiny 0.029}}$} \\
t & 100 & 100 & 1$_\pm{\text{\tiny 0.078}}$ & \textbf{0.86$_\pm{\text{\tiny 0.086}}$} & 1399$_\pm{\text{\tiny 188}}$ & 3338$_\pm{\text{\tiny 95.4}}$ & \textbf{0.9$_\pm{\text{\tiny 0.13}}$} \\
t & 100 & 500 & 2.8$_\pm{\text{\tiny 0.12}}$ & \textbf{1.6$_\pm{\text{\tiny 0.17}}$} & --- & --- & --- \\
\bottomrule
\end{tabular}
\caption{\textbf{Runtime (s).}}
\end{table}

\begin{table}[ht!]
\centering
\begin{tabular}{lcccccc}
\toprule
\multicolumn{3}{c}{} & \multicolumn{4}{c}{MSE of Random Effect Variance (1)} \\
\cmidrule(lr){4-7}
Noise & $m$ & $n_j$ & \multicolumn{1}{c}{TKC} & \multicolumn{1}{c}{Fisher} & \multicolumn{1}{c}{BRMS} & \multicolumn{1}{c}{LQMM} \\
\midrule
ALD & 100 & 10 & \textbf{0.0083} & 0.096 & 0.023 & 0.17 \\
ALD & 100 & 100 & \textbf{0.001} & 0.0025 & --- & 0.76 \\
ALD & 100 & 500 & \textbf{3.5e-06} & 1.3e-05 & --- & --- \\
N & 100 & 10 & \textbf{0.0067} & 0.017 & 0.018 & 0.18 \\
N & 100 & 100 & \textbf{0.01} & 0.035 & --- & 0.34 \\
N & 100 & 500 & \textbf{0.01} & 0.079 & --- & --- \\
t & 100 & 10 & 0.084 & 0.22 & --- & \textbf{0.057} \\
t & 100 & 100 & \textbf{0.032} & 0.085 & --- & 0.31 \\
t & 100 & 500 & \textbf{0.027} & 0.073 & --- & --- \\

\bottomrule
\end{tabular}
\caption{\textbf{MSE of estimated variance component.}}
\end{table}

\newpage
\subsubsection{Crossed Random Effects}
\begin{table}[htbp]
\centering
\begin{tabular}{lcccccc}
\toprule
\multicolumn{3}{c}{} & \multicolumn{4}{c}{RMSE} \\
\cmidrule(lr){4-7}
Noise & $m$ & $n_j$ & \multicolumn{1}{c}{TKC} & \multicolumn{1}{c}{Fisher} & \multicolumn{1}{c}{BayesQR} & \multicolumn{1}{c}{BRMS} \\
\midrule
ALD & 100 & 100 & \textbf{0.019$_\pm{\text{\tiny 5.4e-04}}$} & \textbf{0.019$_\pm{\text{\tiny 5.4e-04}}$} & --- & --- \\
ALD & 100 & 500 & \textbf{0.007$_\pm{\text{\tiny 1.3e-04}}$} & \textbf{0.007$_\pm{\text{\tiny 1.2e-04}}$} & --- & --- \\
N & 100 & 100 & \textbf{0.13$_\pm{\text{\tiny 0.0035}}$} & \textbf{0.13$_\pm{\text{\tiny 0.0036}}$} & --- & --- \\
N & 100 & 500 & \textbf{0.06$_\pm{\text{\tiny 8.1e-04}}$} & \textbf{0.06$_\pm{\text{\tiny 8.2e-04}}$} & --- & --- \\
t & 100 & 100 & \textbf{0.21$_\pm{\text{\tiny 0.0044}}$} & \textbf{0.21$_\pm{\text{\tiny 0.005}}$} & --- & --- \\
t & 100 & 500 & \textbf{0.096$_\pm{\text{\tiny 0.0016}}$} & \textbf{0.096$_\pm{\text{\tiny 0.0016}}$} & --- & --- \\
\bottomrule
\end{tabular}
\caption{\textbf{RMSE for quantile predictions.}}
\end{table}

\begin{table}[htbp]
\centering
\begin{tabular}{lcccccc}
\toprule
\multicolumn{3}{c}{} & \multicolumn{4}{c}{Runtime} \\
\cmidrule(lr){4-7}
Noise & $m$ & $n_j$ & \multicolumn{1}{c}{TKC} & \multicolumn{1}{c}{Fisher} & \multicolumn{1}{c}{BayesQR} & \multicolumn{1}{c}{BRMS} \\
\midrule
ALD & 100 & 100 & 4$_\pm{\text{\tiny 0.32}}$ & \textbf{1.9$_\pm{\text{\tiny 0.16}}$} & --- & --- \\
ALD & 100 & 500 & 69.7$_\pm{\text{\tiny 3.3}}$ & \textbf{21.4$_\pm{\text{\tiny 1.8}}$} & --- & --- \\
N & 100 & 100 & 22.8$_\pm{\text{\tiny 1.2}}$ & \textbf{11.9$_\pm{\text{\tiny 0.94}}$} & --- & --- \\
N & 100 & 500 & 61.3$_\pm{\text{\tiny 5.2}}$ & \textbf{15$_\pm{\text{\tiny 1.2}}$} & --- & --- \\
t & 100 & 100 & 26.5$_\pm{\text{\tiny 2.7}}$ & \textbf{9.2$_\pm{\text{\tiny 0.84}}$} & --- & --- \\
t & 100 & 500 & 56$_\pm{\text{\tiny 6.9}}$ & \textbf{15.9$_\pm{\text{\tiny 1}}$} & --- & --- \\
\bottomrule
\end{tabular}
\caption{\textbf{Runtime (s).}}
\end{table}

\begin{table}[ht!]
\centering
\begin{tabular}{lcccccccccc}
\toprule
\multicolumn{3}{c}{} & \multicolumn{4}{c}{MSE of Random Effect Variance (1)} & \multicolumn{4}{c}{MSE of Random Effect Variance (2)} \\
\cmidrule(lr){4-11}
Noise & $m$ & $n_j $& \multicolumn{1}{c}{TKC} & \multicolumn{1}{c}{Fisher} & \multicolumn{1}{c}{BRMS} & \multicolumn{1}{c}{BayesQR} & \multicolumn{1}{c}{TKC} & \multicolumn{1}{c}{Fisher} & \multicolumn{1}{c}{BRMS} & \multicolumn{1}{c}{BayesQR} \\
\midrule
ALD & 100 & 100 & \textbf{0.05} & 0.086 & --- & --- & 1.7 & \textbf{0.5} & --- & --- \\
ALD & 100 & 500 & 0.17 & \textbf{0.07} & --- & --- & 1.7 & \textbf{0.55} & --- & --- \\
N & 100 & 100 & 0.14 & \textbf{0.12} & --- & --- & 1.3 & \textbf{0.33} & --- & --- \\
N & 100 & 500 & 0.16 & \textbf{0.073} & --- & --- & 1.5 & \textbf{0.21} & --- & --- \\
t & 100 & 100 & \textbf{0.089} & 0.092 & --- & --- & 0.9 & \textbf{0.57} & --- & --- \\
t & 100 & 500 & 0.095 & \textbf{0.052} & --- & --- & 0.99 & \textbf{0.39} & --- & --- \\

\bottomrule
\end{tabular}
\caption{\textbf{Mean squared error of estimated hyperparameters.}}
\end{table}

\newpage
\subsubsection{Gaussian Process}

\begin{table}[ht!]
\centering
\begin{tabular}{lllcccccc}
\toprule
\multicolumn{3}{c}{} & \multicolumn{6}{c}{RMSE} \\
\cmidrule(lr){4-9}
Noise & $n$ & d & TKC & FL & VI & VIVA & QGAM & QGAM Int. \\
\midrule
N & 1{,}000 & 2 & 0.31$_\pm{\text{\tiny 0.034}}$ & 0.31$_\pm{\text{\tiny 0.024}}$ & \textbf{0.26$_\pm{\text{\tiny 0.0037}}$} & 0.63$_\pm{\text{\tiny 0.081}}$ & 0.77$_\pm{\text{\tiny 0.1}}$ & 0.33$_\pm{\text{\tiny 0.022}}$ \\
N & 1{,}000 & 5 & 0.61$_\pm{\text{\tiny 0.015}}$ & 0.57$_\pm{\text{\tiny 0.015}}$ & \textbf{0.54$_\pm{\text{\tiny 0.012}}$} & 0.57$_\pm{\text{\tiny 0.016}}$ & 1.1$_\pm{\text{\tiny 0.022}}$ & 1.1$_\pm{\text{\tiny 0.021}}$ \\
N & 10{,}000 & 2 & 0.35$_\pm{\text{\tiny 0.046}}$ & \textbf{0.27$_\pm{\text{\tiny 0.032}}$} & \textbf{0.27$_\pm{\text{\tiny 0.047}}$} & 0.43$_\pm{\text{\tiny 0.0083}}$ & 1$_\pm{\text{\tiny 0.11}}$ & 0.74$_\pm{\text{\tiny 0.15}}$ \\
N & 10{,}000 & 5 & \textbf{0.9$_\pm{\text{\tiny 0.12}}$} & \textbf{0.76$_\pm{\text{\tiny 0.12}}$} & \textbf{0.89$_\pm{\text{\tiny 0.17}}$} & \textbf{0.74$_\pm{\text{\tiny 0.11}}$} & 1.2$_\pm{\text{\tiny 0.078}}$ & 1.2$_\pm{\text{\tiny 0.085}}$ \\
HetN & 1{,}000 & 2 & \textbf{0.44$_\pm{\text{\tiny 0.045}}$} & \textbf{0.45$_\pm{\text{\tiny 0.06}}$} & \textbf{0.37$_\pm{\text{\tiny 0.05}}$} & 0.86$_\pm{\text{\tiny 0.13}}$ & 1.1$_\pm{\text{\tiny 0.09}}$ & 0.6$_\pm{\text{\tiny 0.059}}$ \\
HetN & 1{,}000 & 5 & \textbf{0.96$_\pm{\text{\tiny 0.078}}$} & \textbf{0.91$_\pm{\text{\tiny 0.073}}$} & \textbf{0.9$_\pm{\text{\tiny 0.07}}$} & \textbf{0.99$_\pm{\text{\tiny 0.095}}$} & 1.5$_\pm{\text{\tiny 0.068}}$ & 1.5$_\pm{\text{\tiny 0.067}}$ \\
HetN & 10{,}000 & 2 & \textbf{0.55$_\pm{\text{\tiny 0.12}}$} & \textbf{0.49$_\pm{\text{\tiny 0.1}}$} & \textbf{0.49$_\pm{\text{\tiny 0.1}}$} & 0.94$_\pm{\text{\tiny 0.34}}$ & 1.5$_\pm{\text{\tiny 0.24}}$ & 1.1$_\pm{\text{\tiny 0.2}}$ \\
HetN & 10{,}000 & 5 & \textbf{1.2$_\pm{\text{\tiny 0.16}}$} & \textbf{1.1$_\pm{\text{\tiny 0.15}}$} & \textbf{1.2$_\pm{\text{\tiny 0.2}}$} & \textbf{1.1$_\pm{\text{\tiny 0.14}}$} & 1.6$_\pm{\text{\tiny 0.097}}$ & 1.6$_\pm{\text{\tiny 0.1}}$ \\
t & 1{,}000 & 2 & 0.44$_\pm{\text{\tiny 0.03}}$ & 0.44$_\pm{\text{\tiny 0.023}}$ & \textbf{0.38$_\pm{\text{\tiny 0.022}}$} & 0.61$_\pm{\text{\tiny 0.049}}$ & 0.79$_\pm{\text{\tiny 0.074}}$ & \textbf{0.41$_\pm{\text{\tiny 0.019}}$} \\
t & 1{,}000 & 5 & 0.66$_\pm{\text{\tiny 0.031}}$ & \textbf{0.6$_\pm{\text{\tiny 0.015}}$} & \textbf{0.59$_\pm{\text{\tiny 0.019}}$} & \textbf{0.58$_\pm{\text{\tiny 0.015}}$} & 1.2$_\pm{\text{\tiny 0.025}}$ & 1.2$_\pm{\text{\tiny 0.028}}$ \\
t & 10{,}000 & 2 & 0.49$_\pm{\text{\tiny 0.045}}$ & 0.45$_\pm{\text{\tiny 0.0082}}$ & \textbf{0.33$_\pm{\text{\tiny 0.046}}$} & 0.48$_\pm{\text{\tiny 0.0077}}$ & 1$_\pm{\text{\tiny 0.11}}$ & 0.75$_\pm{\text{\tiny 0.15}}$ \\
t & 10{,}000 & 5 & \textbf{0.93$_\pm{\text{\tiny 0.11}}$} & \textbf{0.77$_\pm{\text{\tiny 0.12}}$} & \textbf{0.93$_\pm{\text{\tiny 0.17}}$} & \textbf{0.75$_\pm{\text{\tiny 0.11}}$} & 1.3$_\pm{\text{\tiny 0.079}}$ & 1.2$_\pm{\text{\tiny 0.085}}$ \\
\bottomrule
\end{tabular}
\caption{\textbf{RMSE.}}
\end{table}

\vspace{-0.25cm}
\begin{table}[ht!]
\centering
\begin{tabular}{lllcccccc}
\toprule
\multicolumn{3}{c}{} & \multicolumn{6}{c}{Runtime} \\
\cmidrule(lr){4-9}
Noise & $n$ & d & TKC & FL & VI & VIVA & QGAM & QGAM Int. \\
\midrule
N & 1{,}000 & 2 & 6.1$_\pm{\text{\tiny 0.69}}$ & 4.8$_\pm{\text{\tiny 0.84}}$ & 32.8$_\pm{\text{\tiny 1.1}}$ & 35.9$_\pm{\text{\tiny 1.6}}$ & 0.87$_\pm{\text{\tiny 0.39}}$ & \textbf{0.52$_\pm{\text{\tiny 0.074}}$} \\
N & 1{,}000 & 5 & 9.8$_\pm{\text{\tiny 1.4}}$ & 4.5$_\pm{\text{\tiny 0.46}}$ & 67.2$_\pm{\text{\tiny 4.4}}$ & 117$_\pm{\text{\tiny 4.1}}$ & 8.4$_\pm{\text{\tiny 1}}$ & \textbf{1.5$_\pm{\text{\tiny 0.17}}$} \\
N & 10{,}000 & 2 & 163$_\pm{\text{\tiny 10.6}}$ & 144$_\pm{\text{\tiny 23.6}}$ & 547$_\pm{\text{\tiny 22.3}}$ & 431$_\pm{\text{\tiny 13.6}}$ & 8.5$_\pm{\text{\tiny 2.9}}$ & \textbf{4.1$_\pm{\text{\tiny 0.17}}$} \\
N & 10{,}000 & 5 & 3976$_\pm{\text{\tiny 405}}$ & 2351$_\pm{\text{\tiny 166}}$ & 471$_\pm{\text{\tiny 28.8}}$ & 1174$_\pm{\text{\tiny 40.6}}$ & 24.6$_\pm{\text{\tiny 2.6}}$ & \textbf{11$_\pm{\text{\tiny 0.5}}$} \\
HetN & 1{,}000 & 2 & 7.3$_\pm{\text{\tiny 1.3}}$ & 5.4$_\pm{\text{\tiny 0.74}}$ & 44.9$_\pm{\text{\tiny 2.8}}$ & 38.7$_\pm{\text{\tiny 1.6}}$ & 1.8$_\pm{\text{\tiny 0.34}}$ & \textbf{0.76$_\pm{\text{\tiny 0.041}}$} \\
HetN & 1{,}000 & 5 & 9.4$_\pm{\text{\tiny 1}}$ & 5.6$_\pm{\text{\tiny 0.92}}$ & 44.1$_\pm{\text{\tiny 2.2}}$ & 73.5$_\pm{\text{\tiny 2.5}}$ & 9.1$_\pm{\text{\tiny 1.5}}$ & \textbf{2.9$_\pm{\text{\tiny 0.59}}$} \\
HetN & 10{,}000 & 2 & 197$_\pm{\text{\tiny 19.5}}$ & 125$_\pm{\text{\tiny 14.5}}$ & 543$_\pm{\text{\tiny 35}}$ & 414$_\pm{\text{\tiny 19.6}}$ & \textbf{5.9$_\pm{\text{\tiny 0.84}}$} & \textbf{6.2$_\pm{\text{\tiny 0.38}}$} \\
HetN & 10{,}000 & 5 & 6063$_\pm{\text{\tiny 723}}$ & 2356$_\pm{\text{\tiny 238}}$ & 485$_\pm{\text{\tiny 40.1}}$ & 1168$_\pm{\text{\tiny 35.4}}$ & 29.2$_\pm{\text{\tiny 2.2}}$ & \textbf{12.4$_\pm{\text{\tiny 1.4}}$} \\
t & 1{,}000 & 2 & 8.3$_\pm{\text{\tiny 0.8}}$ & 4.7$_\pm{\text{\tiny 0.56}}$ & 41.9$_\pm{\text{\tiny 2.7}}$ & 38.3$_\pm{\text{\tiny 1.6}}$ & 2.1$_\pm{\text{\tiny 0.3}}$ & \textbf{0.82$_\pm{\text{\tiny 0.17}}$} \\
t & 1{,}000 & 5 & 6.6$_\pm{\text{\tiny 1.2}}$ & \textbf{5.2$_\pm{\text{\tiny 0.51}}$} & 45.2$_\pm{\text{\tiny 1.8}}$ & 74$_\pm{\text{\tiny 2.1}}$ & 12.8$_\pm{\text{\tiny 1.9}}$ & \textbf{4.3$_\pm{\text{\tiny 1.1}}$} \\
t & 10{,}000 & 2 & 220$_\pm{\text{\tiny 25.8}}$ & 155$_\pm{\text{\tiny 18.4}}$ & 537$_\pm{\text{\tiny 24.7}}$ & 435$_\pm{\text{\tiny 15.6}}$ & 10.1$_\pm{\text{\tiny 2.3}}$ & \textbf{7.3$_\pm{\text{\tiny 0.95}}$} \\
t & 10{,}000 & 5 & 4170$_\pm{\text{\tiny 406}}$ & 2808$_\pm{\text{\tiny 295}}$ & 474$_\pm{\text{\tiny 23}}$ & 1224$_\pm{\text{\tiny 39.9}}$ & 49.3$_\pm{\text{\tiny 8.2}}$ & \textbf{11.1$_\pm{\text{\tiny 1.1}}$} \\
\bottomrule
\end{tabular}
\caption{\textbf{Runtime (s).}}
\end{table}
\vspace{-0.25cm}
\begin{table}[ht!]
\centering
\begin{tabular}{lcccccc@{\hspace{6pt}}|@{\hspace{6pt}}cccc}
\toprule
\multicolumn{3}{c}{} & \multicolumn{4}{c}{Lengthscale MSE} & \multicolumn{4}{c}{Signal Variance MSE} \\
\cmidrule(lr){4-11}
Noise & $n$ & d & \multicolumn{1}{c}{TKC} & \multicolumn{1}{c}{FL} & \multicolumn{1}{c}{VI} & \multicolumn{1}{c}{VIVA} & \multicolumn{1}{c}{TKC} & \multicolumn{1}{c}{FL} & \multicolumn{1}{c}{VI} & \multicolumn{1}{c}{VIVA} \\
\midrule
N & 1{,}000 & 2 & 0.0049 & \textbf{0.0025} & 0.0035 & 0.027 & 0.057 & 0.055 & \textbf{0.0078} & 0.068 \\
N & 1{,}000 & 5 & 0.091 & 0.16 & 0.063 & \textbf{0.048} & 0.096 & 0.23 & \textbf{0.013} & 0.062 \\
N & 10{,}000 & 2 & 0.03 & 0.021 & \textbf{0.01} & 0.045 & 0.1 & 0.18 & \textbf{0.046} & 0.18 \\
N & 10{,}000 & 5 & 0.013 & 0.046 & 0.1 & \textbf{0.0086} & 0.18 & 0.043 & \textbf{0.024} & 0.029 \\
HetN & 1{,}000 & 2 & 0.0034 & \textbf{0.0031} & 0.0038 & 0.028 & \textbf{0.06} & 0.26 & 0.078 & 0.082 \\
HetN & 1{,}000 & 5 & 0.055 & 0.32 & 0.052 & \textbf{0.042} & 0.29 & 4.4 & \textbf{0.038} & 0.054 \\
HetN & 10000 & 2 & 0.023 & 0.025 & \textbf{0.012} & 0.043 & 3.5 & 27 & \textbf{0.093} & 0.15 \\
HetN & 10{,}000 & 5 & 0.063 & 0.023 & 0.081 & \textbf{0.0069} & 0.97 & 0.028 & \textbf{0.014} & 0.021 \\
t & 1{,}000 & 2 & 0.0089 & \textbf{0.008} & 0.0081 & 0.029 & 0.16 & 0.041 & \textbf{0.019} & 0.066 \\
t & 1{,}000 & 5 & 0.061 & 0.21 & 0.05 & \textbf{0.042} & 0.055 & 0.83 & \textbf{0.012} & 0.057 \\
t & 10{,}000 & 2 & 0.038 & 0.04 & \textbf{0.011} & 0.047 & 0.21 & 0.24 & \textbf{0.041} & 0.1 \\
t & 10{,}000 & 5 & 0.027 & 0.012 & 0.078 & \textbf{0.0061} & 0.035 & 0.05 & 0.02 & \textbf{0.017} \\

\bottomrule
\end{tabular}
\caption{\textbf{Mean squared error of estimated hyperparameters.}}
\end{table}

\clearpage
\subsection{Real world data experiments}
\subsubsection{Single-level Grouped Random Effects}
\begin{table}[ht!]
\centering
\begin{tabular}{lccccc}
\toprule
\multicolumn{1}{c}{} & \multicolumn{5}{c}{Quantile Loss} \\
\cmidrule(lr){2-6}
Dataset & \multicolumn{1}{c}{TKC} & \multicolumn{1}{c}{Fisher} & \multicolumn{1}{c}{BayesQR} & \multicolumn{1}{c}{LQMM} & \multicolumn{1}{c}{BRMS} \\
\midrule
Orthodont & \textbf{0.14$_\pm{\text{\tiny 0.07}}$} & \textbf{0.13$_\pm{\text{\tiny 0.04}}$} & \textbf{0.09$_\pm{\text{\tiny 0.04}}$} & \textbf{0.09$_\pm{\text{\tiny 0.07}}$} & 0.25$_\pm{\text{\tiny 0.11}}$ \\
Labor & \textbf{0.10$_\pm{\text{\tiny 0.03}}$} & \textbf{0.10$_\pm{\text{\tiny 0.03}}$} & \textbf{0.09$_\pm{\text{\tiny 0.03}}$} & \textbf{0.09$_\pm{\text{\tiny 0.04}}$} & 0.28$_\pm{\text{\tiny 0.04}}$ \\
Cars & \textbf{0.095$_\pm{\text{\tiny 0.0098}}$} & \textbf{0.085$_\pm{\text{\tiny 0.0068}}$} & --- & --- \\
\bottomrule
\end{tabular}
\caption{\textbf{Quantile Loss.} }
\end{table}

\begin{table}[ht!]
\centering
\begin{tabular}{lccccc}
\toprule
\multicolumn{1}{c}{} & \multicolumn{5}{c}{Runtime} \\
\cmidrule(lr){2-6}
Dataset & \multicolumn{1}{c}{TKC} & \multicolumn{1}{c}{Fisher} & \multicolumn{1}{c}{BayesQR} & \multicolumn{1}{c}{LQMM} & \multicolumn{1}{c}{BRMS} \\
\midrule
Orthodont & 1.07$_\pm{\text{\tiny 0.49}}$ & \textbf{0.31$_\pm{\text{\tiny 0.17}}$} & 138.18$_\pm{\text{\tiny 69.35}}$ & 1.07$_\pm{\text{\tiny 0.04}}$ & 383.31$_\pm{\text{\tiny 20.50}}$ \\
Labor & 1.53$_\pm{\text{\tiny 0.56}}$ & \textbf{0.42$_\pm{\text{\tiny 0.24}}$} & 2015.07$_\pm{\text{\tiny 1100.49}}$ & 3.18$_\pm{\text{\tiny 0.04}}$ & 948.60$_\pm{\text{\tiny 96.67}}$ \\
Cars & 87.7$_\pm{\text{\tiny 57.4}}$ & \textbf{61.1$_\pm{\text{\tiny 13.1}}$} & --- & --- \\
\bottomrule
\end{tabular}
\caption{\textbf{Runtime (s).}}
\end{table}

\subsubsection{Crossed Random Effects}

\begin{table}[ht!]
\centering
\begin{tabular}{lcccc}
\toprule
\multicolumn{1}{c}{} & \multicolumn{4}{c}{Quantile Loss} \\
\cmidrule(lr){2-5}
Dataset & \multicolumn{1}{c}{TKC} & \multicolumn{1}{c}{Fisher} & \multicolumn{1}{c}{BayesQR} & \multicolumn{1}{c}{BRMS} \\
\midrule
Cars & \textbf{0.093$_\pm{\text{\tiny 0.0088}}$} & \textbf{0.096$_\pm{\text{\tiny 0.013}}$} & --- & --- \\
Ratings & \textbf{0.081$_\pm{\text{\tiny 0.0015}}$} & \textbf{0.077$_\pm{\text{\tiny 0.0018}}$} & --- & --- \\
\bottomrule
\end{tabular}
\caption{\textbf{Quantile Loss.}}
\end{table}

\begin{table}[ht!]
\centering
\begin{tabular}{lcccc}
\toprule
\multicolumn{1}{c}{} & \multicolumn{4}{c}{Runtime} \\
\cmidrule(lr){2-5}
Dataset & \multicolumn{1}{c}{TKC} & \multicolumn{1}{c}{Fisher} & \multicolumn{1}{c}{BayesQR} & \multicolumn{1}{c}{BRMS} \\
\midrule
Cars & \textbf{7043$_\pm{\text{\tiny 1477}}$} & \textbf{4218$_\pm{\text{\tiny 1560}}$} & --- & --- \\
Ratings & 770$_\pm{\text{\tiny 195}}$ & \textbf{337$_\pm{\text{\tiny 64.6}}$} & --- & --- \\
\bottomrule
\end{tabular}
\caption{\textbf{Runtime (s).}}
\end{table}

\clearpage
\subsubsection{Gaussian Process}
\begin{table}[ht!]
\centering
\begin{tabular}{lcccccc}
\toprule
\multicolumn{1}{c}{} & \multicolumn{6}{c}{Quantile Loss} \\
\cmidrule(lr){2-7}
Dataset & TKC & FL & VI & VIVA & QGAM & QGAM Int. \\
\midrule
Heaton & 0.052$_\pm$\text{\tiny 0.0093} & 0.038$_\pm$\text{\tiny 0.0018} & \textbf{0.033$_\pm$\text{\tiny 9.9e-04}} & 0.039$_\pm$\text{\tiny 0.0016} & 0.046$_\pm$\text{\tiny 0.0013} & 0.038$_\pm$\text{\tiny 7.3e-04} \\
House & 0.064$_\pm$\text{\tiny 0.017} & 0.053$_\pm$\text{\tiny 0.0055} & \textbf{0.046$_\pm$\text{\tiny 0.0019}} & 0.083$_\pm$\text{\tiny 0.0062} & 0.065$_\pm$\text{\tiny 0.002} & 0.059$_\pm$\text{\tiny 8.8e-04} \\
Laegern & \textbf{0.066$_\pm$\text{\tiny 0.009}} & \textbf{0.065$_\pm$\text{\tiny 0.0037}} & \textbf{0.063$_\pm$\text{\tiny 0.0018}} & 0.086$_\pm$\text{\tiny 0.0074} & 0.075$_\pm$\text{\tiny 0.0026} & 0.071$_\pm$\text{\tiny 0.0026} \\
MODIS & 0.05$_\pm$\text{\tiny 0.015} & 0.041$_\pm$\text{\tiny 0.0033} & \textbf{0.036$_\pm$\text{\tiny 0.0021}} & 0.042$_\pm$\text{\tiny 0.0021} & 0.065$_\pm$\text{\tiny 0.0021} & 0.046$_\pm$\text{\tiny 0.0028} \\
\bottomrule
\end{tabular}
\caption{\textbf{Quantile loss}.}
\end{table}

\begin{table}[ht!]
\centering
\begin{tabular}{lcccccc}
\toprule
\multicolumn{1}{c}{} & \multicolumn{6}{c}{Runtime} \\
\cmidrule(lr){2-7}
Dataset & TKC & FL & VI & VIVA & QGAM & QGAM Int. \\
\midrule
Heaton & 1196$_\pm$\text{\tiny 238} & 869$_\pm$\text{\tiny 194} & 2283$_\pm$\text{\tiny 10.3} & 2915$_\pm$\text{\tiny 232} & \textbf{4.8$_\pm$\text{\tiny 0.94}} & \textbf{3$_\pm$\text{\tiny 0.91}} \\
House & 979$_\pm$\text{\tiny 267} & 415$_\pm$\text{\tiny 99.6} & 2282$_\pm$\text{\tiny 13.3} & 3031$_\pm$\text{\tiny 103} & \textbf{4$_\pm$\text{\tiny 1.1}} & 9.4$_\pm$\text{\tiny 3.2} \\
Laegern & 1025$_\pm$\text{\tiny 469} & 999$_\pm$\text{\tiny 157} & 2461$_\pm$\text{\tiny 27.1} & 2966$_\pm$\text{\tiny 412} & 12.9$_\pm$\text{\tiny 5.7} & \textbf{6.7$_\pm$\text{\tiny 1}} \\
MODIS & 1569$_\pm$\text{\tiny 367} & 878$_\pm$\text{\tiny 141} & 2285$_\pm$\text{\tiny 23.1} & 2872$_\pm$\text{\tiny 174} & \textbf{2.9$_\pm$\text{\tiny 0.25}} & 3.5$_\pm$\text{\tiny 0.37} \\
\bottomrule
\end{tabular}
\caption{\textbf{Runtime (s).}}
\end{table}

\clearpage
\section{Uncertainty-Aware Prediction Intervals}
\label{appendix:cqr}

We demonstrate how our predictive variance estimates can improve conditional coverage in the Conformalized Quantile Regression (CQR) framework \citep{romano2019conformalized, rossellini2024integrating}.

Given quantile estimates $\hat{q}_{\alpha/2}(x)$ and $\hat{q}_{1-\alpha/2}(x)$ for the lower and upper quantiles, standard CQR constructs prediction intervals by finding a scalar correction $t$ on a held-out calibration set:
\begin{equation}
C_{1-\alpha}(x) = [\hat{q}_{\alpha/2}(x) - t, \, \hat{q}_{1-\alpha/2}(x) + t]
\end{equation}
where $t$ is the smallest value such that 
\begin{equation}
\frac{1}{n_{\text{cal}}} \sum_{i=1}^{n_{\text{cal}}} \mathbbm{1}(y_i \in C_{1-\alpha}(x_i)) \geq 1 - \alpha
\end{equation}
on the calibration set. This procedure guarantees marginal coverage: $\mathbb{P}(y_{n+1} \in C_{1-\alpha}(x_{n+1})) \geq 1-\alpha$.

To improve conditional coverage, \cite{rossellini2024integrating} propose modulating the correction by a measure of local epistemic uncertainty $\hat{\sigma}(x)$:
\begin{equation}
C_{1-\alpha}(x) = [\hat{q}_{\alpha/2}(x) - t \cdot \hat{\sigma}(x), \, \hat{q}_{1-\alpha/2}(x) + t \cdot \hat{\sigma}(x)]
\end{equation}
The scalar $t$ is again calibrated to achieve marginal coverage, but the interval width now adapts to local uncertainty. In regions where the model is more uncertain, i.e. large $\hat{\sigma}(x)$, intervals are wider; in confident regions where $\hat{\sigma}(x)$ is smaller, they are narrower. This can improve conditional coverage:
\begin{equation}
\mathbb{P}(y_{n+1} \in C_{1-\alpha}(x_{n+1}) \mid x_{n+1} = x) \geq 1-\alpha \quad \text{for all } x.
\end{equation}
We use the \emph{predictive variance} from our Laplace approximation as the epistemic uncertainty measure $\hat{\sigma}(x)$. This variance reflects both parameter uncertainty (via the posterior covariance) and structural uncertainty captured by the Gaussian process. 

In the experiment, we generate heteroscedastic data with Gaussian noise and a non-uniform covariate distribution. 
Covariates are drawn from a compact interval with density increasing in $x$:
\begin{align}
x_i &\sim p(x), \qquad p(x) \propto x^2 \, \mathbf{1}_{[0,2]}(x).
\end{align}
This distribution is obtained by importance resampling from an initial uniform proposal, resulting in higher data density near the boundary of the covariate domain. Conditional on $x_i = x$ the response satisfies
\begin{align}
(y_i \mid x_i = x) &\sim \mathcal{N}\bigl(5 + \sin(5x), \, 1 + 0.6 |x|\bigr)
\end{align}
The combination of non-uniform covariate support and input-dependent noise creates regions with sparse data and increased uncertainty, posing a challenge for conditional coverage guarantees.

The evaluation procedure proceeds as follows:
\begin{enumerate}
    \item The data are split into a training set ($n = 500$), a calibration set ($n_{\mathrm{cal}} = 500$), and a large test set ($n_{\mathrm{test}} = 10'000$).
    \item Gaussian process quantile regression models are fitted on the training set for the lower and upper quantiles $\tau = \alpha/2$ and $\tau = 1 - \alpha/2$.
    \item Using the calibration set, conformal adjustment is performed:
    \begin{itemize}
        \item \textbf{Standard CQR}: a single scalar threshold $t$ is selected to achieve the nominal marginal coverage.
        \item \textbf{Uncertainty-aware CQR}: a scalar threshold $t$ is selected after modulating conformity scores by the predictive variance, again achieving nominal marginal coverage.
    \end{itemize}
    \item Conditional coverage is evaluated on the test set, which is sufficiently large to allow stable estimation across covariate regions.
\end{enumerate}

We assess conditional coverage by partitioning the test set into $K=10$ equally-spaced bins based on covariate values and computing the empirical coverage within each bin. The large test set size ensures reliable estimation of bin-specific coverage rates. Figure~\ref{fig:cqr_coverage} shows conditional coverage across bins for both methods. Standard CQR exhibits substantial variation in coverage across bins. In contrast, uncertainty-aware CQR using our predictive variance achieves more uniform coverage across bins, with coverage rates closer to the nominal $1-\alpha$ level in all regions. These results demonstrate that while our predictive distributions may not achieve exact frequentist calibration in all settings, they capture \textit{meaningful uncertainty information that improves conditional coverage in downstream tasks}. The predictive variance successfully identifies regions where additional interval width is needed, leading to more reliable inference across the covariate space.

\begin{figure}[ht!]
    \centering
    \includegraphics[width=1\linewidth]{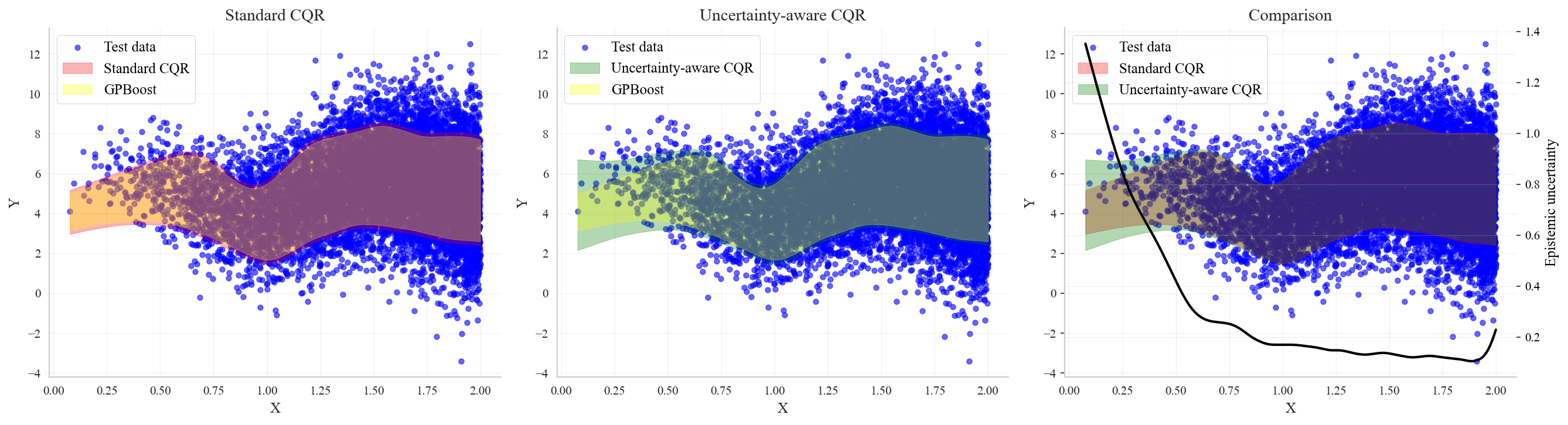}
    \caption{\textbf{Prediction intervals and predictive uncertainty on heteroscedastic test data}. \textbf{Left:} Standard CQR intervals (red) with fitted quantiles (yellow). The scalar correction is nearly invisible as the base model achieves reasonable marginal coverage. \textbf{Middle:} Uncertainty-aware CQR intervals adapt to local data density: intervals widen substantially in the low-density region (left) where epistemic uncertainty is high. \textbf{Right:} Predictive standard deviation from our Laplace approximation increases dramatically in sparse regions, correctly identifying areas requiring wider intervals for reliable coverage.}
\label{fig:cqr_intervals}
    \label{fig:cqr}
\end{figure}

\begin{figure}[ht!]
    \centering    \includegraphics[width=0.95\linewidth]{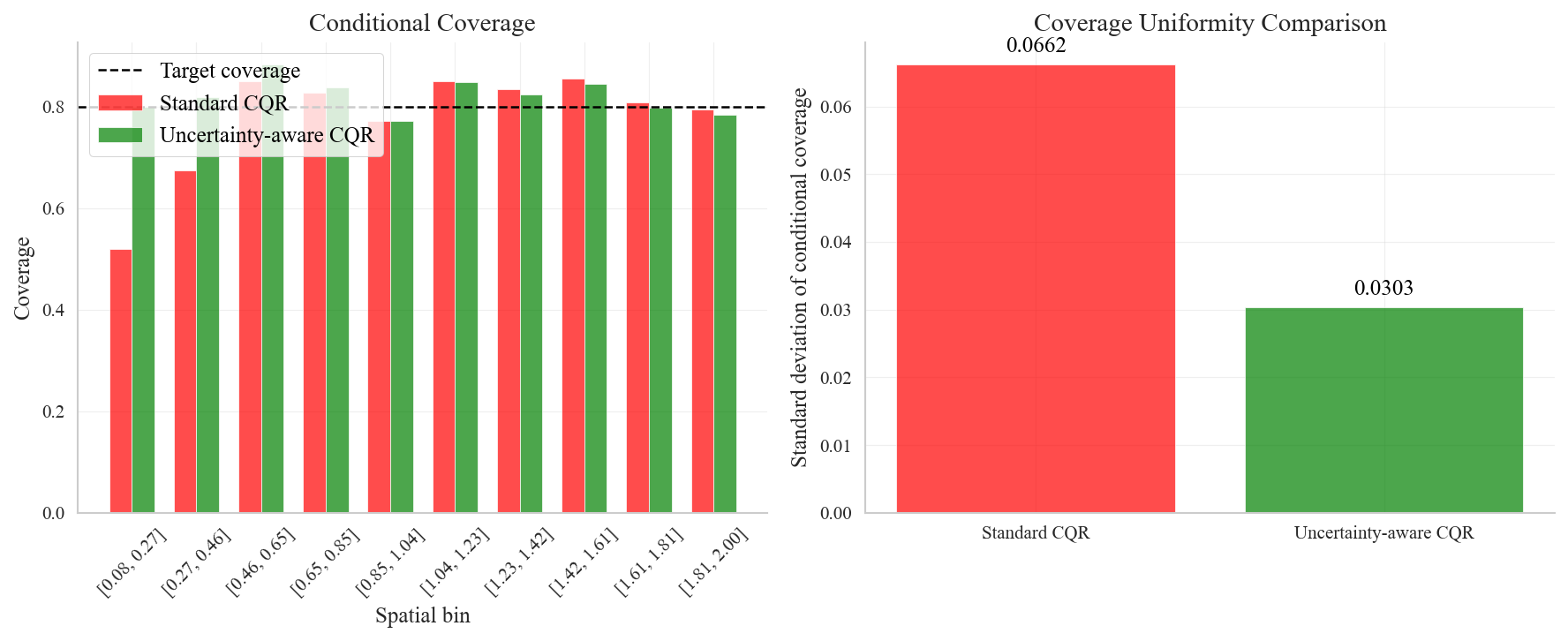}
    \caption{\textbf{Conditional coverage assessment across covariate bins}. \textbf{Left:} Empirical coverage in 10 equally-spaced bins compared to nominal target level (dashed line). Uncertainty-aware CQR (green) achieves more uniform coverage across bins than standard CQR (red). \textbf{Right:} Standard deviation of coverage across bins. Uncertainty-aware CQR exhibits lower variability, indicating more consistent conditional coverage and better adaptation to local uncertainty.}
    \label{fig:cqr_coverage}
\end{figure}

\clearpage
\renewcommand{\refname}{Appendix References}
\putbib[bib_name]
\end{bibunit}

\end{document}